\documentclass[11pt]{article}
\usepackage[dvips]{graphicx}
\usepackage{amssymb,amsmath,color}
\usepackage{url}

\usepackage[mathcal]{eucal}

\DeclareMathAlphabet\mathbfcal{OMS}{cmsy}{b}{n}

\usepackage{hyperref}

\newcommand{\BEAS}{\begin{eqnarray*}}
\newcommand{\EEAS}{\end{eqnarray*}}
\newcommand{\BEA}{\begin{eqnarray}}
\newcommand{\EEA}{\end{eqnarray}}
\newcommand{\BEQ}{\begin{equation}}
\newcommand{\EEQ}{\end{equation}}
\newcommand{\BIT}{\begin{itemize}}
\newcommand{\EIT}{\end{itemize}}
\newcommand{\BNUM}{\begin{enumerate}}
\newcommand{\ENUM}{\end{enumerate}}
\newcommand{\BA}{\begin{array}}
\newcommand{\EA}{\end{array}}
\newcommand{\diag}{\mathop{\rm diag}}
\newcommand{\Diag}{\mathop{\rm Diag}}

\newcommand{\tr}{\mathop{ \rm tr}}

\newcommand{\idm}{I}
\newcommand{\rb}{\mathbb{R}}

\newcommand{\ds}{\displaystyle }
\newcommand{\BlackBox}{\rule{1.5ex}{1.5ex}}  

\newenvironment{proof}{\par\noindent{\bf Proof\ }}{\hfill\BlackBox\\[2mm]}
\newtheorem{lemma}{Lemma}

\newtheorem{proposition}{Proposition}

\newcommand{\mysec}[1]{Section~\ref{sec:#1}}
\newcommand{\eq}[1]{Eq.~(\ref{eq:#1})}
\newcommand{\myfig}[1]{Figure~\ref{fig:#1}}

 \def \ds { \displaystyle}

\def \E{{\mathbb E}}
\def \P{{\mathbb P}}
\def \Z{{\mathbb Z}}

\def \H{{\mathcal H}}

\def \X{{\mathcal X}}
\def \Y{{\mathcal Y}}
\def \A{{\mathcal A}}

 \def \P{{\mathbb P}}

\title{Information Theory with Kernel Methods}

\author{Francis Bach \\
Inria,  Ecole Normale Sup\'erieure \\
PSL Research University \\
\url{francis.bach@inria.fr}}

\date{\today}

\begin{document}
\maketitle

\begin{abstract}

We consider the analysis of probability distributions through their associated covariance operators  from reproducing kernel Hilbert spaces. We show that  the von Neumann entropy and relative entropy of these operators are intimately related to the usual notions of Shannon entropy and relative entropy, and share many of their properties.
They come together with efficient estimation algorithms from various oracles on the probability distributions. We also consider  product spaces and show that for tensor product kernels, we can define notions of mutual information and joint entropies, which can then characterize independence perfectly, but only partially conditional independence. We finally 
show how these new notions of relative entropy lead to new upper-bounds on  log partition functions, that can be used together with convex optimization within variational inference methods, providing a new family of probabilistic inference methods. 

\end{abstract}

\section{Introduction}

 Characterizing and studying probability distributions through moments has a long history. For distributions supported in a vector space, only an infinite number of moments can characterize the distribution. Beyond traditional polynomial moments, kernel methods based on reproducing kernel Hilbert spaces (RKHS)~\cite{scholkopf-smola-book,berlinet2011reproducing} have emerged as a natural tool for studying the interactions between moments and other properties of the underlying distributions (such as independence or conditional independence), as well as providing algorithmic tools to estimate these moments from data.

A natural way of using kernel methods is to consider a feature map $\varphi: \X \to \H$ from the underlying space~$\X$ and a specific Hilbert space $\H$, and to consider the \emph{mean element}~\cite{sriperumbudur2010hilbert}, also referred to as the \emph{mean embedding}, for a probability distribution~$p$, defined as 
\BEQ
\label{eq:mean}
\ds \mu_p = \int_\X \varphi(x) dp(x),
\EEQ
which is an element of the Hilbert space $\H$. The key idea behind kernel methods is to study properties of~$\H$ and its interaction with $\X$ only through the kernel function $k:\X \times \X \to \rb$ defined as
$$
k(x,y) = \langle \varphi(x),\varphi(y) \rangle.
$$
Under classical universality conditions on the kernel function $k$~\cite{sriperumbudur2010hilbert,micchelli2006universal}, the mapping $p \mapsto \mu_p$ is injective (and thus characterizes the distribution $p$), and can be estimated from independent and identically distributed (i.i.d.) samples efficiently, with convergence rates proportional to $1/\sqrt{n}$ where $n$ is the number of observations.

Representing a probability distribution $p$ through $\mu_p \in \H$ defines explicitly a metric between distributions as $(p,q) \mapsto \| \mu_p - \mu_q\|$ (for the $\H$-norm), which has been extensively used in machine learning and data science~\cite{muandet2017kernel},
{in particular for measuring the dependence between two random variables, leading to algorithms for independent component analysis~\cite{cardoso2003dependence,bach2002kernel} and for feature selection~\cite{song2012feature}, or as a model fitting criterion when estimating parameters of probability distributions~\cite{binkowski2018demystifying}.}

 A classical drawback of the underlying geometry is the lack of straightforward connection with classical information theory tools. For example for discrete data, this leads to Euclidean norms between probability mass functions, which is rarely seen as an appropriate geometry for the simplex.

In this paper, we consider the second-order moment, which we call the \emph{covariance operator}\footnote{Note that covariance operators are typically ``centered'', that is, $\mu_p   \mu_p^\ast$ is subtracted.}
 $$  \Sigma_p = \int_\X \varphi(x)\varphi(x)^\ast dp(x),
$$
 which is an  {operator} from $\H$ to $\H$, defined through $ \ds \langle f,  \Sigma_p g \rangle
 = \int \langle f, \varphi(x) \rangle  \langle g, \varphi(x) \rangle  dp(x)
$, 
 self-adjoint and positive semi-definite.\footnote{In this paper, we use the notation $uu^\ast$ for an element $u \in \H$ to denote the operator $uu^\ast: \H \to \H$ such that $uu^\ast (f) = \langle f, u \rangle u$. The usual notation $u \otimes u$ will be used for tensor products in \mysec{multivariate}.}
 
As we show in this paper, it shares similar nice properties with the mean element in terms of universality, estimation from finite samples, and general applicability to all sets $\X$ where positive definite kernels can be defined (e.g., structured discrete objects). Furthermore, owing to the many tools from quantum information theory, and in particular the von Neumann entropy $\tr \big[ \Sigma_p \log \Sigma_p \big]$ and its associated divergence (relative entropy), we define these notions and demonstrate their many nice properties and applications, as well as explicit quantitative links with  regular Shannon entropies (in fact, we can see the regular notion of entropy as the limit of our kernel entropies when the bandwidth of the kernel goes to zero, and for discrete data, they are equivalent).

The paper is organized as follows:
 \BIT
\item We review in \mysec{operators} properties of covariance operators and reproducing kernel Hilbert spaces.

\item In \mysec{quantum}, we review some of the key mathematical results on quantum relative entropies, that we will leverage in further sections.

\item In \mysec{entropy}, we define our notions of kernel entropy and Kullback-Leibler divergence (relative entropy), and give their main properties, which are shared with the classical notions, such as convexity, or additivity for independent variables. In particular our notions of relative entropies will always be lower bounds on the Shannon relative entropy (in this paper, we use interchangeably the terms ``relative entropy'' and ``Kullback-Leibler (KL) divergence'').
\item In \mysec{estimation}, we show how we can leverage existing work in kernel methods to estimate these quantities from finite samples, with a convergence rate proportional to $1/\sqrt{n}$ from $n$ observations, and without the need for regularization. We also propose statistically more efficient estimators if integrals of kernel functions  under $p$ can be computed.

\item In \mysec{multivariate}, we consider product spaces $\X = \X_1 \times \cdots \times \X_d$, and show that for tensor product kernels, we can define notions of mutual information and joint entropies. We can then characterize independence perfectly, but only partially conditional independence, and extend the submodularity property of the regular entropy.

\item In \mysec{duality}, we show how these new notions of relative entropy lead to new upper-bounds on  log partition functions, that can be used together with convex optimization within variational inference methods~\cite{wainwright2008graphical}, providing a new family of probabilistic inference methods. Illustrative applications to $[0,1]$ and $\{-1,1\}^d$ are explicitly developed.

\EIT

\section{Covariance operators}
\label{sec:operators}

In this section, we review the notions of covariance operators, starting with reproducing kernel Hilbert spaces. For more details, see~\cite{berlinet2011reproducing,scholkopf-smola-book,shawe2004kernel}.

\subsection{Kernels and RKHSs}

In this paper we consider  a compact set $\X$  and probability distributions on $\X$  {(which we assume to be Borel probability measures}), as well as a kernel function $k: \X \times \X \to \rb$, such that
\BIT
\item[\textbf{(A1)}] $k$ is a continuous positive definite kernel on the compact set $\X$, with $k(x,x)=1$ for all $x \in \X$.
\EIT
Positive definite kernel functions are functions for which all matrices of pairwise kernel evaluations are positive semi-definite. They are known to  define a Hilbert space $\H$ (with dot product $\langle \cdot, \cdot \rangle$) of functions $f: \X \to \rb$, called a reproducing kernel Hilbert space (RKHS), as well as a map $\varphi: \X \to \H$ such that, for all $x,y \in \X$, 
$$ \mbox{ (a) }   \  f(x) = \langle f, \varphi(x) \rangle, \ \ \ \mbox{ (b)  }  \ k(\cdot,x) = \varphi(x), \ \ \mbox{ and (c) } \  k(x,y) = \langle \varphi(x), \varphi(y) \rangle.$$ 
A useful representation of functions in $\H$ is through linear combinations of kernel functions, that is, {for any $\alpha \in \rb^n$ and arbitrary elements $x_1,\dots,x_n$ of $\X$, the function $f = \sum_{i=1}^n \alpha_i k(\cdot,x_i)$ is in $\H$ and its norm is equal to
$\sum_{i,j=1}^n \alpha_i \alpha_j k(x_i,x_j)$. It turns out that all functions in $\H$ can be generated as appropriate limits of such finite linear combinations.}

Note that from $k(x,x)=1$, we get $\| \varphi(x) \|=1$ for all $x \in \X$. {Note that most of the developments in this paper do not require further structure on the set $\X$, except when characterizing regularity of density functions is needed in \mysec{asymptotics}, where will assume the existence of a distance.}

\paragraph{Universal and characteristic kernels.} We note that our covariance operator $\Sigma_p$ happens to be a mean embedding for the feature map $x\mapsto \varphi(x)   \varphi(x)^*$, from $\X$ to the space of operators from $\H$ to $\H$. Equipped with the Hilbert-Schmidt norm, we obtain a Hilbert space which happens to be isomorphic\footnote{We can build the simple mapping $M\mapsto \langle M, \varphi(x)   \varphi(x)^\ast \rangle $ obtained from
 $f   g^\ast \mapsto \langle f   g^\ast, \varphi(x)   \varphi(x)^\ast \rangle = f(x) g(x)$.}  to the RKHS obtained from the kernel $(x,y) \mapsto k(x,y)^2$.  
 
 Therefore, injectivity of the map $p \mapsto \Sigma_p$ can be obtained from sufficient conditions for injectivity of the mean element map, a property for the kernel referred to as \emph{characteristic}~\cite{sriperumbudur2008injective}. As shown in~\cite[Prop. 5]{fukumizu2009kernel}, this is equivalent for the associated RKHS to be dense in $L_2(q)$ for all probability measures $q$ on $\X$, a sufficient condition being here (because of our compactness assumption) that the RKHS is dense in the set of continuous functions equipped with the uniform norm, a property referred to as \emph{universality}~\cite{steinwart2001influence}. Note that if \textbf{(A1)} is satisfied, then the universality of $k$ implies the universality of $k^2$.

We will see in \mysec{examples} below that many kernels are indeed universal, and thus the covariance operator will be characteristic of the corresponding probability distribution.

\subsection{Covariance operators}

Given a probability distribution $p$ on $\X$, we consider the operator $\Sigma_p: \H 
\to \H $, defined as
\BEQ
\label{eq:cov}
\Sigma_p = \int_{\X} \varphi(x) \varphi(x)^\ast  dp(x).
\EEQ
For any $f,g \in \H$, we have, by definition:
$\ds \langle g, \Sigma_p f \rangle = \E_{X\sim p}\big[  \langle f,\varphi(X)\rangle \langle g,\varphi(X)\rangle \big] 
=  \E_{X\sim p}\big[   f(X) g(X) \big] $.

\begin{proposition}[Properties of covariance operators]
Assume \textbf{(A1)}.
\BIT
\item[(a)] For any probability distribution $p$ on $\X$, the operator $\Sigma_p: \H \to \H$ defined in \eq{cov}, is self-adjoint, positive-semidefinite and has unit trace. 
\item[(b)] If the kernel $k^2$ is universal, then 
the mapping $p \mapsto \Sigma_p$ is injective from probability distributions to self-adjoint positive-semidefinite, unit trace operators.
\item[(c)] {If $p$ has full support in $\X$, then the operator $\Sigma_p$ has a trivial null space.}
\item[(d)] Given the mean element defined in \eq{mean}, we have $\Sigma_p \succcurlyeq \mu_p \otimes \mu_p$.
\EIT
\end{proposition}
\begin{proof}
(a) The operator is well defined and has unit trace because $\tr [ \varphi(x)   \varphi(x)^\ast ] = \| \varphi(x)\|^2 =1$ for all $x \in \X$. 

(b) If $\Sigma_p = \Sigma_q$, then for all $y \in \X$, $\int_\X k(x,y)^2[ dp(x)-dq(x)]=0$.  The representation of functions in the RKHS associated with $k^2$ as linear functions of $x \mapsto k(x,y)^2$, and the universality of $k^2$ implies that 
 $\int_\X f(x) [ dp(x)-dq(x)]=0$ for all continuous functions $f$, hence, $p=q$ and the
injectivity of the map.

Property (c) is a simple consequence of $\ds \langle g, \Sigma_p g \rangle = \E_{X\sim p}\big[ g(X)^2 \big]$, {which is valid for all $g \in \H$. If this expectation is equal to zero, and since elements of $\H$ are continuous functions, we must have $g=0$ as soon as $p$ has full support.}

Property (d) is a consequence of the non-negativity of the variance of $f(X)$, for $f$ in $\H$.
\end{proof}

Since the operators are trace-class, they have a discrete spectrum~\cite{reed1980functional,brezis80analyse}, composed of a summable sequence of strictly positive eigenvalues tending to zero, and potentially zero (if $p$ does not have full support in $\X$). The spectral decay is a key quantity for the study of kernel methods~(see~\cite{de2021regularization} and references therein), as used in \mysec{estimation} to estimate entropies from data.

{Throughout this paper, we will consider \emph{spectral functions} of trace-class self-adjoint operators $A$, such as $\log A$, $\exp(A)$, or more generally $f(A)$, for a function $f: \rb \to \rb$. These are defined through the eigenvalue decomposition by keeping eigenvectors and replacing eigenvalues $\lambda$ by $f(\lambda)$ (for polynomials, we then recover the usual matrix polynomials).}

\paragraph{Uniform distribution and symmetric sets. }
\label{sec:symm}

{In order to define the notion of kernel entropy from kernel relative entropy, we will need to use a specific ``base'' distribution $\tau$, that will play the role of the uniform distribution when this notion is classically defined (e.g., $\X$ is a finite set or a subset of $\rb^d$, like in all examples in \mysec{examples}).}

For thi base distribution $\tau$ on $\X$, we use the notation (instead of $\Sigma_\tau$): $$\ds \Sigma = \int_{\X} \varphi(x) \varphi(x)^\ast d\tau(x).$$

In this paper, some concepts will be easier for \emph{symmetric sets}, that is, sets where there exists a set $\mathcal{T}$ of bijective transformations $t:\X \to \X$, such that for all $x,y \in \X$, $\exists t \in \mathcal{T}$ such that $t(x)=y$. {We will then assume that the kernel $k$ and the base distribution $\tau$ are compatible with this notion of symmetry, that is, $\tau(A) = \tau(t(A))$ for any Borel set $A$ and $t \in \mathcal{T}$, and $k(x,y) = k(t(x),t(y))$ for all transformations $t \in \mathcal{T}$.}

Then the covariance operator for the base distribution will also exhibit symmetries. More precisely, all eigensubspaces will be invariant by symmetry, that is, if $e_1,\dots,e_m \in \H$ is an orthonormal basis of an eigensubspace {corresponding to an eigenvalue $\lambda$}, then\footnote{This is a simple consequence of $(e_1 \circ t, \dots, e_m \circ t)$ being another orthonormal basis of the same eigensubspace and thus  equal to an orthogonal matrix times $(e_1, \dots, e_m)$, so that $\sum_{i=1}^m e_i e_i^\ast = \sum_{i=1}^m (e_i \circ t) (e_i \circ t)^\ast $, leading to $\langle \varphi(x),\sum_{i=1}^m e_i e_i^\ast \varphi(x)\rangle$ constant on $\X$, and thus implying the desired result.}  $\sum_{i=1}^m e_i(t(x))^2 = \sum_{i=1}^m e_i(x)^2$ for all transformations $t \in \mathcal{T}$ and $x \in \X$. {This implies that the function $x \mapsto \sum_{i=1}^m e_i(t(x))^2$ is constant, equal to its expectation under $\tau$ (which is itself equal to $m \lambda$).

As a consequence, for a spectral function $\Sigma \mapsto f(\Sigma)$ based on a real-valued function $f:\rb \to \rb$, the dot-product 
$\langle \varphi(x), f(\Sigma) \varphi(x) \rangle$ can be expanded as the sum of the contributions coming from all eigensubspaces, with the one associated to the eigenvalue $\lambda$ leading to  $f(\lambda) \sum_{i=1}^m e_i(x)^2$ being independent of $x$ (and equal to $ m \lambda f(\lambda) $). This leads notably to 
$\langle \varphi(x), f(\Sigma) \varphi(x) \rangle= \tr [ \Sigma f(\Sigma)]$ for all real-valued functions $f$ defined on~$\rb$, and to $\tr [ \Sigma f(\Sigma)] = \tr [ \Sigma_p f(\Sigma)]$ for any probability distributions $p$.}

\paragraph{Surjectivity.} By definition, every covariance operator is in the closure  $\mathcal{M}$ of the \emph{convex hull} of all $\varphi(x)   \varphi(x)^\ast$, for $x \in \X$. We will also need the closure  $\mathcal{A}$ of the \emph{span} of these vectors. By definition, the mapping $p \mapsto \Sigma_p$ is surjective from probability distributions to $\mathcal{M}$. Moreover, if the kernel $k$ is universal, then\footnote{For the non trivial inclusion, we consider $A = \int_\X   \varphi(x) \varphi(x)^\ast dp(x) $ with $\int_\X dp(x) = 1$, such that $A \succcurlyeq 0$. We have for all $f \in \H$, $\int_\X f(x)^2 dp(x) \geqslant 0$. If $k$ is universal, then all positive continuous functions $g$ are limits of functions $f^2$, so that $\int_\X g(x) dp(x) \geqslant 0$, which implies that $p$ is a probability measure. }
\BEQ
\mathcal{M} = \big\{ A \in \mathcal{A}, \ A \succcurlyeq 0, \ \tr A = 1 \big\}.
\EEQ
For orthonormal embeddings, where $k(x,y)=0$ if $x \neq y$, $\mathcal{M}$ is (in the appropriate basis) the set of non-negative diagonal operators, while $\mathcal{A}$ is the set of diagonal operators (in the same basis).

Note that in most cases $\mathcal{A}$ is strictly included in the set of trace class operators from $\H$ to $\H$, which is equivalent to the existence of bounded self-adjoint operators $M: \H \to \H$ such that $\forall x \in \X, \langle \varphi(x), M \varphi(x) \rangle = 0$. Since $\| \varphi(x)\|^2 = 1$ for all $x \in \X$, this is equivalent to the existence of a positive bounded operator $M$ so that 
$\forall x \in \X, \ \langle \varphi(x), M \varphi(x) \rangle = 1$, which, with an eigenvalue decomposition leads to a finite or countable family of functions $(f_i)_{i \in I}$ in  $\H$ such that $\sum_{i \in I } f_i(x)^2 = 1$, that is, a ``partition of unity'', which are well known to exist in many spaces (see, e.g.,~\cite[Chapter 3]{berger2012differential} or~\cite[Section 1.4]{hormander1984analysis}).

\subsection{Examples}
\label{sec:examples}

\paragraph{Finite set.} If $\X = \{ x_1,\dots,x_m\}$ is a finite set of cardinality $m$, then the Gram matrix $K \in \rb^{m \times m}$ defined as $K_{ij} = \langle \varphi(x_i),\varphi(x_j) \rangle$ completely characterizes the kernel notions up to rotation. If $K = \idm$ (orthonormal embedding), then there exist $m$ orthonormal vectors $u_1,\dots,u_m$ such that $\Sigma_p = \sum_{i=1}^m \P(X=i) u_i   u_i^\ast$, and thus, all covariance matrices commute and therefore share the same eigenbasis, and we will recover exactly the regular notions of discrete entropy.

We get a symmetric set if $K$ is invariant by permutations, that is, $K = \alpha \idm + (1-\alpha) 11^\top$, and as soon as $K$ is invertible, a universal kernel.\footnote{Note that $K$ invertible implies that the Hadamard (i.e., pointwise) product $K \circ K$ is invertible, but not vice-versa.}

\paragraph{Polynomial kernels.}
When $\X$ is a subset of $\rb^d$, then we can consider kernels of the form $k(x,y) = ( \alpha + \beta \, x^\top y)^r$, which correspond to $\varphi(x)$ composed of all monomials of order up to $r$, and the covariance operator is thus composed of traditional moments. Here, we can   have $k(x,x)=1$ for all $x \in \X$ only if $\X$ is included in a centered sphere. Alternatively, we may only impose   $k(x,x) \leqslant 1$ for all $x \in \X$, which still preserves many of our results in later sections.

\paragraph{Translation-invariant kernels on a compact subset of $\rb^d$.} These are the usual kernels $k(x,y)$ of the form $k(x,y) = q(x-y)$ where $q$ is a function with strictly positive Fourier transform, which are all universal~\cite{steinwart2001influence}. Classical examples are the Gaussian kernels $k(x,y) = \exp( - \| x - y\|_2^2 / \sigma^2)$, where $\H$ is a space of infinitely differentiable functions, or the exponential kernel $k(x,y) = \exp( - \| x - y\|_2 / \sigma)$, where $\H$ is the Sobolev space of functions with square-integrable derivatives up to order $(d+1)/2$.
However, they typically do not respect the symmetry of the underlying set (for example for spheres, like below).

\paragraph{Torus.} If $\X = [0,1]$, we can consider kernels of the form $k(x,y) = k(x-y )$ where $k: \rb \to \rb$ is $1$-periodic and has a strictly positive Fourier series (note that we use the same notation for the kernel and the underlying function). We then get a universal kernel and a symmetric set (invariant by periodic translation); {moreover, the base distribution will be the usual uniform distribution}. For example, we will often consider:
$$
k(x) = \frac{ ( 1- e^{-\sigma})^2}{1+e^{-2\sigma} - 2 e^{-\sigma} \cos 2\pi x } = \frac{1}{1 + \frac{\sin^2 \pi x}{\sinh^2 (\sigma/2)}},
$$
for which the Fourier series is  $\ds \hat{k}(\omega) =  \int_0^1 e^{-2i\pi \omega x} k(x) dx = \tanh ( \frac{\sigma}{2} )  e^{-\sigma |\omega|}$ for $\omega \in \Z$.

This extends naturally to $[0,1]^d$, by considering tensor products, that is, $k(x,y) = \prod_{i=1}^d k(x_i-y_i)$, which for example leads to $\hat{k}(\omega) = \tanh^d (\frac{\sigma}{2} ) e^{-\sigma \|\omega\|_1}$ for $\omega \in \Z^d$ in the example above. We will also need an explicit formula for $\sum_{\omega \in \Z} \hat{k}(\omega) \log \hat{k}(\omega) = \log \tanh \frac{\sigma}{2} - \frac{\sigma}{\sinh \sigma}$.

\paragraph{Spheres.} We can also consider $\X$ as the unit $\ell_2$-sphere in dimension $d$, with kernels that are invariant by rotation, leading also to symmetric sets. Kernels are then defined through their expansion in spherical harmonics~\cite{smola2001regularization}, with links with neural networks~\cite{bach2017breaking}.

We could also imagine extensions to Stiefel or Grassman manifolds~\cite{edelman1998geometry}, and of course to all sets $\X$ where positive definite kernels can be defined (such as graphs, trees, permutations, sets, etc.)~\cite{shawe2004kernel}.

\section{Review of quantum information theoretic quantities}

\label{sec:quantum}

We consider here complex-valued Hilbert spaces, but all results will apply to operators on real Hilbert spaces. In quantum mechanics (see, e.g.,~\cite{isham2001lectures}), physical states are represented as elements of a Hilbert space~$\H$, and a physical system is characterized by a convex combination of orthogonal projections on these states, which is often called a density matrix. Our  covariance operators defined as $\Sigma_p = \int_{\X} \varphi(x) \varphi(x)^* dp(x)$ can thus be seen as ``density operators''. It turns out that in quantum information theory~(see, e.g.,~\cite{wilde2013quantum}), information-theoretic quantities are defined based on these density operators. In this section, we review these notions from a mathematical point of view.

 \paragraph{Entropy.} For a compact positive self-adjoint operator $A$ on the complex Hilbert space $\H$ with finite trace, we can define the negative entropy as~\cite{araki2002entropy}:
$$\tr \big[ A \log A \big] = \sum_{\lambda \in \Lambda(A)} \lambda \log \lambda, $$
where $\Lambda(A)$ is the set of eigenvalues of $A$ (with the convention that $0 \log 0 = 0$). It may not be finite if the sequence of eigenvalues is not decreasing sufficiently fast (note that this is unlikely since they are summable). As defined,  $\tr [ A \log A ]$ is always less than $(\tr A) \log (\tr A)$ and can be equal to $-\infty$.

\paragraph{Kullback-Leibler (KL) divergence (relative entropy).}
It is defined as 
$$D(A||B) = \tr \big[ A (\log A - \log B) \big],$$
 for any two positive Hermitian operators with finite trace. It can only be finite if the null space of $B$ is included in the null space of $A$, that is, if $A = B^{1/2} M B^{1/2}$ for a certain bounded operator $M$, and we then have $\tr A - \tr B \leqslant D(A\| B) \leqslant ( \tr A ) \log \| M\|_{\rm op}$.

The quantity $D(A\| B) - \tr A + \tr B$ is always non-negative as the Bregman divergence~\cite{bregman1967relaxation} associated with {the convex function} $A \mapsto \tr \big[ A \log A \big]$, and may not always be finite. The following properties are classical and are proved from first principles in Appendix~\ref{app:monot}.

\begin{proposition}[Properties of quantum entropy and relative entropy]
\label{prop:quantum}
For Hermitian positive operators $A$, $B$, we have:
\BIT
\item[(a)] The mappings $A \mapsto \tr [ A \log A]$ and $(A,B) \mapsto D(A||B)$ are convex.
\item[(b)] We have $D(A\|B) \geqslant \tr A - \tr B$, with equality if and only if $A=B$.
\item[(c)]  If $A_i,B_i$, $i=1,\dots,n,$ are Hermitian operators and $\lambda_i >0$ such that $\sum_{i=1}^n \lambda_i = 1$, then
$$
 D \Big( \sum_{i=1}^n \lambda_i A_i \Big\| \sum_{i=1}^n \lambda_i B_i \Big)
  \leqslant \sum_{i=1}^n \lambda_i D(A_i \| B_i),
$$
with equality if and only if $\log B_i - \log A_i$ does not depend on $i$. Moreover, if there is equality, there exists an operator $M$ such that for all $i$,  $A_i = B_i M$.
\item[(d)]  Monotonicity of quantum operations: Given operators $C_i: \H \to \mathcal{K}$ and $A, B: \H \to \H$,  $i=1,\dots,n$, such that $\sum_{i=1}^n C_i^\ast C_i = \idm$, then
  $$
  D \Big( \sum_{i=1}^n C_i A C_i^\ast \Big\|
  \sum_{i=1}^n C_i B C_i^\ast \Big)  \leqslant D(A\|B).
  $$
  \EIT
\end{proposition}
The last property is classical in quantum information theory as the mapping $A \mapsto  \sum_{i=1}^n C_i A C_i^\ast $ is usually referred to as a ``quantum operation''. Note that quantum operations are usually defined as ``completely positive'' trace-preserving maps, and that owing to  Stinespring representation theorem (see \cite[Theorem 3.1.2]{bhatia2009positive}), they are essentially all of that form.

One particular application of monotonicity is to consider $n$  positive self-adjoint operators $D_i: \H \to \H$, $i=1,\dots,n$, such that $\sum_{i=1}^n D_i = \idm$, and consider $\mu,\nu \in \rb^n$ defined through $\mu_i = \tr (D_i A)$ and $\nu_i = \tr (D_i B)$ for $A,B$ a positive self-adjoint operator with unit trace, then $\mu$ and $\nu$ are on the simplex, and the monotonicity leads to {(see detailed proof in Appendix~\ref{app:monotonicity})}:
\BEQ
\label{eq:munu}
\sum_{i=1}^n \mu_i \log \frac{\mu_i}{\nu_i} \leqslant D(A\| B).
\EEQ

\paragraph{Further results.} In Appendix~\ref{app:jointcvx}, we provide finer results beyond expectation with respect to a probability measure with finite support, as well as equality cases.

 \paragraph{Integral representation.} For the KL divergence, we have, by direct integration~\cite{ando1979concavity}:
\BEQ
\label{eq:intKL}
D(A \| B) = \tr \big[ A (\log A - \log B) \big]   =    - \int_0^{+\infty} \Big( \tr \big[ A ( A + \lambda \idm)^{-1} \big] - \tr \big[ A ( B + \lambda \idm)^{-1} \big] \Big) d\lambda,
\EEQ
which will be used in \mysec{analysisDOF} for the estimation from a finite sample.
As shown in Appendix~\ref{app:monot}, we also have another integral representation formula, that helps in proving additional convexity results.

\section{Kernel entropy and Kullback-Leibler divergence}
\label{sec:entropy}

We define the kernel Kullback-Leibler (KL) divergence, or kernel relative entropy, between two probability distribution $p$ and $q$ as
\BEQ
\label{eq:KL}
D(\Sigma_p\| \Sigma_q) = \tr \big[ \Sigma_p ( \log \Sigma_p - \log \Sigma_q) \big].
\EEQ
{In this section, we list some of its properties, which are similar to the properties of  $(p,q) \mapsto \| \mu_p - \mu_q\|$, but with the added benefit that is has a direct link with Shannon entropy, which is useful in itself when a precise information measure is needed (e.g., for differential privacy~\cite{dwork2006calibrating}), or when used within variational inference (see \mysec{duality}).}

First, the kernel relative entropy is finite under general conditions.
\begin{proposition}[Finiteness of kernel KL divergence]
Assume \textbf{(A1)}. If $p$ is absolutely continuous with respect to $q$, with $\big\| \frac{dp}{dq}\big\|_{\infty}  \leqslant \alpha $, then both the regular and kernel KL divergences are between zero and $ \log \alpha$.
\end{proposition}
\begin{proof}
We have then $\Sigma_p  \preccurlyeq \alpha \Sigma_q$, leading to (note that we must have $\alpha \geqslant 1$):
$$
D(\Sigma_p\| \Sigma_q)  = \tr \big[ \Sigma_p ( \log \Sigma_p - \log \Sigma_q) \big]
\leqslant \tr \big[ \Sigma_p ( \log ( \alpha\Sigma_q) - \log \Sigma_q) \big] = (\log \alpha) \tr \Sigma_p = \log \alpha.
$$
The positivity follows from Prop.~\ref{prop:quantum}, property (b).
\end{proof}
 This implies that if $p$ has a bounded density with respect to the base measure $\tau$ on $\X$, then the relative entropy $D(\Sigma_p \| \Sigma)$ is well-defined, and always non-negative (recall that $\Sigma$ is the covariance operator for the base probability measure $\tau$ on $\X$).
 
We can then define the kernel entropy as
\BEA
\label{eq:H}
  H(\Sigma_p) & = &  - D( \Sigma_p \| \Sigma ) - \min_{x \in \X} \langle \varphi(x), ( \log \Sigma) \varphi(x) \rangle
  \\
\notag  & = &  - \tr \big[ \Sigma_p \log \Sigma_p \big] +  \tr \big[ \Sigma_p \log \Sigma \big]
  - \min_{x \in \X}\, \langle \varphi(x), ( \log \Sigma) \varphi(x) \rangle.
  \EEA
  Note that we add terms on top of $ -\tr \big[ \Sigma_p \log \Sigma_p \big] $, one linear term in $\Sigma_p$, and one constant term, to extend properties of the regular Shannon discrete entropy in all cases.  However, everything simplifies in some cases. Indeed, when there are some invariances that are shared by the kernel {and the base distribution~$\tau$} (that is, a symmetric set as defined in \mysec{symm}), {$x \mapsto  \langle \varphi(x), ( \log \Sigma) \varphi(x) \rangle$ is constant on $\X$, and thus equal to its expectation under $p$, which implies
$  \min_{x \in \X} \langle \varphi(x), ( \log \Sigma) \varphi(x) \rangle
   = \tr [ \Sigma_p \log \Sigma ]$, and then the definition above simplifies to  $ H(\Sigma_p) =  - \tr \big[ \Sigma_p \log \Sigma_p \big]$.
   }

 The kernel information quantities have natural properties mimicking the traditional quantities, as well as a bound by the regular KL divergence.
 
 \begin{proposition}[Properties of kernel entropy and relative entropy]
 \label{prop:propKL}
 Assume \textbf{(A1)}. 
 \BIT
 \item[(a)] The kernel KL divergence defined in \eq{KL} is non-negative, and equal to $0$ for $p=q$. If the kernel $k^2$ is universal, it is zero if and only if $p=q$.
 \item[(b)] The kernel entropy defined in \eq{H} is always between $0$ and $ - \min_{x \in \X} \langle \varphi(x), ( \log \Sigma) \varphi(x) \rangle$ (which is equal to $-\tr \Sigma \log \Sigma$ for symmetric sets). If $p = \tau$ (base measure), it is equal to the upper-bound, while for Dirac measures with a point mass at any minimizers of $x \mapsto \langle \varphi(x), ( \log \Sigma) \varphi(x) \rangle$, it is equal to $0$ (for symmetric sets, this is all of $\X$). If the kernel $k^2$ is universal, these conditions are also necessary.

 \item[(c)] The mapping $(p,q) \mapsto D(\Sigma_p\|\Sigma_q)$ is  convex, and thus the mapping $p \mapsto H(\Sigma_p)$ is concave.
 \item[(d)] For all probability measures $p$ and $q$, $D(\Sigma_p \| \Sigma_q) \leqslant D(p\|q)$. Moreover, the differential entropy {with respect to the base measure $\tau$, that is,} $-\int_\X \log \frac{dp}{d\tau}(x) dp(x)$, is less than $H(\Sigma_p) +  \min_{x \in \X} \langle \varphi(x), ( \log \Sigma) \varphi(x) \rangle$.
 
  \item[(e)] For all probability measures  $p$ and $q$, $D(\Sigma_p\| \Sigma_q) \geqslant \frac{1}{2} \| \Sigma_p - \Sigma_q \|_\ast^2 \geqslant \frac{1}{2} \| \Sigma_p - \Sigma_q \|_{\rm HS}^2 $, where $\| \cdot \|_\ast$ denotes the nuclear norm and $\| \cdot \|_{\rm HS}$ the Hilbert-Schmidt norm.
 \EIT
 \end{proposition}
 \begin{proof}
 Property (a) is the consequence of the non-negativity of the relative entropy when the two arguments have the same trace (here equal to $1$), and  injectivity of the map $p \mapsto \Sigma_p$ when $k^2$ is universal.
  
 For property (b), the upper-bound results follow similarly. For the non-negativity, we notice that $H(\Sigma_p)$ defined in \eq{H} is the sum of two non-negative terms $-\tr [\Sigma_p \log \Sigma_p]$ (negative because all eigenvalues of $\Sigma_p$ are less than one) and $\tr \big[ \Sigma_p \log \Sigma \big]
  - \min_{x \in \X} \langle \varphi(x), ( \log \Sigma) \varphi(x) \rangle$. It is equal to zero if and only the two terms are equal to zero. For the first term, this is equivalent to $\Sigma_p$ being rank one, while for the second one, this is equivalent to $p$ being supported in the set of minimizers. The sufficient condition for being zero thus follows; for the necessary condition, if $k^2$ is universal, $k(x,x')^2 < 1$  if $x \neq x'$, and thus $\Sigma_p$ being rank one implies that $p$ is a Dirac measure.

 Property (c) is a consequence of the convexity of the map $(A,B) \mapsto \tr \big[ A (\log A - \log B) \big] $ {and of the linearity of the mapping $p \mapsto \Sigma_p$}.  Property (d) is obtained from
\BEAS
 D(\Sigma_p \| \Sigma_q)
& = & D \Big( \int_\X \varphi(x)  \varphi(x)^\ast dp(x) \Big\|  
\int_\X \frac{dq}{dp}(x) \varphi(x)   \varphi(x)^\ast dp(x) \Big) \\
& \leqslant & \int_\X
D \Big(
\varphi(x)  \varphi(x)^\ast \Big\|
 \frac{dq}{dp}(x) \varphi(x)   \varphi(x)^\ast
\Big)
dp(x) 
\\
& = &  \int_\X
\| \varphi(x)\|^2 D \Big(1  \Big\|
 \frac{dq}{dp}(x)  
\Big)
dp(x)  = \int_\X \log \Big( \frac{ dp}{dq}(x) \Big) dp(x) = D(p\|q),
 \EEAS
 using joint convexity of the quantum relative entropy, and the fact that $\| \varphi(x)\|=1$ for all $x \in \X$. Note that we can only have equality if for all $x \neq y$, $k(x,y) = 0$, that is, we have an orthonormal embedding, which is only possible for discrete $\X$. Note also that the upper-bound remains valid as soon as $k(x,x) \leqslant 1$ for all $x \in \X$.
 
 Property (e) is simply a consequence of Pinsker inequality for the quantum relative entropy (see~\cite{yu2013strong}).
  \end{proof}

{Note that while the relative entropy results do not depend on the choice of the base measure, the notion of kernel entropy depends on this choice   (however, for the symmetric sets presented in \mysec{examples}, the natural choice of the base measure is the uniform measure).}

As detailed in Appendix~\ref{app:div}, the properties outlined in Prop.~\ref{prop:propKL}, in particular the lower-bound on the relative entropy, is valid for other ``$f$-divergences''~\cite{csiszar1967information,topsoe2000some}, such as the squared Hellinger distance, or the Pearson $\chi^2$-divergence.

   \subsection{Lower bound on kernel relative entropy}
   \label{sec:lower}
Beyond the upper-bound by the Shannon relative entropy (property (d) in Prop.~\ref{prop:propKL}), we can also get a lower bound on the kernel KL divergence, by defining for all $y \in \X$,
 $D(y) =  \Sigma^{-1/2} \big( \varphi(y)   \varphi(y)^\ast  \big) \Sigma^{-1/2}$, which is a positive self-adjoint operator, so that
  $\int_\X D(y) d\tau (y) = \idm$. We can then apply the monotonicity of the relative entropy with operators $D(y)^{1/2}$, $y \in \X$, that is, \eq{munu} extended to general expectations.
  We have 
  $$
  \tr [ D(y) \Sigma_p ] = \int_\X  \langle \varphi(x), \Sigma^{-1/2} \varphi(y) \rangle^2 dp(x)
   = \int_\X h(x,y) dp(x),
  $$
   defining the function $h: \X \times \X \to \rb_+$ as $h(x,y) =  \langle \varphi(x), \Sigma^{-1/2} \varphi(y) \rangle^2$, which is non negative and such that
  $$\ds \forall y \in \X, \  \int_\X h(x,y) d\tau(x) = 
  \langle \varphi(y), \Sigma^{-1/2} \Big( \int_\X \varphi(x)   \varphi(x) ^\ast d\tau (x)\Big) \Sigma^{-1/2} \varphi(y) \rangle = \langle \varphi(y), \varphi(y) \rangle = 1.$$
  The function $h$ can thus be seen as a smoothing kernel\footnote{Note here that the smoothing property of $h$ corresponds to the pointwise non-negativity, but that $h$ is also a positive definite kernel.}. 
  Thus, {by monotonicity of the relative entropy (Appendix~\ref{app:monotonicity})}, we get:
  $$
  D(\tilde{p} \| \tilde{q} ) \leqslant D(\Sigma_p \| \Sigma_q),
  $$
  where $\ds \tilde{p}(y) = \tr [ D(y) \Sigma_p ]  =  \int_\X h(x,y) dp(x)$ is a smoothed version of $p$ (with the same definition for $\tilde{q}$).
  Note that we have directly $  D(\tilde{p} \| \tilde{q} ) \leqslant D(p\| q)$ {from standard Markov chain arguments in information theory}\footnote{{The distributions $\tilde{p}$ and $\tilde{q}$ are obtained by the same transition kernel, respectively from the distributions $p$ and $q$.}}~\cite[Section 2.9]{cover1999elements}. Overall, we have the following sequence of inequalities:
  \BEQ
  \label{eq:sandwitch}
   D(\tilde{p} \| \tilde{q} ) \leqslant D(\Sigma_p \| \Sigma_q) \leqslant D(p\|q).
  \EEQ
These can lead to quantitative bounds between $D(\Sigma_p \| \Sigma_q)$ and $ D(p\|q)$, in particular when the smoothing function $h$ is putting most of its mass on pairs $(x,y)$ where $x$ is close to $y$. This is made explicit below for  distributions on a metric space and distributions with Lipschitz-continuous densities.
  
  \subsection{Small-width asymptotics for metric spaces}
  \label{sec:asymptotics}
  
 We now provide a bound between $ D(\tilde{p} \| \tilde{q} ) $ and $  D( {p} \|  {q} ) $ in \eq{sandwitch}, and thus between $D(\Sigma_p \| \Sigma_q)$ and  $  D( {p} \|  {q} ) $, {when the set $\X$ is equipped with a distance $d$, used to characterize the regularity of density functions.}\footnote{{This section is not crucial to understand the rest of the paper.}}
  
  This corresponds to bounding the difference in the classical data processing inequality for the Shannon entropy, using regularity properties of the densities of $p$ and $q$. The following proposition is shown in Appendix~\ref{app:asymptotics}.
  
  \begin{proposition}[Joint bound on relative entropy]
  \label{prop:asymptotics}
We assume that $p, q$ have strictly positive Lipschitz-continuous densities (denoted $p$ and $q$) with respect to the base measure $\tau$.
  Using definitions from \mysec{lower},  we have:
\BEQ
   0 \leqslant D(p\|q) - D(\tilde{p} \| \tilde{q} ) 
   \leqslant
   E(p,q) \times  
        \sup_{x \in \X} \int_\X h(x,y) d(x,y)^2 dy, \label{eq:entbound}
       \EEQ
       with $E(p,q) =  2 C_{p/q}^2 
       \sup_{x \in \X} \frac{q(x)}{p(x)} +  4 \sup_{x \in \X} \frac{q(x)}{p(x)} ( C_p + C_q)^2  
       \sup_{x \in \X} \frac{p(x)}{q(x)} (1 + C_p^2 {\rm diam}(\X)^2 )$, where $C_p, C_q$ and $C_{p/q}$ are Lipschitz constants satisfying
         for all $x,y \in \X$, $\big| \frac{q(x)}{q(y)} - 1 \big| \leqslant d(x,y) C_q$, 
   $\big| \frac{p(x)}{p(y)} - 1 \big| \leqslant d(x,y) C_p$, and
    $\big| \frac{p(y)}{q(y)} -  \frac{p(x)}{q(x)}  \big| \leqslant d(x,y) C_{p/q}$ for some distance $d$ on $\X$. 
  \end{proposition}
  Thus, for Lipschitz-continuous densities, the approximation of $D(p\|q)$ by $D(\Sigma_p \| \Sigma_q)$ is controlled by a kernel dependent quantity, that we now study in special cases.
  
  \paragraph{Discrete orthonormal embeddings.} For these kernels where $k(x,y) = 1_{x = y}$, then we have $\Sigma = \idm$ and thus $h(x,y)= 1_{x = y}$; therefore the quantity $\sup_{x \in \X} \int_\X h(x,y) d(x,y)^2 dy$ is always zero, which is natural here since in this situation, kernel and Shannon entropies are equal.

  \paragraph{Torus.} We have by definition $h(x,y) =   \langle \varphi(x), \Sigma^{-1/2} \varphi(y) \rangle^2$.
   For a kernel on the torus $[0,1]$ of the form $k(x,y) = k(x-y)$ where $k$ has Fourier series $\hat{k}$, then we have:
  $$h(x,y) = \Big(\sum_{\omega \in \Z} \hat{k}(\omega)^{1/2} e^{ 2i\pi \omega (x - y) } \Big)^2,$$
  and for the particular choice of $\hat{k}(\omega) \propto e^{-\sigma |\omega|}$, we have:
  $\ds
  h(x,y) =  \frac{\tanh \frac{\sigma}{2} }{\tanh^2 \frac{\sigma}{4} }
  \big( 1 + \frac{\sin^2 \pi (x-y)}{\sinh^2 (\sigma/4)} \big)^{-2}$,
  and for $d(x,y) = |\sin \pi(x-y) |$, then, using exact integration, we get:
  $$
  \int_{y-1/2}^{y+1/2} h(x,y)d(x,y)^2 dx=
   \frac{\tanh \frac{\sigma}{2} }{\tanh^2 \frac{\sigma}{4} } 
   \frac{\sinh^{3} (\sigma/4) }{2(1+\sinh^2 (\sigma/4) )^{3/2}}  \leqslant \frac{\sigma^2}{16}.
  $$
  For the $d$-dimensional torus, and, $d(x,y)^2=\sum_{j=1}^d  \sin^2 \pi(x_j-y_j)$,  we get the bound $\ds \frac{d\sigma^2}{16}$, and hence a nice scaling in dimension (non-exponential).
Overall, in this situation, we thus get an approximation of entropy from \eq{entbound} up to $O(\sigma^2)$, where $\sigma$ is the width of the kernel. {Thus, if our goal is to estimate the true relative entropy, and not only have a lower bound, we need to have $\sigma$ going to zero, which corresponds to RKHSs which are bigger and bigger. See discussion in \mysec{examplesEST} when estimating quantities from a finite sample.}

Note that this extends to all translation invariant kernels on $[0,1]^d$ (and not only for $\hat{k}(\omega) \propto e^{-\sigma \| \omega\|_1}$), and more generally to translation-invariant kernels on compact spaces.

\paragraph{Alternative proof through leverage scores.}
While we provided a proof  of the $O(\sigma^2)$ approximation of the entropy through quantum data processing inequalities, the relationship between the regular entropy and the kernel entropy can also be studied through the integral representation in \eq{intKL}. We indeed have, using \eq{intKL}:
\BEAS
D(\Sigma_p\| \Sigma_q)  & = & 
\int_0^{+\infty} \tr \Big[  {\Sigma}_p \big( (  {\Sigma}_q + \lambda \idm)^{-1} -  (  {\Sigma}_p + \lambda \idm)^{-1}\big) \Big] d\lambda \\
& = &  \int_0^{+\infty} \int_{\X} \big\langle {\varphi}(x)  ,  \big( (  {\Sigma}_q + \lambda \idm)^{-1} -  (  {\Sigma}_p + \lambda \idm)^{-1}\big)   {\varphi}(x)\big\rangle dp(x) d\lambda \\
& = &  \int_0^{+\infty} \int_{\X}\Big(
 \big\langle {\varphi}(x)  ,    (  {\Sigma}_q + \lambda \idm)^{-1}{\varphi}(x)\big\rangle
  -  \big\langle {\varphi}(x)  ,   (  {\Sigma}_p + \lambda \idm)^{-1}    {\varphi}(x)\big\rangle \Big) dp(x) d\lambda.
\EEAS

 We now consider that both $p$ and $q$ have densities with respect to the base measure $\tau$ (which we also denote $p$ and $q$). The main idea, using and extending results from \cite{pauwels2018relating}, is that for all $x$ and $\lambda$, and $p$ smooth enough, the quantity $ \langle \varphi(x),   (  {\Sigma}_p + \lambda\idm)^{-1}  \varphi(x) \rangle$, often referred to as a  ``leverage score'' can be approximated as $\langle \varphi(x),     (   p(x)   \Sigma + \lambda \idm)^{-1}  \varphi(x) \rangle$.

We thus consider the potential approximation:
\BEAS
&& \int_0^{+\infty}\int_{\X} \Big[ 
 \langle \varphi(x),     (  q(x)  \Sigma + \lambda \idm)^{-1}  \varphi(x) \rangle
 -  \langle \varphi(x),      (  (p(x)  \Sigma + \lambda \idm)^{-1}   \varphi(x) \rangle
 \Big] p(x)  d\tau(x) d\lambda
 \\
 & = & 
 \int_0^{+\infty}   \int_{\X}\Big( 
\tr \Big[ \Sigma    (   q(x) \Sigma + \lambda \idm)^{-1}  \Big] 
- \tr \Big[ \Sigma    (  p(x)   \Sigma + \lambda \idm)^{-1}  \Big] 
 \Big)  p(x) d\tau(x) d\lambda
 \\
& = & 
   \int_{\X}\tr \big[ 
 \Sigma  ( \log p(x)  \Sigma  - \log q(x)   \Sigma    ) 
   \big]p(x)    d\tau(x) =   \int_{\X} p(x) \log \frac{p(x)}{q(x)} d\tau(x),
\EEAS
which is exactly the traditional KL divergence $D(p\|q)$. In order to show a bound similar to \eq{entbound},  a non-asymptotic extension of results from \cite{pauwels2018relating} could be carried out.

\section{Estimation from finite sample}
\label{sec:estimation}

Given that the covariance operator is an infinite-dimensional operator, naively we would need eigenvalue decompositions of such operators to compute the kernel information quantities, which is computationally hard. Thus finite-dimensional estimation algorithms are needed. In this section, we provide   algorithms  for estimating the kernel entropy $H(\Sigma_p)$ from independent and identically distributed (i.i.d.) samples from $p$ {(as required for model fitting~\cite{binkowski2018demystifying} or for independent component analysis~\cite{muandet2017kernel}}), as well as from certain kernel integrals.  Estimators from other oracles on $p$ or $q$, as well as estimators for $D(\Sigma_p \| \Sigma_q)$ {could} derived and analyzed similarly.\footnote{{With potentially extra terms due to the term $\tr \big[\hat{\Sigma}_p \log \hat{\Sigma}_q\big]$.}}

\subsection{Estimators}
\label{sec:est}
 Given $\ds \Sigma_p = \int_{\X}  \varphi(x)   \varphi(x)^\ast dp(x)$ and 
$x_1,\dots,x_n$ sampled i.i.d.~from $p$, we consider the natural estimator
\BEQ
\label{eq:Shat}
\hat{\Sigma}_p = \frac{1}{n} \sum_{i=1}^n \varphi(x_i)   \varphi(x_i)^\ast.
\EEQ
 The natural estimate   for $\tr \big[ \Sigma_p \log \Sigma_p \big]$ is 
 $\tr \big[ \hat{\Sigma}_p \log \hat{\Sigma}_p \big]$, that can be computed from the kernel matrix $K \in \rb^{n \times n}$ defined as $K_{ij} = k(x_i,x_j)$.

 \begin{proposition}[entropy for empirical covariance estimators]
 With $\hat{\Sigma}_p$ defined in \eq{Shat}, and the kernel matrix defined above, we have
 \BEQ
 \label{eq:entK}
  \tr \big[ \hat{\Sigma}_p \log \hat{\Sigma}_p \big]
 = \tr \Big[ \frac{1}{n} K  \log \big(\frac{1}{n} K\big) \Big]  
. \EEQ
 \end{proposition}
 \begin{proof}
The non-zero eigenvectors of $\hat{\Sigma}_p$ belong to the image space of $\hat{\Sigma}_p$ and are thus linear combinations $f = \sum_{j=1}^n \alpha_j \varphi(x_j)$ for $\alpha \in \rb^n$. Then
 $\ds
 \hat{\Sigma}_p f = \frac{1}{n} \sum_{i=1}^n\sum_{j=1}^n \alpha_j \big[ \varphi(x_i)   \varphi(x_i) ^\ast\big] \varphi(x_j) = \frac{1}{n} \sum_{i=1}^n (K \alpha)_i \varphi(x_i).$ 
 Thus, if $K\alpha = n \lambda \alpha$,   $ \hat{\Sigma}_p f  = \lambda f$, and if  $ \hat{\Sigma}_p f  = \lambda f$ with $\lambda \neq 0$ and $f \neq 0$ (which implies $K\alpha \neq 0$), then $\ds \sum_{i=1}^n \big[ (K \alpha)_i -  n\lambda \alpha_i \big] \varphi(x_i) = 0 $, which implies $K^2 = n\lambda K \alpha$ and then $K\alpha = n \lambda \alpha$ since $K\alpha  \neq 0$. Thus, the non-zero eigenvalues of $\hat{\Sigma}_p$ are exactly the ones of $\frac{1}{n}K$, and we thus get
$
  \tr \big[ \hat{\Sigma}_p \log \hat{\Sigma}_p \big]
 = \tr \Big[ \frac{1}{n} K  \log \big(\frac{1}{n} K\big) \Big]  
.$
\end{proof}
More generally, for any approximation of $\Sigma_p$ as $\sum_{i=1}^n \eta_i \varphi(x_i)   \varphi(x_i)^\ast$, with positive weights $\eta_i$ that sum to one (for example obtained from a finite grid and not sampling), we would get $$ \tr \big[ \hat{\Sigma}_p \log \hat{\Sigma}_p \big]
 = \tr \big[ \Diag(\eta)^{1/2} K   \Diag(\eta)^{1/2} \log \big( \Diag(\eta)^{1/2}K  \Diag(\eta)^{1/2} \big) \big].$$

\paragraph{Running-time complexity.} In order to compute the eigenvalue decomposition of $K$ above (needed to compute the matrix logarithm), we need a running time of $O(n^3)$. This can be reduced by using column sampling techniques~\cite{boutsidis2009improved}, also known as ``Nystr\"om method'' in this context~\cite{williams2000using}: by projecting the kernel matrix to the span of $m$ of its columns, it leads to a running time in $O(m^2 n)$, with some approximation that could be controlled explicitly~\cite{rudi2015less}.

\subsection{Analysis}
\label{sec:analysisDOF}

{In order to quantify the difference between $\tr \big[ \hat{\Sigma}_p \log \hat{\Sigma}_p \big]$ and $\tr \big[  {\Sigma}_p \log  {\Sigma}_p \big]$, we could use regular perturbation arguments for eigenvalues of self-adjoint operators~\cite{kato}. However, the function $t \mapsto t \log t$ has diverging derivatives around zero, which prevents their direct use.}
Instead, given the integral representation in \eq{intKL}, we see that it will be sufficient to estimate quantities of the form
$\tr\big[ (\Sigma_p +\lambda \idm)^{-1} \Sigma_p \big]$, which are usually referred so as ``degrees of freedom'' in kernel methods~\cite{de2021regularization}. The natural estimator is $\tr\big[ (\hat\Sigma_p +\lambda \idm)^{-1} \hat\Sigma_p \big]$ and its performance has already been thoroughly studied (see~\cite{rudi2017generalization}), and an important quantity there is the maximal leverage score for the base distribution $\sup_{x \in \X}\,  \langle \varphi(x),  (\Sigma+\lambda \idm)^{-1} \varphi(x) \rangle$, which could be replaced by the the average one (which is then equal to the degrees of freedom) if ``leverage score sampling'' is used (see \cite{de2021regularization} for more details). For simplicity, we do not consider leverage score sampling, noting that for symmetric sets, it does not change anything.

The next proposition shows that these estimation results for fixed $\lambda$ can be integrated over $\lambda$ to get an estimation bound for the entropy (see proof in Appendix~\ref{app:proof_est}),  \emph{without} the need for extra regularization, leading to an estimation in $O(1/\sqrt{n})$.
\begin{proposition}
\label{prop:proof_est}
Assume \textbf{(A1)} and that $ {p}$ has a density with respect to the base measure which is greater than $\alpha<1$. Assume that $\ds c=  \int_0^{+\infty}\! \! \sup_{x \in \X}\,  \langle \varphi(x),  (\Sigma+\lambda \idm)^{-1} \varphi(x) \rangle^2 d\lambda  $ is finite. Given i.i.d.~samples $x_1,\dots,x_n$ from $p$, and the estimator defined by   \eq{entK}, we have:
$$
\E  \Big[ 
 \big|    \tr \big[ \hat{\Sigma}_p \log \hat{\Sigma}_p \big]  -    \tr \big[  {\Sigma}_p \log  {\Sigma}_p \big]\big| \Big]
  \leqslant  \frac{1 +   c( 8 \log n)^2}{n \alpha }
 +   \frac{17}{\sqrt{n}} \big(
 2 \sqrt{c}   + \log n
 \big)
 .$$
\end{proposition}
Note that as a consequence of Jensen's inequality, we always have
$\E \big(  \tr \big[ \hat{\Sigma}_p \log \hat{\Sigma}_p \big]  \big) \geqslant   \tr \big[  {\Sigma}_p \log  {\Sigma}_p \big]$, as can be seen in simulations (see \myfig{simulations_entropy}).

 \subsection{Examples}
 \label{sec:examplesEST}
We consider the $d$-dimensional torus $[0,1]^d$ with a translation-invariant kernel $k(x,y) = k(x-y)$ with $\hat{k}(\omega) = \hat{k}(0) e^{-\sigma \| \omega\|_1}$. We have by symmetry
$\ds
\langle \varphi(x),  (\Sigma+\lambda \idm)^{-1} \varphi(x) \rangle
= \tr \big[ \Sigma ( \Sigma +\lambda \idm)^{-1} \big]
= \sum_{\omega \in \Z^d} \frac{ \hat{k}(\omega)}{\hat{k}(\omega) + \lambda}.
$
As shown in Appendix~\ref{app:df}, we have the bound $\ds 
\tr \big[ \Sigma ( \Sigma +\lambda \idm)^{-1} \big] \leqslant e^{-\sigma d / 2}  \frac{ d!  }{ \sinh^d(\sigma/2)} \Big[
1 +   \Big( \log \frac{ \tanh^d   \frac{\sigma}{2}}{\lambda}   \Big)^{d} 
\Big]$. Thus, the constant $c$ in Prop.~\ref{prop:proof_est} above  can be upper bounded by a constant times $\sigma^{-d}$ (see Appendix~\ref{app:df} for details).
We thus get an overall estimation rate {for the kernel entropy} proportional to $\sigma^{-d/2}{\sqrt{n}}$.   

When balancing with the $O(\sigma^2)$ estimation of the entropy for Lipschitz-continuous strictly positive densities in \mysec{asymptotics}, we can take $  \sigma \propto {n^{-1/(d+4)}}$ {(lower bandwidth with more observations)}, leading to a rate of ${n^{-2/(d+4)}}$, which is to be compared to the optimal rate equal to $ {n^{-4/(d+4)}}$ for entropy estimation for densities with the same regularity~\cite{han2020optimal}. Although our primary goal was not to estimate entropies, we  get a rate which is the square root of the optimal rate.

\paragraph{Illustrative experiments.} We consider $\X = [0,1]$ and $p(x) =  4\big| x - \frac{1}{2} \big|$, with negative differential entropy $\ds  \int_0^1 p(x) \log p(x) dx = \log 2 - \frac{1}{2} \approx 0.1931$. 
We consider estimation from various numbers of i.i.d.~observations. See left plot in \myfig{simulations_entropy}, where we compare two values of $\sigma$. Note that the estimation is empirically here in expectation from above, as expected because of Jensen's inequality (as opposed to the right plot).

\begin{figure}
\begin{center}
\includegraphics[width=5.5cm]{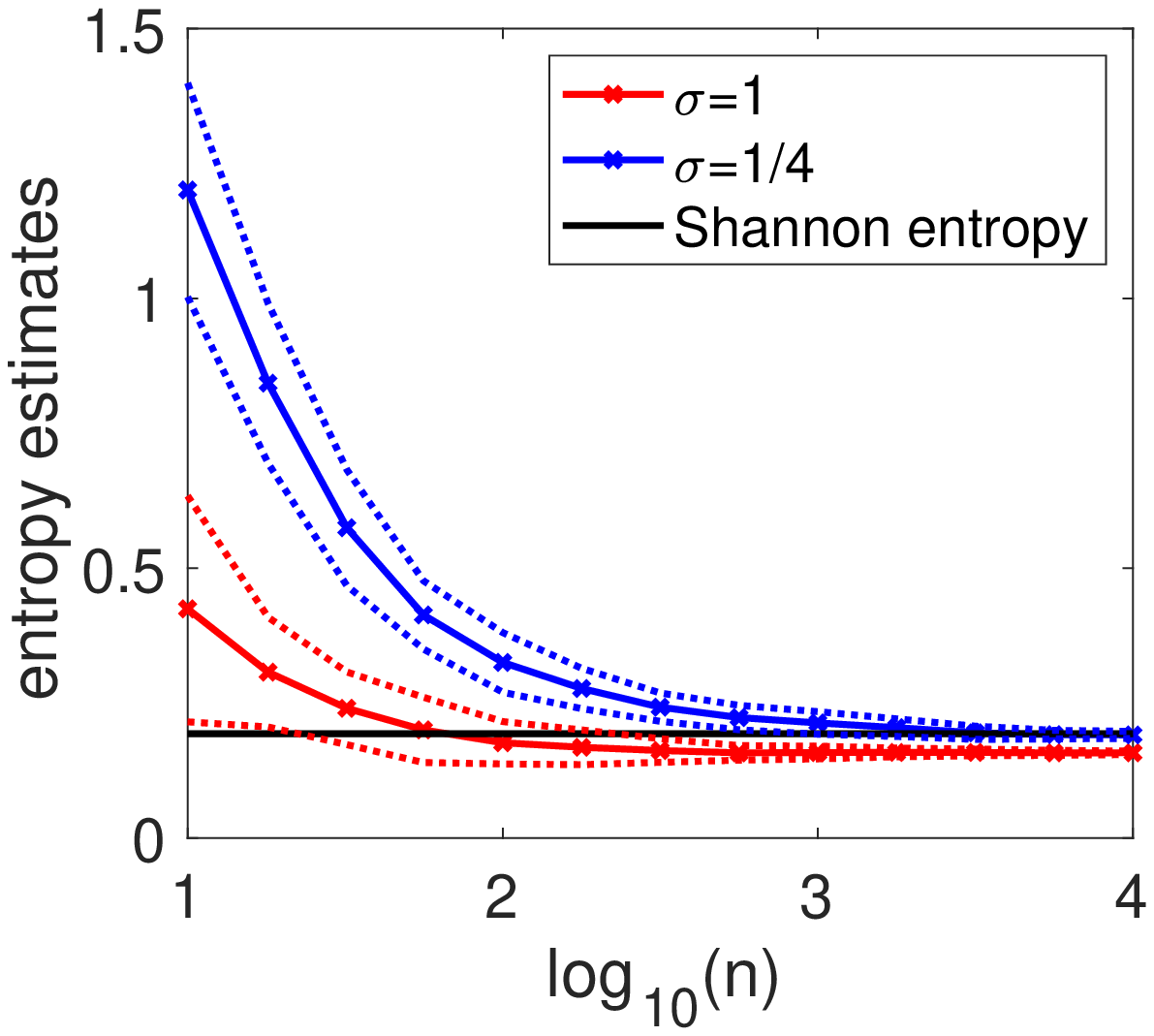} \hspace*{2cm}
\includegraphics[width=5.5cm]{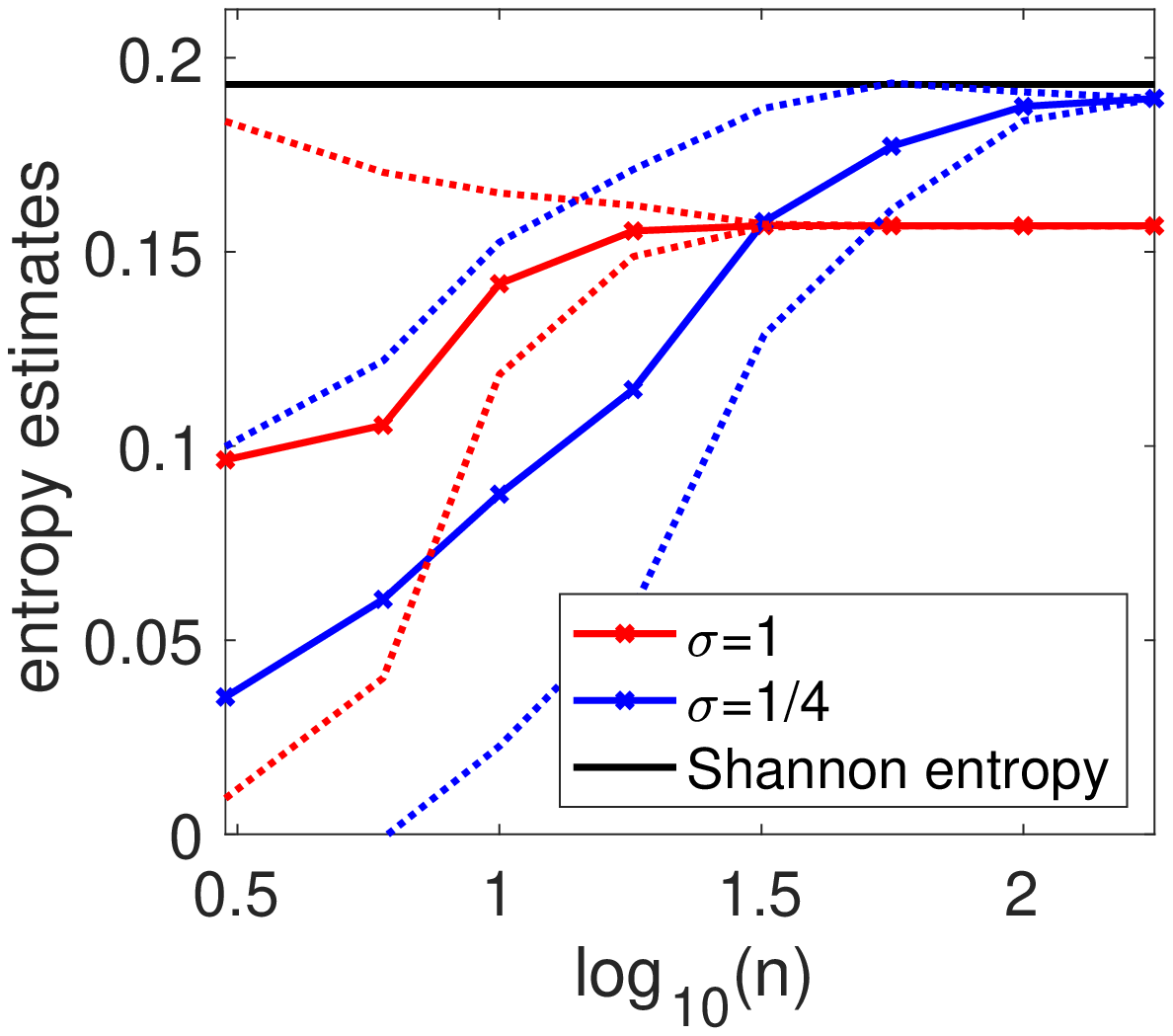}
\end{center}

\vspace*{-.45cm}

\caption{Negative entropy estimation from 20 replications (mean and standard deviation) for two values of the kernel bandwidth $\sigma$, as $n$ increases. Left: from i.i.d.~samples (method from \mysec{est}). Right: from kernel integrals evaluated as uniform samples (method from \mysec{otheroraclaes}).
As expected we estimate the kernel entropy (limit when $n$ tends to infinity) from above on the left plot, and from below on the right plot, with a  faster convergence rate on the right plot.  \label{fig:simulations_entropy}}
\end{figure}

\subsection{Estimation from other oracles}
\label{sec:otheroraclaes}
If we have access to other oracles on $p$ and $q$, we can consider other estimators.\footnote{{This section is not crucial to understand the rest of the paper.}}
For example, if we can estimate $\int_\X k(x,y)k(x,z) dp(x)$ and $\int_\X  k(x,y)k(x,z) dq(x)$ for any $y,z \in \X$, we can simply sample $n$ points $x_1,\dots,x_n$ from the base measure on $\X$, and, given the projection operator $\hat{\Pi}$ on the span of $\varphi(x_1),\dots,\varphi(x_n)$, consider
\BEQ
\label{eq:otherest}
\hat{\Sigma}_p = \hat{\Pi}  {\Sigma}_p \hat{\Pi}  \mbox{ and } \hat{\Sigma}_q = \hat{\Pi}  {\Sigma}_q \hat{\Pi} ,
\EEQ
and the estimator $D(\hat{\Sigma}_p  \| \hat{\Sigma}_q)$ which is now a pointwise lower bound\footnote{{Since $\hat{\Pi} \hat{\Pi}^\ast + (\idm - \hat{\Pi})(\idm -\hat{\Pi})^\ast = \idm$, the monotonicity of relative entropy (Prop.~\ref{prop:monot} in Appendix~\ref{app:monotonicity}) leads to $D(\Sigma_p \| \Sigma_q) \geqslant 
D\big( \hat{\Pi}  {\Sigma}_p \hat{\Pi}  + (\idm - \hat{\Pi} )   {\Sigma}_p (\idm - \hat{\Pi}) \big \|  \hat{\Pi}  {\Sigma}_q \hat{\Pi}  + (\idm - \hat{\Pi} )   {\Sigma}_q (\idm - \hat{\Pi}) \big)
= D\big( \hat{\Pi}  {\Sigma}_p \hat{\Pi} \|  \hat{\Pi}  {\Sigma}_q \hat{\Pi}   \big) + D \big(   (\idm - \hat{\Pi} )   {\Sigma}_p (\idm - \hat{\Pi}) \|  (\idm - \hat{\Pi} )   {\Sigma}_q (\idm - \hat{\Pi}) \big)$, which is greater than $D\big( \hat{\Pi}  {\Sigma}_p \hat{\Pi} \|  \hat{\Pi}  {\Sigma}_q \hat{\Pi}   \big)$.}} on $D( {\Sigma}_p  \|  {\Sigma}_q)$ (and thus on $D(p\|q)$), which can be beneficial in variational inference, as presented in \mysec{duality}.

{Noting that $\langle \varphi(x_i), \Sigma_p \varphi(x_j) \rangle = \int_\X k(x_i,x)k(x,x_j) dp(x)$ (and similarly for $\Sigma_q$), and that $K$ is the matrix of dot-products of $\varphi(x_1),\dots,\varphi(x_n)$, the orthogonal projection may be done as}
$$
D(\hat{\Sigma}_p  \| \hat{\Sigma}_q)
= D\Big( K^{-1/2} \int_\X k(X,x)k(x,X) dp(x) K^{-1/2} \Big\| K^{-1/2} \int_\X k(X,x)k(x,X) dq(x) K^{-1/2} \Big),
$$
where $k(X,x) \in \rb^n$ is the vector with components $k(x_i,x)$, $i=1,\dots,n$, and $K \in \rb^{n \times n}$ the kernel matrix, typically with some added regularization for numetical stability, and incomplete Cholesky decomposition so as to minimize the number of kernel evaluations~(see, e.g.,  \cite{bach2005predictive} and references therein), with an overall running time complexity of $O(n^3)$. A key property is that it always underestimates the kernel relative entropy (and hence the Shannon relative entropy). The following proposition shows that it provides a better estimate (better dependence in $n$)  than with i.i.d.~samples (see proof in Appendix~\ref{app:proof_est_other}).

\begin{proposition}
\label{prop:proof_est_other}
Assume \textbf{(A1)} and that $ {p}$ has a density with respect to the base measure which is bounded.  Given uniform i.i.d.~samples $x_1,\dots,x_n$ from the base distribution, and the estimator defined by \eq{otherest}, we have, if for all $\lambda > 0$,
$ \sup_{x \in \X}\,  \langle \varphi(x),  (\Sigma+\lambda \idm)^{-1} \varphi(x) \rangle
\leqslant     B \log( C / \lambda) ^d
$, for some constants $B,C$:
$$
 0 \leqslant \tr \big[\hat\Pi {\Sigma}_p \hat\Pi \log (\hat\Pi {\Sigma}_p \hat\Pi ) \big] -  \tr \big[  {\Sigma}_p \log  {\Sigma}_p \big]   
  = O \Big( \big(
  \log ( B C ) \big)^d ( \sqrt{BC} + B ) \exp\big( - \frac{1}{4} (Bn)^{1/(d+1)} \big)
 \Big) .
$$
\end{proposition}
The proposition above applies to the torus situation in $[0,1]^d$ and the kernel with exponential Fourier series, with now an improved dependence in $n$ and in $d$, since in that case $B \sim C \sim \sigma^d$. That is, to reach precision $\varepsilon$, we need
$n \propto \sigma^{-d}  \big( \log \frac{\sigma^d}{\varepsilon} \big)^{d+1}$ instead of $n ~\sim  {\sigma^{-d} \varepsilon^{-2}}$. Other decays for the maximal degrees of freedom are considered in Appendix~\ref{app:proof_est_other}.

\paragraph{Illustrative experiments.} We  consider as above $p(x) =  4\big| x - \frac{1}{2} \big|$, and consider estimation from various numbers of i.i.d.~observations for the projection. The required integrals are computed by numerical integration. See right plot in \myfig{simulations_entropy}, where we compare two values of $\sigma$. Note that the estimation is as expected here from below, and with a faster rate of convergence than for the estimation method from \mysec{est}.

\section{Multivariate probabilistic modelling}
\label{sec:multivariate}

Information-theoretic quantities are particularly useful at characterizing statistical dependence between random variables. {This is for example the basis of independent component analysis~\cite{cardoso2003dependence}, and can be used for feature selection~\cite{song2012feature}}. We explore how the new notions of entropy can be used in this goal.

\subsection{Joint entropy, marginal entropy and mutual information}
We assume that we have a product set $\X = \X_1 \times \X_2$, and a feature space $\H = \H_1 \otimes \H_2$, corresponding to the kernel $k(x,y) = k_1(x_1,y_1) k_2(x_2,y_2)$, and feature map $\varphi(x) = \varphi_1(x_1) \otimes \varphi_2(x_2)$.  {We also consider a base distribution on $\X_1 \times \X_2$ which is the product of the base distributions.} The tensor product is here defined as the span of all functions $(x_1,x_2) \mapsto f_1(x_1)   f_2(x_2)$ for $f_i \in \H_i$, $i=1,2$, with the norm $\|f_1\| \cdot  \| f_2\|$. Given a joint probability $p_{X_1X_2}$ and its marginal probability distributions $p_{X_1}$ and $p_{X_2}$, we can define the covariance operators $\Sigma_{p_{X_1X_2}}$ on $\H_1 \otimes \H_2$,
$\Sigma_{p_{X_1}}$ on $\H_1$, and $\Sigma_{p_{X_2}}$ on $\H_2$, as well as $\Sigma_{12}$, $\Sigma_1$ and $\Sigma_2$ the covariance operators of the   base distributions on $\X_1 \times \X_2$, $\X_1$ and $\X_2$:

We then define the marginal and joint entropies as:
\BEAS
H(\Sigma_{p_{X_1}}) & = &- \tr \big[ \Sigma_{p_{X_1}} \log \Sigma_{p_{X_1}} \big] +  \tr \big[ \Sigma_{p_{X_1}} \log \Sigma_1 \big]
  - \min_{x_1 \in \X_1} \langle \varphi_1(x_1), ( \log \Sigma_1) \varphi_1(x_1) \rangle  \\
H(\Sigma_{p_{X_2}}) & = &- \tr \big[ \Sigma_{p_{X_2}} \log \Sigma_{p_{X_2}} \big] +  \tr \big[ \Sigma_{p_{X_2}} \log \Sigma_2 \big]
  - \min_{x_2 \in \X_2} \langle \varphi_2(x_2), ( \log \Sigma_2) \varphi_2(x_2) \rangle    \\
H(\Sigma_{p_{X_1X_2}}) & = & 
- \tr \big[ \Sigma_{p_{X_1 X_2}} \log \Sigma_{p_{X_1 X_2}} \big] +  \tr \big[ \Sigma_{p_{X_1 X_2}} \log \Sigma_{12} \big]
  - \min_{x \in \X} \langle \varphi(x), ( \log \Sigma_{12} ) \varphi(x) \rangle. 
\EEAS
The product distribution $p_{X_1} p_{X_2}$ (with independent components with same  marginal distributions as $X$), {which is different in general from the joint distribution $p_{X_1 X_2}$}, has covariance operator $\Sigma_{p_{X_1} p_{X_2}} = \Sigma_{p_{X_1}} \otimes  \Sigma_{p_{X_2}} $, and its entropy is:
\BEAS
H(\Sigma_{p_{X_1}p_{X_2}}) & = & 
- \tr \big[ \Sigma_{p_{X_1}} \otimes  \Sigma_{p_{X_2}} \log \Sigma_{p_{X_1}} \otimes  \Sigma_{p_{X_2}}  \big] +  \tr \big[ \Sigma_{p_{X_1}} \otimes  \Sigma_{p_{X_2}} \log \Sigma_{12} \big]
  - \min_{x \in \X} \langle \varphi(x), ( \log \Sigma_{12} ) \varphi(x) \rangle. 
\EEAS
We have $\log \Sigma_{p_{X_1}} \otimes  \Sigma_{p_{X_2}}  = 
\log\Sigma_{p_{X_1}} \otimes  \idm + \log \idm \otimes \Sigma_{p_{X_2}} $. Moreover, because the base distribution is separable {by assumption},  we have $\langle \varphi(x), ( \log \Sigma_{12} ) \varphi(x) \rangle = 
\langle \varphi_1(x_1), ( \log \Sigma_{1} ) \varphi_1(x_1) \rangle + \langle \varphi_2(x_2), ( \log \Sigma_{2} ) \varphi_2(x_2) \rangle$, thus, we exactly have:
$$
H(\Sigma_{p_{X_1}p_{X_2}}) = H(\Sigma_{p_{X_1}}) + H(\Sigma_{p_{X_2}}) .
$$
We can therefore define the {kernel} mutual information {(not to be confused with the regular Shannon mutual information)} as
 $$
I(X_1,X_2) = D( \Sigma_{p_{X_1 X_2}} \| \Sigma_{p_{X_1}} \otimes \Sigma_{p_{X_2}}).
  $$
From the developments above, it is equal to
$$ I(X_1,X_2) = H(\Sigma_{p_{X_1}}) + H(\Sigma_{p_{X_2}}) - H(\Sigma_{p_{X_1 X_2}}),$$
mimicking the traditional equality for regular entropies. Moreover, it is equal to zero if and only if the variables are independent.
 \begin{proposition}[Characterization of independence]
  Assume \textbf{(A1)}. Then $I(X_1,X_2)$ is a lower bound on the Shannon mutual information, and $  I( X_1,X_2)=0$ if and only if $ \Sigma_{p_{X_1X_2}} = \Sigma_{p_{X_1}} \otimes \Sigma_{p_{X_2}}
  $. If the kernel $k^2$ is universal, this is equivalent to the random variables $X_1$ and $X_1$ being independent.
  \end{proposition}
  \begin{proof}
  The first statement is straightforward because of the interpretation as a relative entropy. The second statement is simply the consequence of the fact that $\Sigma_{p_{X_1}} \otimes \Sigma_{p_{X_2}}$ is the covariance operator of the product distribution $p_{X_1} p_{X_2}$.
  \end{proof}
We can make the following observations:
\BIT
\item The extension to mutual information between more than two random variables is straightforward by considering tensor products between all marginals, with the same characterization property of independence.
\item Our use of covariance operators is related to the ``kernel generalized variance'' from~\cite{bach2002kernel}, which considers a joint covariance operator $ \ds \bigg( \begin{array}{cc}
\E \big[  \varphi_1(X_1) \varphi_1(X_1)^\ast \big]  & \E \big[  \varphi_1(X_1) \varphi_2(X_2)^\ast \big]  \\
\E \big[  \varphi_2(X_2) \varphi_1(X_1)^\ast \big]  & \E \big[  \varphi_2(X_2) \varphi_2(X_2)^\ast \big] \end{array} \bigg) = \E \bigg[ { \varphi_1(X_1) \choose \varphi_2(X_2) }
 { \varphi_1(X_1) \choose \varphi_2(X_2) }
^\ast \bigg]$, and compute the mutual information for a Gaussian vector with this covariance operator. The kernel generalized variance can be generalized to more than two random variables, but then the independence characterization property is not satisfied. Moreover, the relationship with  the usual mutual information is not as straightforward.

\item The extension to conditional entropies is not straightforward. 
We could define $\Sigma_{p_{X_2|X_1 =x_1}}$ in the natural way, but then we do not have the joint entropy being the sum of the marginal and conditional entropy (otherwise, because of axiomatic characterizations of entropy~\cite{csiszar2008axiomatic}, we would get back exactly the Shannon entropy).
There is also a potential for exponential families and conditional models akin to generalized models, which remains to be explored.

\EIT

\subsection{Data processing inequality}

  We have a form of data processing inequality, which is classical in quantum information theory (see, e.g.,~\cite{petz2003monotonicity}), which we extend here.
  \begin{proposition}[Data processing inequality]
  \label{prop:datap}
  Assume \textbf{(A1)}.  We have
  $$
  D(\Sigma_{p_{X_1X_2}} \| \Sigma_{q_{X_1X_2}} ) \geqslant   D(\Sigma_{p_{X_1}} \| \Sigma_{q_{X_1	}} ) ,
  $$
   with equality implying (when the kernel $k$ is universal)  that the conditional distribution of $X_2$ given $X_1$ is the same for $p$ and $q$, but  the reverse implication is not true in general. 
  \end{proposition}
  \begin{proof}
  The marginal operators $\Sigma_{p_{X_1}}$ and $\Sigma_{q_{X_1}}$ are obtained by taking ``partial traces'' which are quantum operations as defined in \mysec{quantum}, hence the inequality. Following~\cite{ruskai2002inequalities}, we have inequality if and only $\log \Sigma_{p_{X_1X_2}} -  \log \Sigma_{p_{X_1X_2}} 
  = \log (\Sigma_{p_{X_1}} \otimes \idm)  -\log  (\Sigma_{q_{X_1}} \otimes \idm) $, which is unfortunately not always satisfied when conditional distributions are equal. However, in case of equality, the extra condition {(c)} from Prop.~\ref{prop:quantum}, becomes $\Sigma_{p_{X_1X_2}} ( \Sigma_{q_{X_1}} \otimes \idm) = \Sigma_{q_{X_1 X_2}} ( \Sigma_{p_{X_1}} \otimes \idm)$, and can be leveraged. Indeed,  for all test functions $f_1,g_1,f_2,g_2$ in $\H$, it leads to:
  $$
  \int_{\X_1 \times \X_2} \int_{\X_1}
  f_1(x_1) f_2(x_2) g_1(y_1) g_2(y_2) k_1(x_1,y_1)
  dp_{X_1X_2} (x_1,x_2) dq_{X_1}(y_1)
  $$
  $$= \int_{\X_1 \times \X_2} \int_{\X_1}
  f_1(x_1) f_2(x_2) g_1(y_1) g_2(y_2) k_1(x_1,y_1)
  dq_{X_1X_2} (x_1,x_2) dp_1(y_1).
$$ Considering $f_1$ and $g_1$ tending to Dirac functions at $\bar{x}_1$ (which is possible since the kernel is univeral) we get 
{$\int_{ \X_2}  
  f_2(x_2)  g_2(y_2)  
  dp_{X_1X_2} (\bar{x}_1,x_2) dq_{X_1}(\bar{x}_1)
  = \int_{ \X_2}  
 f_2(x_2)  g_2(y_2)  
  dq_{X_1X_2} (\bar{x}_1,x_2) dp_1(\bar{x}_1),$ which implies that  the conditional distributions of $x_2$ given $x_1$ is the same for $p$ and $q$.} \end{proof}

  \subsection{Submodularity and conditional independence}
We can now consider three variables $X_1,X_2,X_3$, and use strong sub-additivity of the quantum relative entropy~\cite{ruskai2002inequalities,petz2003monotonicity}.
\begin{proposition}[Submodularity]
Assume \textbf{(A1)}. We have:
 $$
  H(\Sigma_{p_{X_1X_2 X_3}}) -  H(\Sigma_{p_{X_1X_2}}) -  H(\Sigma_{p_{X_2X_3}}) +  H(\Sigma_{p_{X_2}}) \leqslant 0,
  $$
  with equality (for universal kernels) implying that $X_1$ and $X_3$ are independent given $X_2$, but in general the reverse implication is not true.\end{proposition}
 \begin{proof}
 As  done by~\cite{petz2003monotonicity}, we can simply apply the data processing result (Prop.~\ref{prop:datap}) to $\tilde{X}_1 = (X_1,X_2),$ and $\tilde{X}_2 = X_3$, with the distributions
 $\tilde{p}_{\tilde{X}_1 \tilde{X}_2} = p_{X_1X_2X_3}$ and  $\tilde{q}_{\tilde{X}_1 \tilde{X}_2} = p_{X_1}p_{X_2X_3}$.
 \end{proof}
As a classical consequence, the entropy function is submodular like the regular entropy~(see, e.g.,~\cite[Section 6.5]{bach2013learning}). 

Note that we do not obtain a necessary and sufficient  characterization  of conditional independence, which can be obtained with other tools based on covariance operators~\cite{fukumizu2007kernel}.

\section{Convex duality, log-partition functions and variational inference}
\label{sec:duality}

Given that our kernel notions of relative entropies are lower bounds on the regular notions, by the traditional convex duality results between maximum entropy and log-partition functions (see, e.g.,~\cite{wainwright2008graphical}), we should obtain upper-bounds on log-partition functions. In this section, we show how such bounds can be obtained and computed. We then show that such bounds can also be optimized with respect to the positive definite kernel.

{This leads to a new family of variational inference methods, which can then be used in various probabilistic inference tasks~\cite{wainwright2008graphical}, and in particular in Bayesian inference~\cite{robert2007bayesian}.}

In this section, we will sometimes only assume that for our kernels satisfy (on top of being positive definite) $k(x,x) \leqslant 1$ for all $x \in \X$, noting that the upper-bound result on the regular Shannon relative entropy is still valid in this case.

\subsection{Convex duality between operators}
\label{sec:dualop}
We have for self-adjoint operators (regardless of the link with input spaces and probabilities):
\BEAS
\sup_{A \succcurlyeq 0, \ \tr A = 1}  \tr ( AM ) - D(A\|B)
& = & \sup_{A \succcurlyeq 0, \ \tr A = 1}  \tr ( A (M + \log B) ) - \tr  [A \log A] \\
& = & \log \tr \exp ( M + \log B)  ,
\EEAS
with the optimal operator $A$ equal to $\ds A = \frac{1}{ \tr \exp ( M + \log B)} \exp ( M + \log B)$. This implies the representation
$$
D(A\|B) = \sup_{ M } \ \tr A M - \log \tr \exp ( M + \log B)  ,
$$
with the constraints that $A \succcurlyeq 0$ and $\tr A = 1$ are automatically satisfied (if $ \tr A \neq 1$, then by replacing $M$ by $M + t \idm$, we can make the quantity go to infinity). 

We also have (when the constraint for unit trace is removed):
\BEAS
\sup_{A \succcurlyeq 0 }  \tr ( AM ) - D(A\|B)
& = & \sup_{A \succcurlyeq 0 }  \tr ( A (M + \log B) ) - \tr  [A \log A] \\
& = &  \tr \exp ( M + \log B - \idm),
\EEAS
with an optimal $A = \frac{1}{e} \exp( M + \log B)$.
This implies the representation
$$
D(A\|B) = \sup_{ M } \ \tr A + \tr A M -   \tr \exp ( M + \log B)  ,
$$
with the constraint that $A \succcurlyeq 0$ is automatically satisfied (but not $\tr A = 1$), with the optimal $M$ equal to $\log A - 
\log B$.

\subsection{Bounds on the log-partition function}
We can now apply the duality relationships above to covariance operators. Given a bounded function $f: \X \to \rb$ and a distribution $q$ on $\X$, our goal is to obtain upper-bounds on the log-partition function $ \log \int_\X e^{f(x)} dq(x)$ which is equal to, by convex duality:
$$
  \log \int_\X e^{f(x)} dq(x) = \sup_{ p \ {\rm probability}} \int_\X f(x) dp(x) - D( p \| q).
$$

 In our situation, our operators are subject to belong to the subspace $\A$, equal to the span of all $\varphi(x)\varphi(x)^\ast$, for $x \in \X$, which leads to slight adaptations. We do not necessarily assume that the kernels are universal, so that the hull of all $\varphi(x)\varphi(x)^\ast$ may not be equal to the intersection of $\A$ and positive operators with unit trace. 
\paragraph{Isotropic kernels.}
If $k(x,x)=1$, for all $x \in \X$, then $\tr \Sigma_p = 1$ is equivalent to $\int_\X dp(x) = 1$.
We define for a bounded function $f: \X \to \rb$, 
$$
a(f) =   \sup_{ p \ {\rm probability\ measure}} \int_\X f(x) dp(x) - D(\Sigma_p \| \Sigma_q).
$$
Using convex duality tools from \mysec{dualop}, we have:
\BEAS
a(f) & = & \sup_{ p \ {\rm non-negative\  measure}}  \inf_{M} \int_\X f(x) dp(x) - \tr \Sigma_p  M + \log \tr \exp ( M + \log \Sigma_q)  
\\
& = &  \inf_{M} \log \tr \exp ( M + \log \Sigma_q)  \mbox{ such that } \forall x \in \X, \ f(x) \leqslant  \langle \varphi(x), M \varphi(x) \rangle.
\EEAS
Note that we can get rid of the constraint $\int_\X dp(x) = 1$ because it is equivalent to $\tr \Sigma_p = 1$.

A relaxation is to relax $p$ not to be a non-negative measure (and use the positivity of $\Sigma_p$ as the only constraint), leading to 
\BEQ
\label{eq:bf}
 b(f)   =   \sup_{ p \ {\rm  measure}}  \ \int_\X f(x) dp(x) - D(\Sigma_q \| \Sigma_q ) .
\EEQ
By construction, we have $a(f) \leqslant b(f)$, and by convex duality,
\BEA
\label{eq:bfdual}
b(f) &= &   \sup_{ p \ {\rm  measure}}  \inf_{M} \int_\X f(x) dp(x) - \tr \Sigma_p  M + \log \tr \exp ( M + \log \Sigma_q) \\
\notag &  = &    \inf_{M}\  \log \tr \exp ( M + \log \Sigma_q)  \mbox{ such that } \forall x \in \X, \ f(x) =  \langle \varphi(x), M \varphi(x) \rangle.
\EEA

Note that that if the problem above is non-feasible because $\varphi$ is too small (so that $f$ cannot be represented as a quadratic form in $\varphi$), then we get a value $+\infty$.

If we can write (usually non-uniquely), $f(x) = \langle \varphi(x), F \varphi(x) \rangle$ for some self-adjoint operator $F$, then we have another representation as:
\BEAS
b(f)
& = &  \inf_{N \in \mathcal{A}^\perp} \ \log \tr \exp ( F + N + \log \Sigma_q) .
\EEAS
Since $D(\Sigma_p \| \Sigma_q) \leqslant D(p\|q)$, we have:
$$
b(f) \geqslant a(f)   \geqslant   \sup_{ p \ {\rm probability}} \int_\X f(x) dp(x) - D( p \| q)
= \log \int_\X e^{f(x)} dq(x).
$$
As expected, we thus obtain upper bounds on the usual log-partition function.

Note that we could also get lower bounds using tools from \mysec{lower}, in particular for functions $f$ that can be written as $f(x) = \int_\X h(x,y) g(y) d\tau(y)$.

\paragraph{Non isotropic kernels.}

If we only assume that $k(x,x) \leqslant 1$ for all $x \in \X$, then $b(f)$ defined in \eq{bf} and \eq{bfdual} is not any more an upper-bound but we can define instead
\BEAS
b^{(c)}(f) & = & \sup_{ p \ {\rm  measure}, \ \int_\X dp(x)=1, \ \int_\X  k(x,x) dp(x) \leqslant 1}  \inf_{M} \int_\X f(x) dp(x) - \tr \Sigma_p  M - \tr \Sigma_p + \tr \exp ( M + \log \Sigma_q)  
\\& = & \sup_{ p \ {\rm  measure} }  \inf_{M, \, c, \, b \geqslant 0 } \int_\X f(x) dp(x) - \tr \Sigma_p  M - \tr \Sigma_p +  \tr \exp ( M + \log \Sigma_q)  \\
& & \hspace*{6cm}
- c \Big( \int_\X dp(x) - 1 \Big) 
- b \Big( \int_\X k(x,x) dp(x) - 1 \Big) 
\\
 & = &  \inf_{M, \, c, \, b \geqslant 0 }   \ c + b +   \tr \exp ( M + \log \Sigma_q)  \mbox{ such that } \forall x \in \X, \ f(x) = c +   \langle \varphi(x), ( M + \idm + b \idm) \varphi(x) \rangle
 \\
 & = &  \inf_{M, \, c  }   \ c + \inf_{b \geqslant 0} \big( b +  e^{-b}  \tr \exp ( M + \log \Sigma_q)  \big) \mbox{ such that } \forall x \in \X, \ f(x) = c +   \langle \varphi(x), ( M + \idm  ) \varphi(x) \rangle \\
& = &  \inf_{M, \, c  }   \ c + 1 + \widetilde{\log} \tr \exp ( M + \log \Sigma_q)  \big) \mbox{ such that } \forall x \in \X, \ f(x) = c +   \langle \varphi(x), ( M + \idm  ) \varphi(x) \rangle
\\
& = &  \inf_{M, \, c  }   \ c + 1 + \widetilde{\log} \big( \frac{1}{e} \tr \exp ( M + \log \Sigma_q)  \big) \big)\mbox{ such that } \forall x \in \X, \ f(x) = c +   \langle \varphi(x), M \varphi(x) \rangle,
\EEAS
with $\widetilde{\log}(a) = \log a$ if $a>1$, and $a-1$ otherwise.

If $f$ can be represented through the operator  $F$, and the constant function through the operator $U$, then we get
$$
b^{(c)}(f) = 
\inf_{M, \, c  }   \ c + 1 + \widetilde{\log} \big( \frac{1}{e} \tr \exp ( M + \log \Sigma_q)  \big) \big)\mbox{ such that } 
F - c U - M \in \mathcal{A}^\perp.
$$
We can check that if $U=\idm$, then the constraints in the definition of $b^{(c)}(f)$ are equivalent to $\tr \Sigma_p$, and we exactly recover the expression of $b(f)$ in \eq{bfdual}.

\paragraph{Properties of upper-bounds.}

The definition of upper-bounds on the log-partition function naturally leads to several questions and applications:
\BIT

\item Can these upper-bounds be tight? In situations like in \mysec{asymptotics} where the entropy lower-bound can be made tight (for example when the kernel bandwidth tends to zero), we should also get tight upper-bounds on the log-partition functions, that is, upper-bounds that tend to the classical log-partition function when the bandwidth goes to zero. We leave for future work a detailed study of these approximations.

\item Can these upper-bounds be estimated efficiently from simple oracles on the function $f$? We provide efficient algorithms in \mysec{computable}.

\item Can these upper-bounds be minimized? We show in \mysec{variational} that these upper-bounds happen to be convex in the kernel, so minimizing with respect to the kernel is possible.

\EIT
 
 \subsection{Relationship with optimization}
 
When adding a temperature parameter $\varepsilon > 0$, we can extend the traditional link between between optimization and log-partition functions. This corresponds essentially to considering a potential $\frac{1}{\varepsilon}f $, or defining
$$
b_\varepsilon(f) =  \varepsilon b\big( \frac{1}{\varepsilon}f \big) = \sup_{ p \ {\rm  measure}}  \ \int_\X f(x) dp(x) - \varepsilon D(\Sigma_q \| \Sigma_q ) .
$$
By convex duality, we get:
$$
b_\varepsilon(f) =     \inf_{M}\  \varepsilon \log \tr \exp \big( \frac{1}{\varepsilon} M + \log \Sigma_q\big)  \mbox{ such that } \forall x \in \X, \ f(x) =  \langle \varphi(x), M \varphi(x) \rangle.
$$
When $\varepsilon$ tends to zero, then $b_\varepsilon(f)$ converges to
$$
     \inf_{M}\  \lambda_{\max}(M) \mbox{ such that } \ \forall x \in \X, \ f(x) =  \langle \varphi(x), M \varphi(x) \rangle.
$$
Given that $\| \varphi(x)\|=1$ for all $x \in \X$, by writing $M = c \idm - A$ for $A$ positive,  this is equal to
$$
     \inf_{c \in \rb, \ A \succcurlyeq 0 }\ c\ \  \mbox{ such that }\  \forall x \in \X, \ f(x) =  c -\langle \varphi(x), A \varphi(x) \rangle,
$$
which is exactly the optimization formulation of~\cite{rudi2020finding} based on ``kernel sums-of-squares''. Following the traditional relationship between log-partition functions and maxima, we can thus consider our bound on log-partition functions as a smoothed version on the maximum. Therefore,  some of the approximation techniques developed for the optimization formulation~\cite{rudi2020finding} can be extended to our set-up as well. Moreover, as common in convex optimization, the entropy could be used for smoothing and accelerated optimization algorithms~\cite{nesterov2005smooth}.

\subsection{Computable bounds}
\label{sec:computable}

In order to compute or approximate $b(f)$, which depends on the function $f: \X \to \rb$ and the feature map $\varphi:\X\to\H$, we will consider a finite-dimensional approximation $\tilde{\varphi}: \X \to \tilde{\H}$ where $\tilde{\H}$ will be a finite-dimensional subspace of $\H$. All of our approximations will always be upper-bounds on the true log-partition function.

We will always need an efficient representation of $\tilde \A$ the span of all $\tilde \varphi(x) \tilde\varphi(x)^\ast$, for $x \in \X$, and sometimes of $\tilde{\mathcal{M}}$ the hull of all  $\tilde \varphi(x) \tilde\varphi(x)^\ast$ for $x \in \X$. See examples in \mysec{extor} and \mysec{exhyp} for the torus and the hypercube.

In order to estimate the relative entropy $D(\Sigma_p \| \Sigma_q)$, we will use $D(\tilde\Sigma_p \| \tilde\Sigma_q)$ instead. If the kernel $(x,y) \mapsto k(x,y) - \tilde{k}(x,y)$ is positive definite, then this is a lower bound, and thus the bounds on the regular quantities are preserved.

In order to compute $D(\tilde\Sigma_p \| \tilde\Sigma_q)$, we assume we can compute explicitly  $\ds \int_\X \varphi(x)\varphi(x)^\ast dp(x)$ and  $\ds \int_\X \varphi(x)\varphi(x)^\ast dq(x)$. If $\tilde{\varphi}$ is obtained from projection on a span of some $\varphi(x_1),\dots,\varphi(x_n)$ like in \mysec{otheroraclaes}, then this can be done if we can compute expectations of $k(x,y) k(x,z)$ for any $y,z \in \X$. We can also build set-specific approximate feature maps (as done in \mysec{extor} and \mysec{exhyp}).

\paragraph{With an explicit representation of $f$.}
If we have an explicit (non-unique) representation of $f$ as $f(x) = \langle \tilde\varphi(x), \tilde{F}  \tilde \varphi(x) \rangle$.  We then have if $\langle \tilde \varphi(x), \tilde\varphi(x) \rangle =1$ for all $x \in \X$ (a similar approximation can be obtained for $b^{(c)}(f)$ for non-isotropic approximations):
\BEAS
\tilde{b}(f) & = & \sup_{ \tilde C \in \mathcal{A}, \ \tilde C \succcurlyeq 0, \ \tr \tilde C = 1}   \tr [ C \tilde{F} ] 
- D(  C \| \tilde \Sigma_q)  \\
& = & \inf_{\tilde M}   \ \log \tr \exp ( \tilde M + \log \tilde\Sigma_q)  
+ \sup_{\tilde C \in \A} \tr [ \tilde C ( F - \tilde M )  ]\\
& = & \inf_{ \tilde M}   \ \log \tr \exp ( \tilde M + \log \tilde\Sigma_q)  
\mbox{ such that } \tilde M - \tilde{F} \in \mathcal{A}^\perp,
\EEAS
which is a finite-dimensional convex optimization problem.
Since the function $\tilde M \mapsto  \log \tr \exp ( \tilde M + \log \tilde\Sigma_q)  $ is smooth, the optimization problem above  can be approximately solved by projected gradient descent (with unit step-size), with iteration:
$$
\tilde M \leftarrow \Pi_{\tilde{F} + \mathcal{A}^\perp} \Big( \tilde M - \frac{  \exp ( \tilde M + \log \tilde\Sigma_q)  }{\tr  \exp ( \tilde M + \log \tilde\Sigma_q)  } \Big).
$$
We then obtain the (approximated) optimal $\tilde C$ as $\tilde C =  \frac{  \exp ( \tilde M + \log \tilde\Sigma_q)  }{\tr  \exp ( \tilde M + \log \tilde\Sigma_q)  } $. It can be accelerated using extrapolation~\cite{beck2009fast}. Note that for any feasible $\tilde M$ we obtain a lower bound on $\tilde{b}(f)$.

\paragraph{Without an explicit representation of $f$.}
If an explicit representation of $f$ as $f(x) = \langle \tilde\varphi(x), \tilde F  \tilde\varphi(x) \rangle$ is not available, we can build an approximate representation of $f$ such that  $\big\|  f  -  \langle\tilde \varphi(\cdot), \tilde F \tilde\varphi(\cdot) \rangle  \big\|_\infty$ is as small as possible, and then use the previous algorithm, with this extra approximation factor. {On the torus, we could also use the same technique as \cite{blakecolt22} based on approximating the $L_\infty$-norm with Fourier transforms.}

\paragraph{Illustrative experiments.} We consider here   $\X = [0,1]$, and the kernel  with Fourier transforms proportional to $e^{-\sigma |\omega|}$. We consider the function $f(x) = \cos (2\pi x)$, and approximate $\log( \int_0^1 e^{f(x)}dx ) \approx 0.2359$. We could either use the finite feature map from \mysec{otheroraclaes} based on random projections (which is a non-isotropic approximation). For simplicity, we consider $\tilde{\varphi}(x)_\omega = \hat{k}(\omega) e^{2i\pi \omega x}$, only for $\omega \in \{-r,\dots,r\}$. We have $\| \tilde{\varphi}(x) \|^2 = \sum_{\omega = -r}^r \tanh ( \frac{\sigma}{2} )  e^{-\sigma |\omega|}
= \tanh ( \frac{\sigma}{2} ) \big( -1 + 2 \frac{1 - e^{-\sigma(r+1)}}{1-e^{-\sigma}} \big)
= 1 - \frac{ e^{-\sigma(r+1/2)}}{\cosh  ( \frac{\sigma}{2} ) } $, and the set $\mathcal{A}$ is exactly the set of Toeplitz matrices. See \myfig{simulations_logpartition} for an illustration with several values of $\sigma$.

\begin{figure}
\begin{center}
\includegraphics[width=5.5cm]{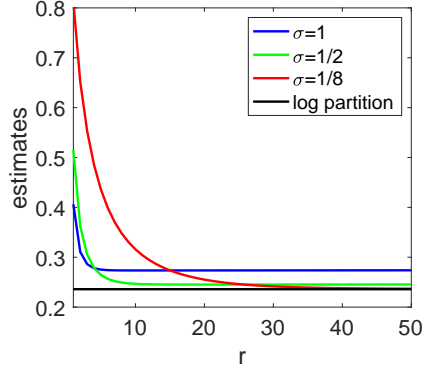}  
\end{center}

\vspace*{-.45cm}

\caption{Log-partition estimation with different values of $\sigma$, as a function of $r$, the number of considered frequencies.
\label{fig:simulations_logpartition}}
\end{figure}

\subsection{Kernel learning}
\label{sec:variational}

The key property we will leverage is the \emph{concavity} of $D(\Sigma_p \| \Sigma_q)$ with respect to the kernel. This allows to maximize efficiently with respect to the kernel to get as close to the Shannon relative entropy $D(p\|q)$ as possible.

\paragraph{Concavity with respect to the kernel.}
Given two distributions $p$ and $q$, then $D(\Sigma_p \| \Sigma_q)$ depends on the kernel function $k: \X \times \X \to \rb$. This happens to be a \emph{concave} function of $k$ (which is only assumed to be positive definite with no bounds on $k(x,x)$). Indeed, if $k_1$ and $k_2$ are two kernels with feature maps $\varphi_1:\X \to \H_1$ and $\varphi_2:\X \to \H_2$, with the associated covariance operators $\Sigma_p^{(1)}$, $\Sigma_p^{(2)}$, $\Sigma_q^{(1)}$, $\Sigma_q^{(2)}$, and $\eta_1, \eta_2 \geqslant 0$ such that $\eta_1 + \eta_2 = 1$,
then with $\H = \H_1 \otimes \H_2$ and $\varphi(x) = ( \sqrt{\eta_1} \varphi_1(x), \sqrt{\eta_2} \varphi_2(x))$, using $E_1 = \bigg( \begin{array}{cc}
\idm & 0 \\
0 & 0 \end{array} \bigg)$ and $E_2 = \bigg( \begin{array}{cc}
0 & 0 \\
0 & \idm \end{array} \bigg)$,   using the monotonicity of relative entropy for quantum operations, we have:
$$
\eta_1 D(\Sigma_p^{(1)} \| \Sigma_q^{(1)} )
+\eta_2 D(\Sigma_p^{(2)} \| \Sigma_q^{(2)} )
  =    D\big( E_1 \Sigma_p E_1 + E_2 \Sigma_p E_2 \big\|  E_1 \Sigma_q E_1 + E_2 \Sigma_q E_2\big) \\
 \leqslant  D ( \Sigma_p \| \Sigma_q),$$
 which shows the concavity.
 
Moreover, the relative entropy is monotonic, that is, if $k_1-k_2$ is positive definite, $D(\Sigma_p^{(2)} \| \Sigma_q^{(2)} )
\leqslant D(\Sigma_p^{(1)} \| \Sigma_q^{(1)} )$. Therefore, when optimizing over the kernel, we only need to look at maximal elements with respect to the positive definite order.

\paragraph{Linear parametrization of kernels.}
Given a finite-dimensional feature map $ {\psi}: \X \to \rb^d$, we look at kernels of the form
$$k(x,y) = \psi(x)^\top \Lambda \psi(x),$$ where $\Lambda \in \rb^{d \times d}$ is a symmetric positive semidefinite matrix. In order to have the upper-bound on entropy, we consider the set $\mathcal{L}$ of matrices $\Lambda$ such that for all $x \in \X$, $\psi(x)^\top \Lambda \psi(x) \leqslant 1$. This is a finite-dimensional convex set, for which we need to know \emph{inner} approximations (to preserve directions of bounds). We then have that
$$
D(\Sigma_p \| \Sigma_q) = D\Big( \int_\X \Lambda^{1/2} \psi(x) \psi(x)^\top \Lambda^{1/2}  dp(x) 
\Big\| \int_\X \Lambda^{1/2} \psi(x) \psi(x)^\top \Lambda^{1/2}  dq(x)  \Big)
$$
is convex in $\Lambda$ (see a direct proof in Appendix~\ref{app:concav}). As mentioned earlier, by monotonicity, we only need to consider maximal elements of $\mathcal{L}$.

\paragraph{Relative entropy maximization.}
Given $p$ and $q$, from which we can compute
$
C_p = \E_{X \sim p} \big[ \psi(X) \psi(X)^\top\big]
$
and $
C_q =  \E_{X \sim q} \big[\psi(X) \psi(X)^\top \big], 
$ we maximize 
$D( \Lambda^{1/2}  C_p  \Lambda^{1/2}  \| \Lambda^{1/2}  C_q  \Lambda^{1/2}  )$ with respect to $ \Lambda \in \mathcal{L}$.

We can decompose it as the difference of two convex functions as:
\BEAS
D( \Lambda^{1/2}  C_p  \Lambda^{1/2}  \| \Lambda^{1/2}  C_q  \Lambda^{1/2}  )
& = & \tr\big[  C_p^{1/2}  \Lambda C_p^{1/2} \log( C_p^{1/2}  \Lambda C_p^{1/2}
 ) \big] - \tr\big[  \Lambda^{1/2}  C_p  \Lambda^{1/2}   \log(  \Lambda^{1/2}  C_q  \Lambda^{1/2}  
 ) \big].
\EEAS
This naturally leads to an optimization algorithm where we lower-bound the first convex function by its tangent and then maximize this affine function minus the other convex functions. This is a parameter-free globally and monotonically convergent algorithm.

\paragraph{Log-partition minimization.} Given $q$ and $f: \X \to \rb$, we can try to minimize the lower bound on the log-partition function instead. We assume for simplicity that we know a representation of $f$ as $f(x) = \psi(x)^\top G \psi(x)$, that the span $\mathcal{B}$ of all $\psi(x)\psi(x)^\top$ for $x \in \X$ is manageable, and that (for simplicity), $\psi(x)^\top \Lambda \psi(x) = 1$ for all $\Lambda \in \mathcal{L}$.

We thus need to minimize with respect to $\Lambda \in \mathcal{L}$ the following function (which is convex in $\Lambda$ as the supremum of convex functions):
\BEAS
& & \sup_{C \in \mathcal{B}, \ C \succcurlyeq 0, \ \tr C = 1}
\tr ( C G) - D( \Lambda^{1/2}  C  \Lambda^{1/2}  \| \Lambda^{1/2}  C_q  \Lambda^{1/2}  ) \\
& = & \inf_{M} \sup_{C \in \mathcal{B}} \tr C ( G - \Lambda^{1/2}  M \Lambda^{1/2}  )
+ \log \tr \exp( M + \log(\Lambda^{1/2}  C_q  \Lambda^{1/2} ))
\\
& = & \inf_{M} 
\  \log \tr \exp( M + \log(\Lambda^{1/2}  C_q  \Lambda^{1/2} )) \mbox{ such that }
 G - \Lambda^{1/2}  M \Lambda^{1/2} \in \mathcal{B}^\perp.
\EEAS
The optimal $C$ and $M$ can be obtained by accelerated projected gradient descent. We can then minimize with respect to $\Lambda$ by updating $\Lambda$ like for the entropy maximization, once $M$ (and thus $C$) have been estimated.

We now provide two illustrating examples, on the torus and the hypercube.

\subsection{Example 1: Torus $[0,1]$}
 \label{sec:extor}
 
  We consider $\X = [0,1]$, and assume that we have $m = 2r+1$ frequencies $\omega \in \{-r,\dots,+r\}$, and $\psi_\omega(x) =  e^{2i\pi \omega x}$. Then we have
   $(C_p)_{\omega \omega'} =  \hat{p}(\omega - \omega')$ and $(C_q)_{\omega \omega'} =  \hat{q}(\omega - \omega')$, and the constraint to be in $\mathcal{B}$ is that $C_p $ and $C_q$ are Toeplitz matrices  (note that this requires to know $\hat{p}(\omega)$  and $\hat{q}(\omega)$ for $\omega \in \{-2r, \dots, 2r\}$).
 
If $f$ has band-limited Fourier series $\hat{f}$ (zero outside $\{-2r,\dots,2r\}$), then we have $f(x) = \langle \psi(x), G \psi(x) \rangle$, with any matrix $G$ such that $\sum_{\omega - \omega' = \delta}  G_{\omega \omega'} = \hat{f}(\delta)$ for all $\delta \in \{-2r,\dots,2r\}$. We then have $\tr [C_p G ] = \sum_{\omega, \omega'=-r}^r  \hat{p}(\omega - \omega')^\ast G_{\omega \omega'} = \sum_{\delta=-2r}^{2r} \hat{p}(\delta)^\ast \hat{f}(\delta) $.

Moreover, we consider $\Lambda = \Diag(\eta)$ with $\eta$ in the simplex, so that,
for $\hat{q}(\omega) = 1_{\omega \neq 0}$ (for the uniform distribution), we have
$$D( \Lambda^{1/2}  C_p  \Lambda^{1/2}  \| \Lambda^{1/2}  C_q  \Lambda^{1/2}  )
= \tr\big[  C_p^{1/2}  \Lambda C_p^{1/2} \log( C_p^{1/2}  \Lambda C_p^{1/2}
 ) \big]   - \sum_{\omega=-r}^r  \eta(\omega) \log \eta(\omega).$$

\paragraph{Relative entropy maximization.}
We can  maximize with respect to $\eta$, using the (here convergent by convexity) the majorization-minimization approach, e.g., by using the affine Taylor expansion of the first term, and minimizing
$\ds
 \tr \big[  C_p^{1/2} \Diag(\eta_+)   C_p^{1/2}  \log(  C_p^{1/2}  \Diag(\eta_+)   C_p^{1/2}  )\big] - \sum_{\omega=-r}^r \eta_+(\omega)\log \eta_+(\omega)
$
with respect to $\eta^+$ is in the simplex, thus leading to $\log \eta^+$ equal to a constant plus   
$\ds\diag\big[
 C_p^{1/2}  \log(   C_p^{1/2}  \Diag(\eta)  C_p^{1/2}  )  C_p^{1/2}  \big].$
This makes an algorithm in closed form to maximize $D(\Sigma_p \| \Sigma_q) 
$ with respect to $\eta$. 

This is illustrated in \myfig{kernel_learning_entropy}, where we compare various lower bounds on the negative entropies, either with a fixed kernel, or with kernel learning, showing the benefits of estimating $\hat{k}(\omega)$.

\begin{figure}
\begin{center}
\includegraphics[width=5.5cm]{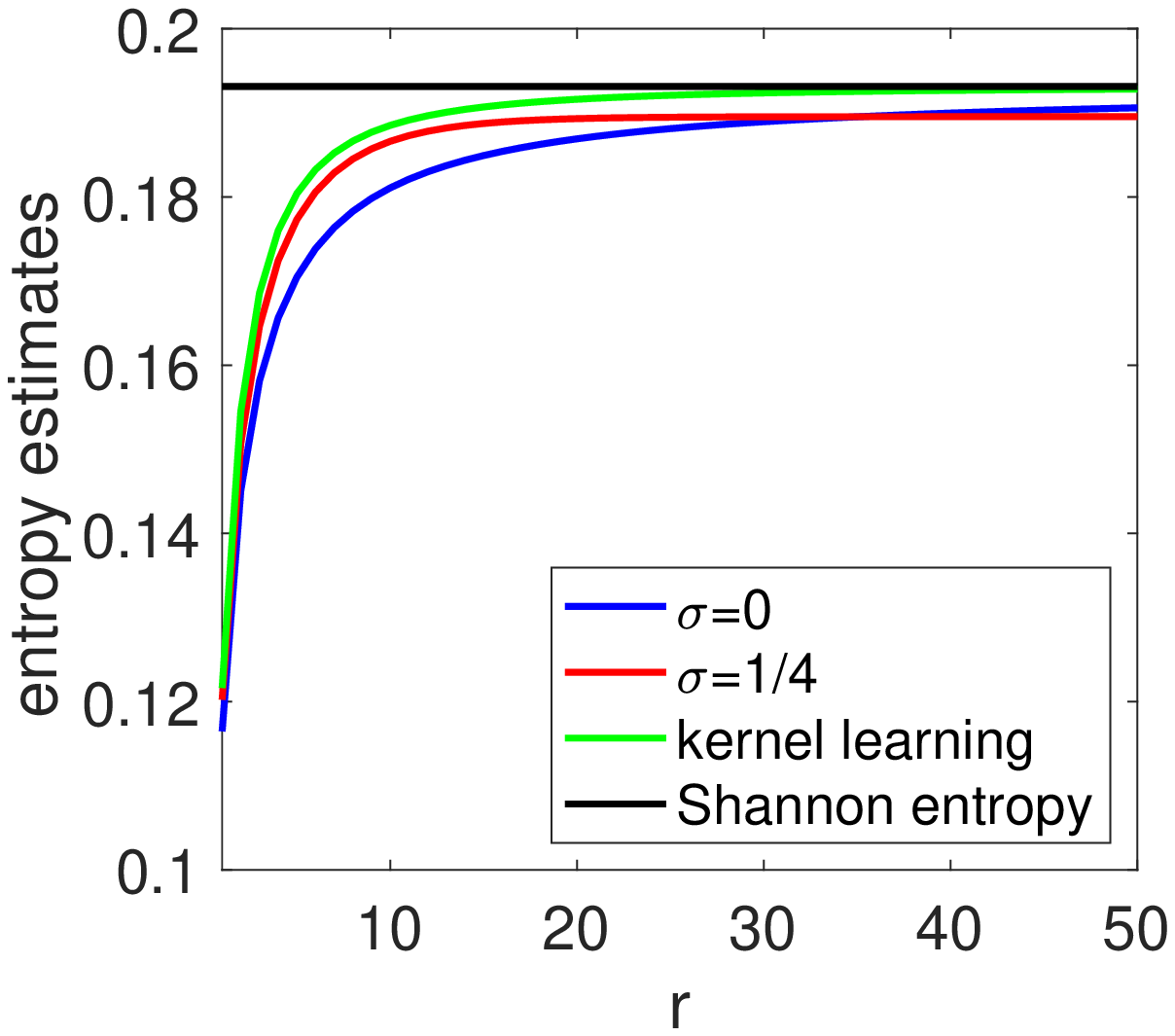} \hspace*{2cm}
\includegraphics[width=5.5cm]{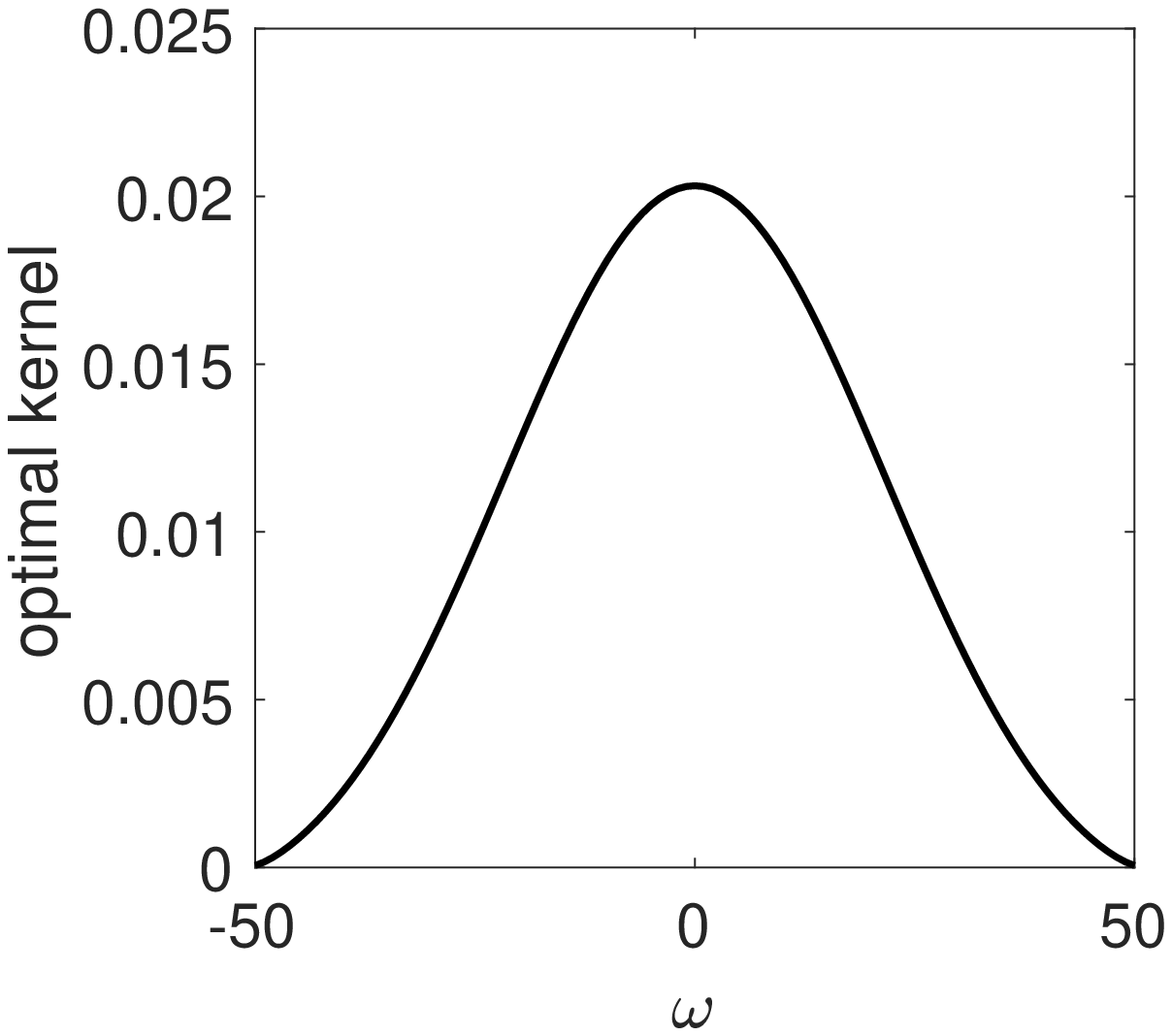}
\end{center}

\vspace*{-.45cm}

\caption{Kernel learning for entropy estimation on the torus $[0,1]$. Left: Entropy estimates for increasing order $r$. Right: optimal kernel for $r=50$.\label{fig:kernel_learning_entropy}}
\end{figure}

\paragraph{Log-partition minimization.} In order to compute (and then minimize) the upper-bound on the log-partition function, we can write
\BEAS
b(f) & = & \sup_{\hat{p}}
\sum_{\delta=-2r}^{2r} \hat{p}(\delta)^\ast \hat{f}(\delta) - 
\tr \big[ C_p^{1/2}  \Diag(\eta)  C_p^{1/2}  \log ( C_p^{1/2}  \Diag(\eta )  C_p^{1/2}    ) \big] +\sum_{\omega=-r}^r  \eta(\omega) \log \eta(\omega)
\\
& = & \inf_M   \ \log \tr \exp ( M  )  + \sum_{\omega=-r}^r \eta(\omega) \log \eta(\omega)
\mbox{ s.~t.}  \sum_{\omega - \omega' = \delta} \eta(\omega)^{1/2} \eta(\omega')^{1/2} M_{\omega \omega'} = \hat{f}(\delta), \ \forall \delta \{-2r,\dots,2r\}.
\EEAS
Note here that by the classical Caratheodory
interpolation theorem (see, e.g.,~\cite[Chapter 2]{foias1990commutant}), this is exactly $a(f)$ which is computed.

It can be minimized by (accelerated) projected gradient descent in $M$, since the function $M \mapsto \log \tr \exp(M)$
is $1$-smooth. We simply need to be able to project on the affine subspace, which can be done by solving a linear system.\footnote{By Lagrangian duality, the minimizer of  $\frac{1}{2}\| M - Z \|_F^2$ such that $\forall \delta \in \{-r,\dots,r\}, \tr(MA_\delta) = \hat{f}(\delta)$ obtained as $Z + \sum_{\delta} \lambda_\delta A_\delta$ with $\lambda$ maximizing $\sum_{\delta=-r}^r \lambda_\delta (\hat{f}_\delta - \tr Z A_\delta) - \frac{1}{2}\sum_{\delta,\delta'=-r}^r \lambda_\delta \lambda_{\delta'}  \tr A_\delta A_{\delta'}$.}

In order to update the kernel $\eta$, we use the exact same update as for maximizing the entropy.
This is illustrated in \myfig{kernel_learning_logpar}, where we show benefit of estimating the kernel rather than using a uniform one.

\begin{figure}
\begin{center}
\includegraphics[width=5.5cm]{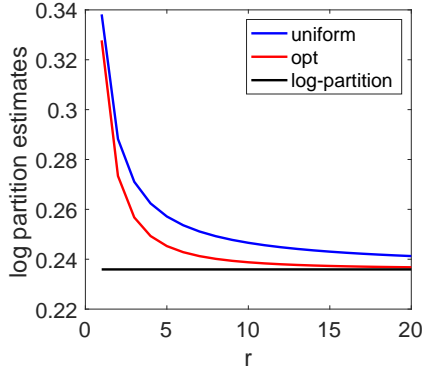}   \end{center}

\vspace*{-.45cm}

\caption{Kernel learning for log partition estimation on the torus $[0,1]$ with increasing order $r$. \label{fig:kernel_learning_logpar}}
\end{figure}

\subsection{Example 2: Hypercube}
\label{sec:exhyp}

We consider $\X = \{-1,1\}^d$, and $\ds \varphi(x) = \Diag(\eta)^{1/2} {x \choose 1} \in \rb^{d+1}$, for $\eta$ in the simplex in $\rb^{d+1}$.
Here $\mathcal{B}$ is exactly the set of all self-adjoint operators with diagonal equal to $\eta$.
To compute the entropy, we consider the set $C$ of positive matrices with unit diagonal, corresponding to $\E \big[  {x \choose 1} {x \choose 1}^\top \big]$.

 To maximize the relative entropy with respect to the uniform distribution, we simply maximize
 $$
 \tr \big[ C(p)^{1/2} \Diag(\eta) C(p)^{1/2} \log(  C(p)^{1/2} \Diag(\eta) C(p)^{1/2} ) \big]
 - \sum_{j=1}^{d+1} \eta_j \log \eta_j,
 $$
 which is a lower bound on $\sum_{x \in \{-1,1\}^d} p(x) \log p(x) + d \log 2$.
 Thus, we get an upperbound on the entropy  
 $$-\!\!\!\!\sum_{x \in \{-1,1\}^d} p(x) \log p(x) 
 \leqslant d \log 2 - 
  \tr \big[  \Diag(\eta) ^{1/2} C(p)\Diag(\eta) ^{1/2}\log(   \Diag(\eta) ^{1/2} C(p)\Diag(\eta) ^{1/2} ) \big]
 +\sum_{j=1}^{d+1} \eta_j \log \eta_j.
 $$
 With $\eta$ constant equal to $1/(d+1)$, we get the upper bound:
 $\ds
 d \log 2 - \frac{1}{d+1} \tr C(p) \log C(p),
 $
 which happens to be tight for $d=1$.

 This is to be compared to the upper-bound of~\cite{jordan2003semidefinite}, which is
 $ \ds \frac{1}{2} \log \det \bigg( C(p) + \frac{1}{3}
 \bigg( \begin{array}{cc}
\idm & 0 \\ 0 & 0  \end{array} \bigg) \bigg)+ \frac{d}{2} \log \frac{\pi e}{2}.
 $
In \myfig{comparison}, we compare these bounds on matrices $C$ corresponding to independent components with mean vectors uniformly distributed, showing the potential benefits of the new upperbound on entropy.

\begin{figure}
\begin{center}
\includegraphics[width=5.5cm]{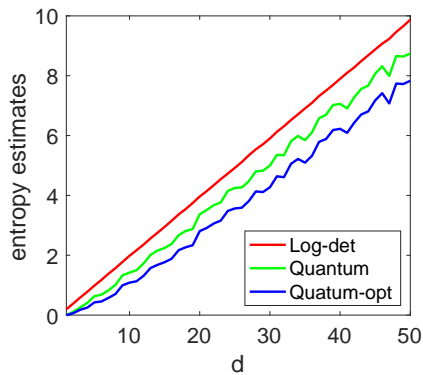} 
\end{center}

\vspace*{-.45cm}

\caption{Comparison of entropy bounds with increasing dimensions (the log-determinant bound from~\cite{jordan2003semidefinite}, our bound based on quantum entropy with uniform $\eta$ and the optimized bound). We plot the bounds for independent distributions with random  probabilities for each component. We substract the true entropy and compare upper-bounds (with zero meaning that the bounds are tight).\label{fig:comparison}}
\end{figure}

\section{Conclusion}
\label{sec:conclusion}

In this paper, we presented properties of the von Neumann entropy and relative entropy of covariance operators obtained from reproducing kernel Hilbert spaces. These notions are intimately related to the usual notions of Shannon entropy and relative entropy,  share many of their properties, and can be applied to any type of data where positive definite kernels can be defined. They also come with additional computational benefits, with several estimation algorithms with theoretical guarantees. We have also highlighted several properties in terms of multivariate probabilistic modeling or variational inference, but several interesting avenues for future research are worth exploring. First, exploring further approximation guarantees for log-partition function estimation can probably lead to new improved algorithms for variational inference~\cite{wainwright2008graphical} in continuous and discrete domains. Second, we could extend to R\'enyi entropies~\cite{muller2013quantum,frank2013monotonicity} with similar developments (along the lines of Appendix~\ref{app:div}). Finally, the optimization with respect to the probability measure, which is at the core of log-partition estimation, can be extended to other optimization problems on probability measures, such as already done with similar tools for optimal transport~\cite{vacher2021dimension} or optimal control~\cite{berthier2021infinite}.

\subsection*{Acknowledgements}
The author thanks Alessandro Rudi and Loucas Pillaud-Vivien for fruitful discussions related to this work.
This work was funded in part by the French government under management of Agence Nationale
de la Recherche as part of the ``Investissements d'avenir'' program, reference ANR-19-P3IA-0001 (PRAIRIE 3IA
Institute). We also acknowledge support the European Research Council (grant SEQUOIA 724063). {The comments and suggestions of the anonymous reviewers were greatly appreciated.}

\appendix

\section{Proof of convexity results on operators}
\label{app:monot}

In this section, we provide simple direct proofs of convexity results that we use in this paper. All of these proofs are obtained from the literature, {and in particular~\cite{lesniewski1999monotone,ruskai2007another,ruskai2002inequalities}.}

\subsection{Joint convexity of the relative entropy}
\label{app:jointcvx}
We consider a function $f: \rb_+^\ast \to \rb$, which is operator-convex, that is, for which the function $A\mapsto f(A)$ is convex as an operator-valued function, i.e., $f \Big(\sum_{i=1}^n \lambda_i A_i\Big) \preccurlyeq \sum_{i=1}^n \lambda_i f(A_i)$. These functions can exactly be represented as:
$ \ds
f(t) = \alpha + \beta (t-1) + \gamma (t-1)^2 + \delta \frac{(t-1)^2}{t} + \int_0^\infty \frac{(t-1)^2}{\lambda + t} d\nu(\lambda),
$
where $\gamma,\delta>0$ and $\nu$ is a positive measure with finite mass on $\rb_+$~\cite{lesniewski1999monotone}. We have $f(1)=\alpha$. The interesting case for this paper is 
$\ds f(t) = - \log(t) = 
1 - t + \int_0^\infty  \frac{ (t-1)^2}{\lambda + t} \frac{1}{(1+\lambda)^2} d\lambda,$
but there are others, such as, for $p \in (0,1]$:
$$
-t^p =  -1 - p(t-1) + \frac{\sin p\pi}{\pi} \int_0^{+\infty}   \frac{(t-1)^2}{\lambda + t} \frac{\lambda^p}{(1+\lambda)^2} d\lambda.
$$

  Given $A$ and $B$ two Hermitian operators, we denote by $\mathcal{R}_A$ the right multiplication by $A$ and $\mathcal{L}_A$ the left multiplication by $A$.
We have, since $ \mathcal{L}_B$ and  $\mathcal{R}_A$ commute, and $\log \mathcal{L}_B = \mathcal{L}_{\log B}$
and $\log \mathcal{R}_A = \mathcal{R}_{\log A}$:
\BEAS
D(A \| B ) & = & \tr A^{1/2} (\log A ) A^{1/2}
- \tr A^{1/2} (\log B ) A^{1/2} = \tr  \big[ (\log \mathcal{R}_A)(A^{1/2}) A^{1/2}\big]
-  \tr  \big[ A^{1/2} (\log \mathcal{L}_B)(A^{1/2})  \big]
\\
& = & \big\langle A^{1/2},   \big[\log \mathcal{R}_A -     \log \mathcal{L}_B    \big] A^{1/2} \big \rangle
=\big\langle A^{1/2}, - \log \big[ \mathcal{L}_B  \mathcal{R}_A^{-1} \big] A^{1/2} \big \rangle, 
\EEAS
where $\langle M, N \rangle = \tr ( M^* N)$ denotes the usual dot-products between operators.
Following~\cite{ruskai2007another}, we thus have:
\BEAS
D(A\|B)
& = & 
 \big\langle A^{1/2}, -\log \big[  \mathcal{L}_B   \mathcal{R}_A^{-1}\big]A^{1/2}  \big \rangle
 \\
 & & +  \int_0^\infty 
  \big\langle  (\mathcal{L}_B   \mathcal{R}_A^{-1} - \idm ) A^{1/2}, 
(  \mathcal{L}_B   \mathcal{R}_A^{-1} + \lambda \idm)^{-1}
 (\mathcal{L}_B   \mathcal{R}_A^{-1} - \idm )  A^{1/2} \big \rangle
  \frac{1}{(1+\lambda)^2} d\lambda \\
  & = & \tr A - \tr B
  +  \int_0^\infty 
  \big\langle  \mathcal{R}_A^{-1/2} (B-A), 
(  \mathcal{L}_B   \mathcal{R}_A^{-1} + t \idm)^{-1}
  \mathcal{R}_A^{-1/2} (B-A) \big \rangle
  \frac{1}{(1+\lambda)^2} d\lambda
\\
  & = & \tr A - \tr B
  +  \int_0^\infty 
  \big\langle    (B-A), 
(  \mathcal{L}_B   + \lambda \mathcal{R}_A)^{-1} 
(B-A) \big \rangle
  \frac{1}{(1+\lambda)^2} d\lambda .
\EEAS

Note we can ensure the finiteness of $D(A\|B)$ as soon as $B^{-1/2} AB^{-1/2}$ is bounded. Given that the mapping $(u,M) \mapsto \langle u, M^{-1} u \rangle$ is convex, we obtain the joint convexity of the relative entropy. It can also be shown more directly as shown below.

  \begin{proposition}
  \label{prop:jointconcv}
  Given positive invertible self-adjoint operators $A_i$, $B_i$,
  then
  $$
  D \Big( \sum_{i=1}^n \lambda_i A_i \big\| \sum_{i=1}^n \lambda_i B_i \Big)
  \leqslant \sum_{i=1}^n \lambda_i D(A_i \| B_i),
  $$
  with equality if and only if for all $j$ such that $\lambda_j > 0$,
  $$
 \log B_j - \log A_j= \log \Big(\sum_{i=1}^n \lambda_i B_i \Big)  - \log  \Big(\sum_{i=1}^n \lambda_i A_i \Big).
  $$
  \end{proposition}
  \begin{proof}
  Following~\cite[Theorem 7]{ruskai2002inequalities}, we have: $D(A_i\| B_i) - \tr A_i ( \log A - \log B) 
  = \tr A_i ( \log A_i - \log A + \log B -  \log B_i ) =  \tr A_i - \tr \exp(  \log A - \log B + \log B_i)
  + B(A_i \| \exp(  \log A - \log B + \log B_i) )$, where $B(\cdot \| \cdot)$ is the Bregman divergence (always non-negative), defined as $B(A\|B) = D(A\|B) - \tr A + \tr B$.
  Thus
  \BEAS
 & &   \sum_{i=1}^n \lambda_i   D(A_i\| B_i) -D(A\|B) \\
   & = & 
  \sum_{i=1}^n \lambda_i \Big[ D(A_i\| B_i) - \tr A_i ( \log A - \log B)  \Big]\\
  & = & \tr A - \sum_{i=1}^n \lambda_i \tr \exp(  \log A - \log B + \log B_i)
  +  \sum_{i=1}^n \lambda_i  B(A_i \| \exp(  \log A - \log B + \log B_i) ),
  \EEAS
which is  greater than zero since the function $C \mapsto \tr \exp(K + \log C)$ is concave. There is equality if and only for all $i$ such that $\lambda_i >0$,
  $
  A_i =  \exp(  \log A - \log B + \log B_i),
  $
  because of the Bregman divergences being equal to zero.
  \end{proof}
  More generally, we get for random operators $A$ and $B$, 
  $\ds
  \E \big[ D(A\|B) \big] \geqslant D( \E A \| \E B),
  $
  with equality if and only if almost surely, $\log A - \log B =  \log \E A - \log \E B$. {Note that the concavity of the function $B \mapsto \tr \exp( M + \log B)$, which is crucial in deriving concentration inequalities for matrices and operators~\cite{tropp2015introduction}, is simply obtained by partial Fenchel dualization of the relative entropy.}

  We can also provide more refined relationships at equality, following~\cite{petz1986sufficient}.
  \begin{proposition}
  Given positive invertible self-adjoint operators $A_i$, $B_i$, and \emph{strictly positive} real numbers $\lambda_i$ that sum to $1$,
  then
  $$
  D \Big( \sum_{i=1}^n \lambda_i A_i \big\| \sum_{i=1}^n \lambda_i B_i \Big)
  = \sum_{i=1}^n \lambda_i D(A_i \| B_i),
  $$
  implies that for all $i$, $A_i B_i^{-1} = A B^{-1}$, where $A =  \sum_{i=1}^n \lambda_i A_i$
  and $B =  \sum_{i=1}^n \lambda_i B_i$.
  \end{proposition}
  \begin{proof}
  Given the representation above, and the fact that  $(u,M) \mapsto \langle u,  M^{-1} u\rangle $ is convex, we get the joint convexity of $(A,B) \mapsto D(A\|B)$ is convex. In order to study the equality sign, we can use the lemma below, and 
  for each $\lambda > 0$, 
  $\big[ \mathcal{L}_{B_i}   + \lambda \mathcal{R}_{A_i}  \big]^{-1} ( B_i - A_i )  
  = \big[ \mathcal{L}_{B}   + \lambda \mathcal{R}_{A}  \big]^{-1} ( B - A )
    $.
    By letting $\lambda$ go to zero, we get the result.
    
    Note that following~\cite{petz1986sufficient}, we can also consider
    that
      $$
  (  \mathcal{L}_{B_i}  + \lambda \mathcal{R}_{A_i} )^{-1} 
  (B_i - A_i)  =  (  \mathcal{L}_{B_i} \mathcal{R}_{A_i}^{-1} + \lambda I )^{-1}  \big[  \mathcal{L}_{B_i}  \mathcal{R}_{A_i}^{-1} - \idm
  \big] \idm \mbox{ is independent of } i,
  $$
  and is equal to $  (  \mathcal{L}_{B} \mathcal{R}_{A}^{-1} + \lambda I )^{-1}  \big[  \mathcal{L}_{B}  \mathcal{R}_{A}^{-1} - \idm
  \big] \idm$. Thus by using representations of a function $f: \rb_+ \to \rb$ as $\ds
  f(t) = (t-1) \int_0^{\infty} \frac{d \nu(\lambda)}{t + \lambda}
  $ for a certain signed measure $d\nu$, we can deduce that  $
  f\big[\mathcal{L}_{B_i}  \mathcal{R}_{A_i}^{-1}\big] \idm \mbox{ is independent of } i,
  $
  and equal to $ f\big[\mathcal{L}_{B}  \mathcal{R}_{A}^{-1}\big] \idm $.   For example. with $f(t) = \log t$, we get $\log B_i - \log A_i$ independent of $i$.
  For example. with $f(t) = t-1$, we get $B_i A_i^{-1} $ independent of $i$. For $f(t) = t^z$,
  $B_i^{z} A_i^{-z}$ independent of $i$, for any complex $z$.
  \end{proof}

  \begin{lemma}
  If $M_1,\dots,M_n$ are invertible and all $\lambda_i$ are strictly positive and sum to $1$, then
  $$
  \sum_{i=1}^n \lambda_i u_i^\top M_i^{-1} u_i
  = \Big( \sum_{i=1}^n \lambda_i u_i \Big)
  \Big( \sum_{i=1}^n \lambda_i M_i \Big)^{-1}
  \Big( \sum_{i=1}^n \lambda_i u_i \Big)
  $$
  if and only if $M_i^{-1} u_i$ is independent of $i$, and equal to $M^{-1}u$ where $M= \sum_{i=1}^n \lambda_i M_i  $ and $u =   \sum_{i=1}^n \lambda_i u_i $.
  \end{lemma}
  \begin{proof}
 We consider the function $f(u,M) = u^\top M^{-1} u$, with Taylor expansion
 $$
 f(u+\delta,M+\Delta) = f(u,M) + 2 \delta^\top M^{-1}u  - u^\top M^{-1} \Delta M^{-1} u
 + 2 (\delta - \Delta M^{-1}u)^\top M^{-1}  (\delta - \Delta M^{-1}u) + O( \| \Delta\|^3 + \| \delta\|^3).
 $$
We now use the following result on the remainder in Jensen's inequality (which itself is a direct consequence of the Taylor expansion with integral remainder taken at $X$, while expanded at $\E X$):
\BEQ
\label{eq:jensen}
\E f(X) - f(\E X) = \int_0^1 \E \Big[
(X - \E X)^\top f''( tX + (1-t) \E X) ( X - \E X) 
\Big] dt.
\EEQ
{In our particular situation, by the invertibility assumption, the matrix $M$ is  invertible almost surely, and having equality in Jensen's inequality imposes to have $(X - \E X)^\top f''( tX + (1-t) \E X) ( X - \E X) = 0 $ almost surely in $t$ and $X$. Thus, at $t=0$ and $X = (u_i,M_i)$, we get: $u_i - u - (M_i - M) M^{-1} u = 0$, which is exactly 
 $M_i^{-1} u_i = M^{-1} u$.}
\end{proof}

\subsection{Monotonicity of the relative entropy}
\label{app:monotonicity}
  \begin{proposition}\label{prop:monot}
  Given operators $C_i$,  $i=1,\dots,n$, such that $\sum_{i=1}^n C_i^\ast C_i = \idm$, then
  $$
  D \Big( \sum_{i=1}^n C_i A C_i^\ast \Big\|
  \sum_{i=1}^n C_i B C_i^\ast \Big)  \leqslant D(A\|B).
  $$
  \end{proposition}
  \begin{proof}
  We have, with $\omega = \exp( 2i\pi/n)$, using the joint convexity of the relative entropy and the expression $ \frac{1}{n} \sum_{k=1}^n \omega^{k(i-j)}= 1_{i = j}$:
  \BEAS
    D \Big( \sum_{i=1}^n C_i A C_i^\ast \Big\|
  \sum_{i=1}^n C_i B C_i^\ast \Big)  
  & \leqslant &   \sum_{i=1}^n  D ( C_i A C_i^\ast  \|
  C_i B C_i^\ast )  
=   D \Big( \sum_{i=1}^n (e_i e_i^\top ) \otimes E_i A E_i^\ast \Big\|
 \sum_{i=1}^n (e_i e_i^\top ) \otimes C_i B C_i^\ast \Big)  
\\
  & = &    D \Big( \frac{1}{n} \sum_{k=1}^n \omega^{k(i-j)}\sum_{i,j=1}^n (e_i e_j^\top ) \otimes C_i A C_j^\ast \Big\|
 \frac{1}{n} \sum_{k=1}^n \omega^{k(i-j)}\sum_{i,j=1}^n (e_i e_j^\top ) \otimes C_i B C_j^\ast \Big)  .
\EEAS
Consider the matrices $M$ and $N$ defined by blocks $M_{ij} = C_i A C_j^\ast $ and 
$N_{ij} = C_i B C_j^\ast $, as well as, the matrix $U_k$ defined by blocks $(U_k)_{ij} = 1_{i=j} \omega^{ki} I$, we get from the equation above:
 \BEAS
    D \Big( \sum_{i=1}^n C_i A C_i^\ast \Big\|
  \sum_{i=1}^n C_i B C_i^\ast \Big)  
  & \leqslant &   D \Big( \frac{1}{n} \sum_{k=1}^n U_k M U_k^\ast \Big\|
 \frac{1}{n} \sum_{k=1}^n U_k N U_k^\ast \Big)  
\\
  & \leqslant  &  \frac{1}{n} \sum_{k=1}^n  D \Big(  U_k M U_k^\ast  \Big\|
 U_k M U_k^\ast \Big)   = D(M \| N) 
  \EEAS
  because each $U_k$ is unitary. Moreover, we can write $M = V A V^*$ and $N = V A V^\ast$, with a an operator $V$ defined by a column of blocks $C_i$, $i=1,\dots,n$, so that $V^* V = \sum_{i=1}^n C_i^\ast C_i = \idm$. Thus $D(M \| N)  = D(A\|B)$, which shows the desired bound.
  
  Using propositions above, there is is equality if and only:
$ \log (C_i B C_i^\ast)  - \log( C_i A C_i^\ast) $  independent of $i$ and   $ \log  [  U_k N U_k^\ast \  ]  - \log  [ U_k M U_k^\ast  ]
  =  U_k  (\log N - \log M)  \big)U_k^\ast $ independent of $k$, that is $\log N - \log M$ block diagonal. Overall, we get equality if $\log N - \log M$ is block diagonal with equal blocks.
  
We can also obtain necessary conditions for equality, that is,
 $ (C_i B C_i^\ast) ( C_i A C_i^\ast)^{-1}$  independent of $i$, and  $ U_k N U_k^\ast    [   U_k N U_k^\ast]^{-1}
  = U_k N M^{-1} U_k^\ast$ independent of $k$, that is,
$ NM^{-1}$ block-diagonal with equal blocks.

  \end{proof}
  This extends to a random operator $C$, such that $\E \big[ C^\ast C \big] = \idm$, for which we get:
  $$
  D\Big( \E \big[ C AC^\ast \big] \, \big\|  \, \E \big[ C BC^\ast \big] \Big) \leqslant D(A\| B).
  $$
  A particular example is obtained from $n$ Hermitian positive operators $D_1,\dots,D_n$ such that $\sum_{i=1}^n D_i = \idm$, and we get, with $\mu_i = \tr (D_iA)$ and $\nu_i = \tr (D_iB)$:
  $$
  \sum_{i=1}^n \mu_i \log \frac{\mu_i}{\nu_i} \leqslant D(A\| B).
  $$
  This corresponds to, if $D_i = \sum_{k} \lambda_{ik} u_{ik} \otimes u_{ik}$, to 
  $C_{ik} =\lambda_{ik}^{1/2}   e_i u_{ik}^\ast$, 
  so that 
  $$
  \sum_{i,k} E_i A E_i^\ast = 
 \sum_{i,k} \lambda_{ik}   e_i u_{ik}^\ast M u_{ik} e_i^\ast
 = \sum_{i=1}^n ( \tr A D_i) e_i e_i^\ast.
  $$
 
 \subsection{Concavity of relative entropy}
 \label{app:concav}
 In this section, we consider two positive self-adjoint operators, $M$ and $N$ such that $M \preccurlyeq \alpha N$ for some $\alpha$. We consider the function:
 $$
 f: \Lambda \to D( \Lambda^{1/2} M  \Lambda^{1/2} \|  \Lambda^{1/2} N \Lambda^{1/2}).
 $$
 It happens to be convex in $(M,N)$. It is also concave in $\Lambda$ in the set of positive bounded operators. This is due to the representation, valid for invertible operators:
 \BEAS
& &   f(\Lambda) - \tr M\Lambda + \tr N\Lambda \\
  & = &  
     \int_0^\infty 
  \big\langle   \Lambda^{1/2}   (N-M)  \Lambda^{1/2} , 
(  \mathcal{L}_{  \Lambda^{1/2} N \Lambda^{1/2}}   + \lambda \mathcal{R}_{ \Lambda^{1/2} M  \Lambda^{1/2}})^{-1} 
 \Lambda^{1/2}   (N-M)  \Lambda^{1/2} \big \rangle
  \frac{1}{(1+\lambda)^2} d\lambda  \\
  & = & 
   \int_0^\infty 
  \big\langle      (N-M)    , 
(  \mathcal{L}_{  N  }  \mathcal{R}_{\Lambda}^{-1}   + \lambda   \mathcal{L}_{\Lambda}^{-1}  \mathcal{R}_{ M} )^{-1} 
  (N-M)  \big \rangle
  \frac{1}{(1+\lambda)^2} d\lambda.
 \EEAS
 Since the mapping $(A,B) \mapsto (A^{-1}+ B^{-1})^{-1} = B - B(A+B)^{-1} B$ is matrix-concave~\cite[Corollary 1.1]{ando1979concavity}, we get the desired concavity.

 \subsection{Extension to other divergences}
 \label{app:div}
 
 As mentioned in~\cite{ruskai2007another}, both the joint convexity and the monotonicity of the relative entropy $D(A\|B)$ in Prop.~\ref{prop:jointconcv} and Prop.~\ref{prop:monot} can be extended to most operator-convex functions on $\rb_+^\ast$, and 
 $$
 D_f(A \| B) =  \big\langle A^{1/2}, f \big[  \mathcal{L}_B   \mathcal{R}_A^{-1}\big]A^{1/2}  \big \rangle,
 $$
 with the relative entropy being $D_{\rm - \log}$. Among other interesting cases, we have:
\BIT
\item $\ds f(t) = ( t - 1)^2$, with $D_f(A\| B) = \tr \big[ (B-A)^2 A^{-1} \big]$.
\item $\ds f(t) = \frac{(t-1)^2}{t}$, with $D_f(A\| B) = \tr \big[ (B-A)^2 B^{-1} \big]$.
\item $\ds f(t) = \frac{(t-1)^2}{1+t} $, with $D_f(A\|B) = \langle A^{1/2}, ( \mathcal{L}_B + \mathcal{R}_A)^{-1}  A^{1/2} \rangle$.
\item $\ds f(t) =  1-\sqrt{t}$, with $D_f(A\|B) = \tr (A)  - \big\langle A^{1/2},  [ \mathcal{L}_B  \mathcal{R}_A^{-1} \big]^{1/2} A^{1/2} \big \rangle$.
\EIT   
Specialized to covariance operators, we have, with $A=\Sigma_p$ and $B=\Sigma_q$, and following the same reasoning as in the proof of Prop.~\ref{prop:propKL}:
\BEAS
D_f(\Sigma_p \| \Sigma_q)
& \leqslant & \int_\X D_f \Big( \varphi(x)\varphi(x)^\ast \Big\| \frac{dq}{dp}(x)  \varphi(x)\varphi(x)^\ast \Big) dp(x) \\
& = &  \int_\X f \Big(  \frac{dq}{dp}(x)  \Big) dp(x) = D_f(p\|q),
\EEAS
which is the \emph{reverse} $f$-divergence associated with the function $f$~\cite{csiszar1967information,topsoe2000some}. We thus immediately get lower bounds on:
\BIT
\item the square Hellinger distance $\ds D_f(p\|q) = \int_\X \Big[ 1  - \Big( \frac{dq}{dp}(x) \Big)^{1/2} \Big] dp(x) 
= \frac{1}{2} \int_\X \Big[   \Big( \frac{dq}{d\tau}(x) \Big)^{1/2} -  \Big( \frac{dp}{d\tau}(x) \Big)^{1/2} \Big]^2 d\tau(x)  $ for $f(t) = 1-\sqrt{t}$, 
\item the reverse Pearson $\chi^{2}$-divergence $\ds D_f(p\|q) = \int_\X \Big[ 1  -  \frac{dq}{dp}(x)   \Big]^2 dp(x) 
= \frac{1}{2} \int_\X \Big[   \frac{dq}{d\tau}(x)   -  \frac{dp}{d\tau}(x)   \Big]^2 \frac{d\tau(x)}{ dp/d\tau(x)}  $ for $f(t) = (t-1)^2$, 

\item  the  Pearson $\chi^{2}$-divergence $\ds D_f(p\|q) = \int_\X \Big[ 1  -  \frac{dq}{dp}(x)   \Big]^2 \frac{dp(x) }{dq/dp(x)}
= \frac{1}{2} \int_\X \Big[   \frac{dq}{d\tau}(x)   -  \frac{dp}{d\tau}(x)   \Big]^2 \frac{d\tau(x)}{ dq/d\tau(x)}  $ for $\ds f(t) = \frac{(t-1)^2}{t}$, 

\item the Vincze-Le Cam distance  $\ds D_f(p\|q) = \int_\X \Big[ 1  -  \frac{dq}{dp}(x)   \Big]^2 \frac{dp(x) }{1+ dq/dp(x)}$\\  equal to $\ds \frac{1}{2} \int_\X \Big[   \frac{dq}{d\tau}(x)   -  \frac{dp}{d\tau}(x)   \Big]^2 \frac{d\tau(x)}{ dq/d\tau(x)
+  dp/d\tau(x) }  $ for $\ds f(t) = \frac{(t-1)^2}{1+t}$.

\EIT

  \section{Proof of lower bound (Prop.~\ref{prop:asymptotics})}
  \label{app:asymptotics}
  
  \begin{proof}
  Starting from $ \tilde{p}(y) =   \int_\X h(x,y) dp(x)$ and $ \tilde{q}(y)  = \int_\X h(x,y) dq(x)$, we can provide an exact expression in the usual data processing inequality for the regular KL divergence~\cite{cover1999elements}, with two joint distributions $\bar{p}(x,y) = p(x) h(x,y)$ and $\bar{q}(x,y) = q(x) h(x,y) $, so that, with $f(u) =u \log u$:
  \BEAS
  D( p \| q ) & = &  D(\bar{p} \| \bar{q}) = \int_{\X} \int_\X \bar{q}(x,y) f \Big(
  \frac{\bar{p}(x,y)}{\bar{q}(x,y)}
  \Big)d\tau(x)d\tau(y) \\
  & \geqslant & \int_\X \tilde{q}(y) f\Big(
  \int_\X
  \bar{q}(x|y) \frac{\bar{p}(x,y)}{\bar{q}(x,y)}
  d\tau(x)
  \Big) d\tau(y)  \mbox{ by Jensen's inequality},
  \\
   & = & \int_\X \tilde{q}(y) f\Big(
  \ \frac{\tilde{p}(y)}{\tilde{q}(y)}
  \Big) d\tau(y) = D(\tilde{p} \| \tilde{q}).
  \EEAS
  Using \eq{jensen}, we thus get
    \BEAS
  D(p\|q) - D(\tilde{p} \| \tilde{q} ) 
   & = & \int_\X\int_\X \tilde{q}(y) h(x,y) \Big(
  \frac{p(x)}{q(x)}-  \frac{\tilde{p}(y)}{\tilde{q}(y)}
  \Big)^2 \Big( \int_0^1 \frac{dt}{t  \frac{p(x)}{q(x)}+(1-t) \frac{\tilde{p}(y)}{\tilde{q}(y)}} \Big)  d\tau(x) d\tau(y).
  \EEAS
  Since for all $x \in \X$, $t  \frac{p(x)}{q(x)}+(1-t) \frac{\tilde{p}(y)}{\tilde{q}(y)} \geqslant \inf_{x \in \X} \frac{p(x)}{q(x)}$, we can then bound as
  \BEAS
    D(p\|q) - D(\tilde{p} \| \tilde{q} ) 
  & \leqslant & \sup_{x \in \X} \frac{q(x)}{p(x)} \int_\X \int_\Y \tilde{q}(y) h(x,y) \Big(
  \frac{p(x)}{q(x)}-  \frac{\tilde{p}(y)}{\tilde{q}(y)}
  \Big)^2 d\tau(x) d\tau(y)
  \\
    & = & \sup_{x \in \X} \frac{q(x)}{p(x)} \int_\X\int_\X \tilde{q}(y) h(x,y) \Big(
  \frac{p(x)}{q(x)}-  \frac{p(y)}{q(y)} + \frac{p(y)}{q(y)}  - \frac{\tilde{p}(y)}{\tilde{q}(y)}
  \Big)^2 d\tau(x) d\tau(y)\\
    & \leqslant & 2 \sup_{x \in \X} \frac{q(x)}{p(x)} \int_\X\int_\X \tilde{q}(y) h(x,y) \Big(
  \frac{p(x)}{q(x)}-  \frac{p(y)}{q(y)}   \Big)^2 d\tau(x) d\tau(y) \\
  & & 
  +  2\sup_{x \in \X} \frac{q(x)}{p(x)} \int_\X \tilde{q}(y)  \Big(
\frac{p(y)}{q(y)}  - \frac{\tilde{p}(y)}{\tilde{q}(y)}\Big)^2  d\tau(y).
  \EEAS
  We assume Lipschitz-continuity for $p$ and $q$ in the following form:  for all $x,y$, $\big| \frac{q(x)}{q(y)} - 1 \big| \leqslant d(x,y) C_q$, 
   $\big| \frac{p(x)}{p(y)} - 1 \big| \leqslant d(x,y) C_p$, and
    $\big| \frac{p(y)}{q(y)} -  \frac{p(x)}{q(x)}  \big| \leqslant d(x,y) C_{p/q}$ for some distance $d$ on $\X$. Then the first term in the bound above is less than 
$\ds    2 C_{p/q}^2 
       \sup_{x \in \X} \frac{q(x)}{p(x)} \sup_{x \in \X} \int_\X h(x,y) d(x,y)^2 d\tau(y)  $.
    For the second term, we write:
    \BEAS
    \Big|
\frac{p(y)}{q(y)}  - \frac{\tilde{p}(y)}{\tilde{q}(y)}\Big|
& \leqslant & \frac{p(y) | \tilde{q}(y) - q(y)| +  q(y) | \tilde{p}(y) - p(y)|}{q(y) \tilde{q}(y)}
\\
& \leqslant & \frac{p(y) \int_\X | q(y) - q(y)| h(x,y) d\tau(x) +  q(y) \int_\X | p(y) - p(y)| h(x,y) d\tau(x)}{q(y) \tilde{q}(y)}
\\
& \leqslant & \frac{C_q p(y) q(y) \int_\X d(x,y) h(x,y) d\tau(x) +  C_p q(y) p(y)  \int_\X d(x,y) h(x,y) d\tau(x)}{q(y) \tilde{q}(y)}
\\
& \leqslant &   \frac{(C_p + C_q) p(y) }{\tilde{q}(y)} \int_\X d(x,y) h(x,y) d\tau(x) \leqslant 
\frac{(C_p + C_q) p(y) }{\tilde{q}(y)} \bigg( \int_\X d(x,y)^2 h(x,y) d\tau(x) \bigg)^{1/2}.
    \EEAS
    Moreover
    \BEAS
    \int_\X \frac{p(y)^2}{\tilde{q}(y)} d\tau(y)
& \leqslant &    \sup_{x \in \X} \frac{p(x)}{q(x)}   \int_\X \frac{2 \tilde{p}(y)^2 + 2 ( \tilde{p}(y) - p(y))^2}{\tilde{p}(y)} d\tau(y) 
\\
& \leqslant  &  2 \sup_{x \in \X} \frac{p(x)}{q(x)} \Big( 1 +   \int_\X \frac{ ( \tilde{p}(y) - p(y))^2}{\tilde{p}(y)} d\tau(y)  \Big)
\leqslant 2 \sup_{x \in \X} \frac{p(x)}{q(x)} (1 + C_p^2 {\rm diam}(\X)^2 ).    \EEAS
    Overall, we get
    $\ds 
       D(p\|q) - D(\tilde{p} \| \tilde{q} ) 
         \leqslant   E(p,q) \times  
        \sup_{x \in \X} \int h(x,y) d(x,y)^2 dy,         $       with $E(p,q) =  2 C_{p/q}^2 
       \sup_{x \in \X} \frac{q(x)}{p(x)} +  4 \sup_{x \in \X} \frac{q(x)}{p(x)} ( C_p + C_q)^2  
       \sup_{x \in \X} \frac{p(x)}{q(x)} (1 + C_p^2 {\rm diam}(\X)^2 )$.
        \end{proof}

  \section{Convergence of estimation algorithms}
  
   In this Appendix, we give convergence proofs for the estimation algorithms from \mysec{estimation}, starting with concentration inequalities for operators and bounds on degrees of freedom.

   \subsection{Concentration of sums of random self-adjoint operators}
  We recall in this Appendix classical results on the concentation of covariance operators.
  \begin{lemma}[Concentration of covariance operators~\cite{rudi2017generalization,tropp2015introduction}]
  \label{lemma:conc}
Let $\varphi$ be mapping $\varphi: \X \to \H$ between a measurable space $\X$ and a separable Hilbert space $\H$, such that for all $x \in \X$, $\| \varphi(x) \| \leqslant R$. Given $\Sigma_p = \ds \int_\X \varphi(x)  \varphi(x) ^\ast dp(x) $, assume that $\lambda_{\max}(\Sigma_p) \leqslant \mu$. Given $x_1,\dots,x_n \in \X$ $n$ i.i.d.~samples from $p$, then, for all $u>0$,
$$
\P \bigg( \lambda_{\max} \Big(
\Sigma_p - \frac{1}{n} \sum_{i=1}^n
\varphi(x_i) \otimes \varphi(x_i) \Big) > u
\bigg) \leqslant \frac{\tr \Sigma_p}{\mu} \Big( 1 + \frac{3\mu^2}{u^4 n^2} \big( R^2 + \frac{u}{3} \big)^2 \Big)
\exp \Big( - \frac{ nu^2}{ 2 \mu ( R^2 + \frac{u}{3}  )} \Big).
$$
\end{lemma}
We will apply the lemma above to $\tilde{\varphi}(x) = ( \Sigma_p  + \lambda \idm)^{-1/2} \varphi(x)$, with $\ds \Sigma_p  = 
\int_\X \varphi(x) \otimes \varphi(x) dp(x)$
and $\ds \hat{\Sigma}  = \frac{1}{n} \sum_{i=1}^n
\varphi(x_i) \otimes \varphi(x_i) $, so that,
$\ds R^2   = \sup_{x \in \X}\  \langle \varphi(x),  ( \Sigma_p  + \lambda \idm)^{-1} \varphi(x) \rangle = {\rm df}^{\max}(\lambda)$, and $\tr \Sigma_p$ is now $\tr \big[ \Sigma_p ( \Sigma_p + \lambda \idm)^{-1} \big]
= {\rm df}(\lambda) $, and $\mu = 1$, leading to:
$$
 \P \Big(  \lambda_{\max}\Big(
  (\Sigma_p+\lambda \idm   )^{-1/2} ( \Sigma_p - \hat{\Sigma}_p)  (\Sigma_p+\lambda \idm  )^{-1/2}  \Big)> u \Big)   \leqslant
  {\rm df}(\lambda) \Big(1 + \frac{3(  {\rm df}^{\max}(\lambda) + u/3)^2}{u^4 n^2}  \Big)
  \exp\Big( \frac{ - n u^2 }{2  {\rm df}^{\max}(\lambda) + 2u/3} \Big).
$$
   
  \subsection{Degrees of freedom estimation}

\begin{proposition}[Estimation of degrees of freedom]
\label{prop:df}
Let $\varphi: \X \to \H$ a mapping between a measurable space $\X$ and a separable Hilbert space $\H$. Let $\Sigma_p = \ds \int_\X \varphi(x)   \varphi(x)^\ast dp(x) $, with $p$ a probability distribution and $x_1,\dots,x_n \in \X$ $n$ i.i.d.~samples from $p$, and $\hat{\Sigma}_p = \frac{1}{n} \sum_{i=1}^n
\varphi(x_i) \otimes \varphi(x_i)$. Then,
$$
\E \big| \tr [ \Sigma_p ( \Sigma_p + \lambda \idm)^{-1} ] - 
\tr [ \hat\Sigma_p ( \hat\Sigma_p + \lambda \idm)^{-1} ] \big| \leqslant
 \frac{d}{\sqrt{n}} + 4
  \frac{d^2}{n}
  + 16 d n  
  \exp\Big( \frac{ - 9 n  }{32 d + 8} \Big) ,
$$
where $ {\rm df}(\lambda) 
= \tr \big[ \Sigma_p ( \Sigma_p + \lambda \idm)^{-1} \big]$, and
${\rm df}^{\max}(\lambda) = \sup_{x \in \X}\  \langle \varphi(x),  ( \Sigma_p  + \lambda \idm)^{-1} \varphi(x) \rangle$, which is assumed to satisfy $d = {\rm df}^{\max}(\lambda)  \leqslant n$.
\end{proposition}
\begin{proof}
We follow closely the proof of~\cite{rudi2017generalization}, which starts with:
\BEAS
A
 & = &  \tr [ \Sigma_p ( \Sigma_p + \lambda \idm)^{-1} ] - 
\tr [ \hat\Sigma_p ( \hat\Sigma_p + \lambda \idm)^{-1} ]  \\
& = &  \lambda \tr \big[  (\Sigma_p+\lambda \idm  )^{-1/2} B (\Sigma_p+\lambda \idm  )^{-1/2} \big]
 +   \lambda \tr \big[  (\Sigma_p+\lambda \idm  )^{-1/2} B ( \idm - B)^{-1} B (\Sigma_p+\lambda \idm  )^{-1/2}\big] ,
 \EEAS
 leading to
$\ds
|A|
   \leqslant     \lambda   \big| \tr \big[  (\Sigma_p+\lambda \idm  )^{-1/2} B (\Sigma_p+\lambda \idm  )^{-1/2}\big]  \big| + 
     \| B \|_{ \rm HS}^2 \frac{1}{1-\lambda_{\max}(B)}$,
as soon as $\lambda_{\max}(B)<1$, for $B =  (\Sigma_p +\lambda \idm  )^{-1/2} ( \Sigma_p  - \hat{\Sigma}_p )  (\Sigma_p+\lambda \idm  )^{-1/2}$.

Almost surely, $|A|  
\leqslant \max\{ {\rm df}(\Lambda), n \}$.
Thus,  for $u<1$,
\BEAS
|A| & \leqslant & 
 1_{ \lambda_{\max}(B) > u} \max\{ {\rm df}(\Lambda), n\} + 
  1_{ \lambda_{\max}(B) \leqslant u} \Big(   \lambda   \big| \tr \big[  (\Sigma_p+\lambda \idm  )^{-1/2} B (\Sigma_p+\lambda \idm  )^{-1/2}\big]  \big| +   \| B \|_{ \rm HS}^2 \frac{1}{1-\lambda_{\max}(B)}\Big)  \\
  & \leqslant & 
 1_{ \lambda_{\max}(B) > u}\max\{ {\rm df}(\Lambda), n\} + 
 \| B \|_{ \rm HS}^2 \frac{1}{1-u}  +    \lambda   \big| \tr \big[  (\Sigma_p+\lambda \idm  )^{-1/2} B (\Sigma_p+\lambda \idm  )^{-1/2}\big]  \big|. \EEAS
Thus,
\BEAS
\E |A|
& \leqslant &  \P (  \lambda_{\max}(B) > u )  \cdot  \max\{ {\rm df}(\Lambda), n\}
+ \frac{1}{1-u} \E \big[  \| B \|_{ \rm HS}^2  \big]
+   \E \big[    \lambda   \big| \tr \big[  (\Sigma_p+\lambda \idm  )^{-1/2} B (\Sigma_p+\lambda \idm  )^{-1/2}\big]  \big|\big].
\EEAS
We have:
$$
 \P (  \lambda_{\max}(B) > u )   \leqslant
  {\rm df}(\lambda) \Big(1 + \frac{3(  {\rm df}^{\max}(\lambda) + u/3)^2}{u^4 n^2}  \Big)
  \exp\Big( \frac{ - n u^2 }{2  {\rm df}^{\max}(\lambda) + 2u/3} \Big),
$$
and 
\BEAS
\E \big[  \| B \|_{ \rm HS}^2 \big]& = & \frac{1}{n}
\E \Big[ \tr \Big( (\Sigma_p +\lambda \idm  )^{-1}   ( \Sigma_p - \varphi(x_i) \otimes \varphi(x_i))  (\Sigma +\lambda \idm  )^{-1} 
( \Sigma_p - \varphi(x_i) \otimes \varphi(x_i))
\Big) \Big]
\\
& \leqslant & \frac{1}{n}
\E \Big[ \tr \Big( (\Sigma_p +\lambda \idm  )^{-1}    \varphi(x_i) \otimes \varphi(x_i)  (\Sigma_p +\lambda \idm  )^{-1} 
  \varphi(x_i) \otimes \varphi(x_i)
\Big) \Big]
 \leqslant  \frac{1}{n}  {\rm df}(\lambda)    {\rm df}^{\max}(\lambda).
\EEAS

Moreover 
\BEAS
   \E \big[   \lambda   \big| \tr \big[  (\Sigma_p+\lambda \idm  )^{-1/2} B (\Sigma_p+\lambda \idm  )^{-1/2}\big]  \big| \big]
 & \leqslant &  \lambda \sqrt{ \E \big[   \big| \tr \big[   (\Sigma_p+\lambda \idm)^{-1/2} B (\Sigma_p+\lambda \idm)^{-1/2}  \big] \big|^2 \big]} \\
 & \leqslant &  \lambda \sqrt{ \frac{1}{n} \E \big[    \langle \varphi(x), (\Sigma_p+\lambda \idm  )^{-1 }    (\Sigma_p+\lambda \idm )^{-1 } \varphi(x) \rangle  ^2 \big]} \\
 & \leqslant &   \sqrt{ \frac{1}{n} \E \big[    \langle \varphi(x), (\Sigma_p+  \lambda \idm)^{-1 } \varphi(x) \rangle  ^2 \big]}  
  \leqslant     \sqrt{ \frac{1}{n}  {\rm df}^{\max}(\lambda) {\rm df}(\lambda)}.
\EEAS

 This leads to
 \BEAS
  \E |A| &  \leqslant  & 
   \sqrt{ \frac{1}{n}  {\rm df}^{\max}(\lambda) {\rm df}(\lambda)}
 +   \frac{1}{1-u}   \frac{1}{n}  {\rm df}(\lambda)    {\rm df}^{\max}(\lambda)
  \\
  & &  \hspace*{2cm} + \max\{ {\rm df}(\lambda), n \}  {\rm df}(\lambda) 
   \Big(1 + \frac{3(  {\rm df}^{\max}(\lambda) + u/3)^2}{u^4 n^2}  \Big)
  \exp\Big( \frac{ - n u^2 }{2  {\rm df}^{\max}(\lambda) + 2u/3} \Big) \Big].
 \EEAS
With $u = 3/4$, we get:
\BEAS
 \E |A| &  \leqslant  & 
   \sqrt{ \frac{1}{n}  {\rm df}^{\max}(\lambda) {\rm df}(\lambda)}
 +  4  \frac{1}{n}  {\rm df}(\lambda)    {\rm df}^{\max}(\lambda)
   \\
   && \hspace*{2cm}  + \max\{ {\rm df}(\lambda), n \}  {\rm df}(\lambda)   \Big(1 + \frac{3(  {\rm df}^{\max}(\lambda) + 1/4)^2}{(3/4)^4 n^2}  \Big)
  \exp\Big( \frac{ - n (3/4)^2 }{2  {\rm df}^{\max}(\lambda) + 1/2} \Big) \Big].
 \EEAS
 Using $ {\rm df}(\lambda)  \leqslant  {\rm df}^{\max}(\lambda) = d $, we have the bound:
\BEAS
\E|A|
&  \leqslant  & 
 \frac{d}{\sqrt{n}} + 4
  \frac{d^2}{n}
  + d \max\{d, n \}    \Big(1 + \frac{16 (  4 d + 1)^2}{27 n^2}  \Big)
  \exp\Big( \frac{ - 9 n  }{8(4 d + 1)} \Big) .
 \EEAS
If $d \leqslant n$, we further get:
\BEAS
 \E |A| 
 &  \leqslant  & 
 \frac{d}{\sqrt{n}} + 4
  \frac{d^2}{n}
  + d n  \Big(1 + \frac{16 (  4 d + 1)^2}{27 n^2}  \Big)
  \exp\Big( \frac{ - 9 n  }{32 d + 8} \Big) 
\leqslant 
 \frac{d}{\sqrt{n}} + 4
  \frac{d^2}{n}
  + 16 d n  
  \exp\Big( \frac{ - 9 n  }{32 d + 8} \Big) 
.
 \EEAS
\end{proof}

\subsection{Degrees of freedom estimation (from projections)}

\begin{proposition}[Estimation of degrees of freedom (from projections)]
Let   $\varphi: \X \to \H$ a mapping between a measurable space $\X$ and a separable Hilbert space $\H$.
Let $\Sigma_p = \ds \int_\X \varphi(x)   \varphi(x)^\ast dp(x) $, with $p$ a probability distribution and $x_1,\dots,x_n \in \X$ $n$ i.i.d.~samples from $q$ such that $\forall x \in \X,   q(x) \geqslant \alpha p(x) $, and $\hat{\Pi}$ the orthogonal projection on the span of all $\varphi(x_1),\dots,\varphi(x_n)$.
Let  ${\rm df}^{\max}(\lambda) = \sup_{x \in \X}\  \langle \varphi(x),  ( \Sigma_q  + \lambda \idm)^{-1} \varphi(x) \rangle$.
Then
$$ \E \big| \tr [ \Sigma_p ( \Sigma_p + \lambda \idm)^{-1} ] - 
\tr [ \hat{\Pi}\Sigma_p \hat{\Pi}( \hat{\Pi}\Sigma_p\hat{\Pi} + \lambda \idm)^{-1} ] \big|
\hspace*{7cm}$$
$$ \hspace*{4cm} \leqslant
16  {\rm df}(\alpha \lambda){\rm df}(\mu) 
  \exp\Big( \frac{ - 9 n  }{32  {\rm df}^{\max}(\mu)  + 8} \Big) 
+ {8\mu }   \tr \big[  (\Sigma_q+\alpha \lambda \idm  )^{-1}\Sigma_q  (   {\Sigma}_q+\mu \idm)^{-1} \big],
$$
if    $\ds {\rm df}^{\max}(\mu) = \sup_{x \in \X}\  \langle \varphi(x), ( \Sigma_q + \mu \idm)^{-1} \varphi(x) \rangle  \leqslant n$.
\end{proposition}
\begin{proof}
We follow a more direct strategy as the previous proof, following estimation results for random feature expansions and column sampling~\cite{rudi2015less,bach2017equivalence}.
\BEAS
A
 & = &  \tr [ \Sigma_p ( \Sigma_p + \lambda \idm)^{-1} ] - 
\tr [ \hat{\Pi}\Sigma_p \hat{\Pi}( \hat{\Pi}\Sigma_p\hat{\Pi} + \lambda \idm)^{-1} ]  \\
 & = &  \lambda \tr [   ( \hat{\Pi}\Sigma_p\hat{\Pi} + \lambda \idm)^{-1} - ( \Sigma_p + \lambda \idm)^{-1}  ]  \\
 & = &  \lambda \tr [  ( \Sigma_p + \lambda \idm)^{-1} ( \Sigma_p - \hat{\Pi}\Sigma_p\hat{\Pi} )   ( \hat{\Pi}\Sigma_p\hat{\Pi} + \lambda \idm)^{-1}   ]  \mbox{ using standard expansions of the matrix inverse,} \\
 & = &  \lambda \tr \big[  (\Sigma_p+\lambda \idm  )^{-1} (\Sigma_p -  \hat{\Pi} \Sigma_p  \hat{\Pi}) (\Sigma_p+\lambda \idm  )^{-1} \big] \\
 & & + \lambda \tr \big[  (\Sigma_p+\lambda \idm  )^{-1} \textcolor{red}{ (\Sigma_p -  \hat{\Pi} \Sigma_p  \hat{\Pi}) 
(   \hat{\Pi} \Sigma_p  \hat{\Pi} + \lambda \idm)^{-1}
  (\Sigma_p -  \hat{\Pi} \Sigma_p  \hat{\Pi}) } (\Sigma_p+\lambda \idm  )^{-1} \big] .
  \EEAS
  Using $\hat{\Pi}^2 = \hat{\Pi}$, we can further bound
  \BEAS A
  & \leqslant &  \lambda \tr \big[  (\Sigma_p+\lambda \idm  )^{-1} (\Sigma_p -  \hat{\Pi} \Sigma_p  \hat{\Pi}) (\Sigma_p+\lambda \idm  )^{-1} \big] \\
 & & + \lambda \tr \big[  (\Sigma_p+\lambda \idm  )^{-1}
\textcolor{red}{\big[  (\idm - \Pi) \Sigma_p  (\idm -  \hat{\Pi})  + \lambda^{-1} \Sigma_p ( \idm -  \hat{\Pi}) \Sigma_p \big]}
  (\Sigma_p+\lambda \idm  )^{-1} \big]  \mbox{ using  } \hat{\Pi}^2 = \hat{\Pi},
\\
 & = &   \tr \big[  (\Sigma_p+\lambda \idm  )^{-2}
 (
 \lambda \Sigma_p - \lambda  \hat{\Pi} \Sigma_p  \hat{\Pi}
 +  \lambda \Sigma_p  - \lambda \Sigma_p   \hat{\Pi} - \lambda   \hat{\Pi} \Sigma_p + \lambda  \hat{\Pi} \Sigma_p  \hat{\Pi}
 + \Sigma_p ( \idm -  \hat{\Pi}) \Sigma_p
 )
 \big] \\
 & = &   2 \lambda \tr \big[  (\Sigma_p+\lambda \idm  )^{-2} \Sigma_p  (\idm -  \hat{\Pi}) \big]
 +  \tr \big[  (\Sigma_p+\lambda \idm  )^{-2}\Sigma_p^2 (\idm -  \hat{\Pi})  \big]
 \\
  & \leqslant &  2   \tr \big[  (\Sigma_p+\lambda \idm  )^{-1}\Sigma_p  (\idm -  \hat{\Pi})  \big] .
 \EEAS
 We can thus use, with $\hat{\Sigma}_q = \frac{1}{n} \sum_{i=1}^n \varphi(x_i)\varphi(x_i)^\ast$,
 $$\idm -  \hat{\Pi} 
 = \lim_{\mu \to 0} \idm - \hat{\Sigma}_q(\hat{\Sigma}_q+\mu\idm)^{-1}
 =   \lim_{\mu \to 0} \mu(\hat{\Sigma}_q+\mu\idm)^{-1} \preccurlyeq \mu (\hat{\Sigma}_q+\mu \idm)^{-1} \preccurlyeq \frac{1}{1-u} \mu ( {\Sigma}_q+\mu \idm)^{-1},$$
 if we have $ ( {\Sigma}_q+\mu \idm)^{-1/2 } ( \Sigma_q - \hat{\Sigma}_q)  ( {\Sigma}_q+\mu \idm)^{-1/2 }  \preccurlyeq u \idm$. 
This leads to, using Lemma~\ref{lemma:conc},
$$
A \leqslant 
\frac{2\mu }{1-u}   \tr \big[  (\Sigma_p+\lambda \idm  )^{-1}\Sigma_p  (   {\Sigma}_q+\mu \idm)^{-1} \big],
$$
with probability greater than
$\ds
1 -   {\rm df}(\mu) \Big(1 + \frac{3(  {\rm df}^{\max}(\mu) + u/3)^2}{u^4 n^2}  \Big)
  \exp\Big( \frac{ - n u^2 }{2  {\rm df}^{\max}(\mu) + 2u/3} \Big).
$
Therefore, we get, with $u = 3/4$:
\BEAS
0 \leqslant \E A  & \leqslant  & 
\tr ( \Sigma_q  + \alpha\lambda \idm)^{-1} \Sigma_q \cdot {\rm df}(\mu) \Big(1 + \frac{3(  {\rm df}^{\max}(\mu) + u/3)^2}{u^4 n^2}  \Big)
  \exp\Big( \frac{ - n u^2 }{2  {\rm df}^{\max}(\mu) + 2u/3} \Big) \\
  & & \hspace*{2cm}
  + \frac{2\mu }{1-u}   \tr \big[  (\Sigma_q+\alpha \lambda \idm  )^{-1}\Sigma_q  (   {\Sigma}_q+\mu \idm)^{-1} \big]
  \\
  & \leqslant  & 
 {\rm df}(\alpha \lambda){\rm df}(\mu) \Big(1 + \frac{3(  {\rm df}^{\max}(\mu) + u/3)^2}{u^4 n^2}  \Big)
  \exp\Big( \frac{ - n u^2 }{2  {\rm df}^{\max}(\mu) + 2u/3} \Big) \\
  & &  \hspace*{2cm}
  + \frac{2\mu }{1-u}   \tr \big[  (\Sigma_q+\alpha \lambda \idm  )^{-1}\Sigma_q  (  {\Sigma}_q+\mu \idm)^{-1} \big]\\
   & \leqslant  & 
 {\rm df}(\alpha \lambda){\rm df}(\mu) \Big(1 + \frac{3(  {\rm df}^{\max}(\mu) + 1/4)^2}{(3/4)^4 n^2}  \Big)
  \exp\Big( \frac{ - n (3/4)^2 }{2  {\rm df}^{\max}(\mu) + 1/2} \Big) \\
  & &  \hspace*{2cm}
  +  {8\mu }   \tr \big[  (\Sigma_q+\alpha \lambda \idm  )^{-1}\Sigma_q  (   {\Sigma}_q+\mu \idm)^{-1} \big] \mbox{ with } u = 3/4.\EEAS

Using the same reasoning as in the proof of Prop.~\ref{prop:df}, the probabilistic term is upper-bounded by \\
$\ds
16  {\rm df}(\alpha \lambda){\rm df}(\mu) 
  \exp\Big( \frac{ - 9 n  }{32  {\rm df}^{\max}(\mu)  + 8} \Big) 
  $ if $  {\rm df}^{\max}(\mu)  \leqslant n$.

\end{proof}

\subsection{Proof of Proposition~\ref{prop:proof_est}}
\label{app:proof_est}

\begin{proof} We assume that $p(x) \geqslant q(x) / \alpha$ for all $x \in \X$.
Given the integral representation of the entropy, we have $$
   \tr \big[ \hat{\Sigma}_p \log \hat{\Sigma}_p \big]  -    \tr \big[  {\Sigma}_p \log  {\Sigma}_p \big]
 = \int_0^{+\infty}\big(  \tr \big[ \Sigma_p ( \Sigma_p +\alpha \lambda \idm)^{-1} \big] -  \tr \big[  \hat{\Sigma}_p  (  \hat{\Sigma}_p  + \alpha\lambda \idm)^{-1} \big]\big) \alpha d\lambda
 ,$$ 
 it turns out it is possible to truncate large values of $\lambda$. Indeed, 
 \BEAS
&& \int_{\lambda_1}^{+\infty}\big(  \tr \big[ \Sigma_p ( \Sigma_p + \alpha \lambda \idm)^{-1} \big] -  \tr \big[  \hat{\Sigma}_p  (  \hat{\Sigma}_p  +\alpha  \lambda \idm)^{-1} \big]\big) \alpha d\lambda
 \\
 & = &  \tr  \hat\Sigma_p \log ( \hat\Sigma_p + \alpha \lambda_1 \idm)-   \tr  \Sigma_p \log ( \Sigma_p + \alpha\lambda_1 \idm)
\\
 & \leqslant &  \tr  \hat\Sigma_p \log ( \idm + \alpha  \lambda_1 \idm)-   \tr  \Sigma_p \log ( \alpha \lambda_1 \idm)
 = \log(1+\alpha \lambda_1) - \log(\alpha \lambda_1) \leqslant \frac{1}{\alpha \lambda_1},
 \EEAS
 with a similar bound for the opposite.

 Thus, we  have the bound, with $\lambda_1 \geqslant  \lambda_0 $:
$$
\E  
A =  \big|   \tr \big[ \hat{\Sigma}_p \log \hat{\Sigma}_p \big]  -    \tr \big[  {\Sigma}_p \log  {\Sigma}_p \big]\big|
   \leqslant   \frac{1}{\alpha \lambda_1} +  \int_0^{\lambda_0} \E \big|
{\rm df}(\alpha\lambda ) 
 -  \widehat{\rm df}(\alpha\lambda ) \big| \alpha d\lambda+  \int_{\lambda_0}^{\lambda_1} \E \big|
{\rm df}(\alpha\lambda ) 
 -  \widehat{\rm df}(\alpha\lambda ) \big|\alpha d\lambda,
 $$
 where $ {\rm df}(\alpha\lambda ) = \tr \big[ \Sigma_p ( \Sigma_p + \alpha\lambda \idm)^{-1} \big]$
 and $\widehat{\rm df}(\alpha\lambda ) = \tr \big[ \hat{\Sigma}_p  ( \hat{\Sigma}_p  + \alpha\lambda \idm)^{-1} \big]$ are the usual degrees of freedom for $\Sigma_p$.
We have
\BEAS
{\rm df}(\alpha\lambda)  & \leqslant &  {\rm df}^{\max} (\alpha\lambda )  \leqslant \sup_{x \in \X}\  \langle \varphi(x),    ( \alpha\Sigma 
+ \alpha \lambda \idm)^{-1}  \varphi(x) \rangle = \frac{1}{\alpha} C(\lambda),
\EEAS
where $\ds C(\lambda) = \tr [ \Sigma ( \Sigma + \lambda \idm)^{-1} ]$. Thus, we get:
\BEAS
A
 & \leqslant & \frac{1}{\alpha \lambda_1}
 + n \alpha \lambda_0 +  \int_0^{\lambda_0}  C(\lambda) d\lambda
  + 
 \int_{\lambda_0}^{\lambda_1} \E \big|
{\rm df}(\alpha\lambda ) 
 -  \widehat{\rm df}(\alpha\lambda ) \big| \alpha d\lambda
 .\EEAS
For $\lambda \geqslant \lambda_0$, we have $ {\rm df}(\lambda) \leqslant  {\rm df}^{\max} (\lambda ) = d  \leqslant 
 \frac{1}{\alpha} C(\lambda_0)= d_0$, which we assume to be less than $n$.

Thus
$$
16 d n  
  \exp\Big( \frac{ - 9 n  }{32 d + 8} \Big)  \leqslant 16 d n  
  \exp\Big( \frac{ - 9 n  }{32 d_0 + 8} \Big),
$$
and we get the bound, using $C(\lambda) \leqslant \frac{1}{\lambda}$:
\BEAS
A
 & \leqslant &  \frac{1}{\alpha \lambda_1}
 + n \alpha \lambda_0 +  \int_0^{\lambda_0}  C(\lambda) d\lambda
+ \int_{\lambda_0}^{\lambda_1} 
 \big(
  \frac{d}{\sqrt{n}} + 4
  \frac{d^2}{n}
  + 16 d n  
  \exp\Big( \frac{ - 9 n  }{32 d_0+ 8} \Big) 
  \big) \alpha d\lambda
\\
& \leqslant &  \frac{1}{\alpha \lambda_1}
 + n \alpha \lambda_0 +  \int_0^{\lambda_0}  C(\lambda) d\lambda
  + \big( \frac{1}{\sqrt{n}} + 16  n  
  \exp\Big( \frac{ - 9 n  }{32 d_0+ 8} \Big)  \Big)
 \int_{\lambda_0}^{\lambda_1} 
 C(\lambda) d\lambda  + \frac{4}{n \alpha }  \int_{\lambda_0}^{\lambda_1} 
 C(\lambda)^2 d\lambda
.
 \EEAS

We take $\lambda_0$ such that  $d_0 = \frac{n}{\log n}$, with a bound:
\BEAS
A
& \leqslant &  \frac{1}{\alpha \lambda_1}
 + n \alpha \lambda_0 +  \int_0^{\lambda_0}  C(\lambda) d\lambda
  +   \frac{17}{\sqrt{n}} 
 \int_{\lambda_0}^{\lambda_1} 
 C(\lambda) d\lambda  + \frac{4}{n \alpha }  \int_{\lambda_0}^{\lambda_1} 
 C(\lambda)^2 d\lambda
\\
& \leqslant &  \frac{1}{\alpha \lambda_1}
 + n \alpha \lambda_0 +  \int_0^{\lambda_0}  C(\lambda) d\lambda
  +   \frac{17}{\sqrt{n}} 
 \int_{\lambda_0}^{\lambda_1} 
 C(\lambda) d\lambda  + \frac{4c }{n \alpha } .
\EEAS
The condition on $C$, implies that $C(\lambda)^2 \frac{\lambda}{2} \leqslant \int_{\lambda/2}^\lambda C(\lambda')^2 d\lambda' \leqslant c$, leading to $C(\lambda) \leqslant \sqrt{\frac{ 2c}{\lambda}}$, leading to $\frac{n}{\log n}\alpha \leqslant \sqrt{\frac{ 2c}{\lambda_0}}$, and thus $\lambda_0 \leqslant  {2c} \frac{ ( \log n)^2}{n^2 \alpha^2} $.
Then
\BEAS
A
& \leqslant &  \frac{1}{\alpha \lambda_1}
 +  {2c} \frac{ ( \log n)^2}{n  \alpha} +  {2c} \frac{ \log n }{n  \alpha}    +   \frac{17}{\sqrt{n}} 
 \int_{\lambda_0}^{\lambda_1} 
 C(\lambda) d\lambda  + \frac{4c }{n \alpha } .
\EEAS
We take $\lambda_1 = \lambda_0 + n$
so that
 $\ds\int_{\lambda_0}^{\lambda_1} 
 C(\lambda) d\lambda 
 \leqslant  \int_{0}^{\lambda_0+1} 
 C(\lambda) d\lambda + \log \frac{\lambda_0+n}{\lambda_0+1}
 \leqslant 2 \sqrt{c} \sqrt{1+\lambda_0} + \log n
 $.
Then, overall, we get:
\BEAS
A
& \leqslant &  \frac{1 + 8c ( \log n)^2}{n \alpha }
 +   \frac{17}{\sqrt{n}} \big(
 2 \sqrt{c} +2\sqrt{2} c\frac{  \log n }{n \alpha}  + \log n
 \big) \leqslant    \frac{1 +  c ( 8 \log n)^2}{n \alpha }
 +   \frac{17}{\sqrt{n}} \big(
 2 \sqrt{c}   + \log n
 \big),
 \EEAS
which is the desired bound.
\end{proof}

\subsection{Proof of Proposition~\ref{prop:proof_est_other}}
\label{app:proof_est_other}

\begin{proof}
We can combine with the proof of Proposition~\ref{prop:proof_est}, to get:
$$ A =  \tr \big[  \Pi {\Sigma}_p \Pi \log \Pi {\Sigma}_p \Pi \big] -  \tr \big[  {\Sigma}_p \log  {\Sigma}_p \big]  
=
\int_0^{+\infty}\big(  \tr \big[ \Sigma_p ( \Sigma_p +  \lambda \idm)^{-1} \big] -  \tr \big[  \Pi {\Sigma}_p  \Pi(  \Pi{\Sigma}_p  \Pi+  \lambda \idm)^{-1} \big]\big)   d\lambda
 ,$$ 
which is always non-negative  and
such that, for $\mu = \alpha \varepsilon \lambda$, if $ {\rm df}^{\max}( \alpha \varepsilon \lambda_0)  \leqslant n$:
\BEAS
A & \leqslant & 
\int_0^{\lambda_0}  \tr \big[ \Sigma_p ( \Sigma_p +  \lambda \idm)^{-1} \big]  d\lambda
+\int_{\lambda_0}^{\lambda_1}\big(  \tr \big[ \Sigma_p ( \Sigma_p +  \lambda \idm)^{-1} \big] -  \tr \big[  \Pi {\Sigma}_p  \Pi(  \Pi{\Sigma}_p  \Pi+  \lambda \idm)^{-1} \big]\big)   d\lambda
+ \frac{1}{\lambda_1}
\\
& \leqslant & 
\int_0^{\lambda_0} {\rm df}(\alpha \lambda) d\lambda
 \\
 &&\hspace*{-1cm}+\int_{\lambda_0}^{\lambda_1}
\Big[
16  {\rm df}(\alpha \lambda){\rm df}( \alpha \varepsilon \lambda) 
  \exp\Big( \frac{ - 9 n  }{32  {\rm df}^{\max}( \alpha \varepsilon \lambda_0)  + 8} \Big) 
+ {8\alpha \varepsilon \lambda }   \tr \big[  (\Sigma_q+\alpha \lambda \idm  )^{-1}\Sigma_q  (   {\Sigma}_q+\alpha \varepsilon \lambda \idm)^{-1} \big]
\Big]
 d\lambda
+ \frac{1}{\lambda_1} \\
& \leqslant & 
\frac{1}{\alpha} \int_0^{\alpha \lambda_0} {\rm df}(  \lambda) d\lambda
+
 \frac{16 }{\alpha}  \frac{n}{\beta_n}
  \exp\Big( \frac{ - 9    }{32  \beta_n^{-1} + 8/n} \Big) 
  \int_{\alpha\lambda_0}^{\alpha\lambda_1}
 {\rm df}( \lambda)d\lambda
  + \frac{1}{\lambda_1}
 \\
 &&+\int_{\alpha\lambda_0}^{\alpha\lambda_1}
  {8  \varepsilon \lambda }   \tr \big[  (\Sigma_q+  \lambda \idm  )^{-1}\Sigma_q  (   {\Sigma}_q+  \varepsilon \lambda \idm)^{-1}  
\Big]
 d\lambda
\\
& \leqslant & 
\frac{1}{\alpha} \int_0^{\alpha \lambda_0} {\rm df}(  \lambda) d\lambda
+
 \frac{16 }{\alpha}  \frac{n}{\beta_n}
  \exp\Big( \frac{ - 9    }{32  \beta_n^{-1} + 8/n} \Big) 
  \int_{0}^{\alpha\lambda_1}
 {\rm df}( \lambda)d\lambda
  + \frac{1}{\lambda_1}
 \\
 &&+\int_{0}^{\alpha\lambda_1}
  {8  \varepsilon \lambda }   \tr \big[  (\Sigma_q+  \lambda \idm  )^{-1}\Sigma_q  (   {\Sigma}_q+  \varepsilon \lambda \idm)^{-1}  
\Big]
 d\lambda, \mbox{ extending integration limits},
\EEAS
with  ${\rm df}^{\max}( \alpha \varepsilon \lambda_0) = \frac{n}{\beta_n}$.

We also have
\BEAS
\int_0^{\lambda'} {\rm df}(\lambda) d\lambda
&=&  \sum_{\sigma \in \Lambda(\Sigma_q)}
\int_0^{\lambda'}  \frac{\sigma}{\sigma+\lambda} d\lambda
= \sum_{\sigma \in \Lambda(\Sigma_q)}
   {\sigma}   \log{(1 + \lambda' \sigma^{-1} ) }  ,
\EEAS
and
\BEAS
\int_0^{\lambda'} {\varepsilon \lambda }   \tr \big[  (\Sigma_q+  \lambda \idm  )^{-1}\Sigma_q  (   {\Sigma}_q+  \varepsilon \lambda \idm)^{-1}  \big] d\lambda
&=&  \sum_{\sigma \in \Lambda(\Sigma_q)}
\int_0^{\lambda'}  \frac{\sigma}{\sigma+\lambda}  \frac{\varepsilon \lambda}{\sigma+\varepsilon\lambda} d\lambda
\\
& = & 
\sum_{\sigma \in \Lambda(\Sigma_q)}
 \frac{\varepsilon  \sigma}{1-\varepsilon} 
\int_0^{\lambda'}  \Big[ \frac{1}{\sigma+\varepsilon\lambda} - \frac{1}{\sigma+\lambda} \Big] d\lambda
\\
& = & 
\sum_{\sigma \in \Lambda(\Sigma_q)}
 \frac{\varepsilon  \sigma}{1-\varepsilon} 
\Big[    \frac{1}{\varepsilon} \log\frac{\sigma+\varepsilon\lambda'}{\sigma}  - \log\frac{\sigma+\lambda'}{\sigma}\Big] 
\\
& \leqslant & 
\sum_{\sigma \in \Lambda(\Sigma_q)}
 \frac{  \sigma}{1-\varepsilon} 
 \log{(1 + \lambda' \varepsilon \sigma^{-1} ) }   = \frac{   1}{1-\varepsilon} 
 \int_0^{\varepsilon \lambda'} {\rm df}(\lambda) d\lambda.
\EEAS
Thus, with  $\varepsilon = \frac{1}{\alpha\lambda_0} {\rm df}^{-1} ( n / \beta_n)$,
\BEAS
A & \leqslant & 
\frac{1}{\alpha} \int_0^{\alpha \lambda_0} {\rm df}(  \lambda) d\lambda
+ 
 \frac{16 }{\alpha}  \frac{n}{\beta_n}
  \exp\Big( \frac{ - 9    }{32  \beta_n^{-1} + 8/n} \Big)  
  \int_{0}^{\alpha\lambda_1}
 {\rm df}( \lambda)d\lambda
 + \frac{8}{1-\varepsilon}  \int_{0}^{\varepsilon\alpha\lambda_1}
  {\rm df}( \lambda)d\lambda
  + \frac{1}{\lambda_1}.
\EEAS

\paragraph{Polynomial decays.}
If ${\rm df}(\lambda) \leqslant B \lambda^{-1/s}$, then we have:
$ \ds
\frac{n}{\beta_n} \leqslant B ( \alpha \lambda_o \varepsilon)^{-1/s},
$
and thus
$\varepsilon \leqslant \frac{1}{\alpha\lambda_0} \Big( \frac{B \beta_n}{n} \Big)^{s}$.
Moreover, 
$\ds \int_0^{\alpha \lambda_0} {\rm df}(  \lambda) d\lambda \leqslant \frac{B}{1-1/s} (\alpha \lambda_0)^{1-1/s}$.
Thus, with $\varepsilon < 1/2$,
\BEAS
A & \leqslant & 
\frac{B\alpha^{-1/s}}{1-1/s} ( \lambda_0)^{1-1/s}
+ 
 \frac{16 }{\alpha}  \frac{n}{\beta_n}
  \exp\Big( \frac{ - 9    }{32  \beta_n^{-1} + 8/n} \Big)  
\frac{B}{1-1/s} (\alpha \lambda_1)^{1-1/s}
 + \frac{8}{1-\varepsilon} \frac{B}{1-1/s} (\varepsilon \alpha \lambda_1)^{1-1/s}
  + \frac{1}{\lambda_1}.
  \EEAS
  Now turning to asymptotic expansions, we get
  \BEAS
  A
  & = & 
O \Big(B   \lambda_0^{1-1/s}
+ 
  \frac{n}{\beta_n}
  \exp(- \beta_n/4)  
B \lambda_1^{1-1/s}
 +  B ( \lambda_1 /\lambda_0)^{1-1/s}
 \Big( \frac{B \beta_n}{n} \Big)^{s-1}
  + \frac{1}{\lambda_1} \Big).
\EEAS
We can take $\beta_n \propto \log n$ so that the term $ \exp(- \beta_n/4)  $ is negligible. 

We can choose $\lambda_0$ such that $B   \lambda_0^{1-1/s} \sim B ( \lambda_1 /\lambda_0)^{1-1/s}
 \Big( \frac{B \log n }{n} \Big)^{s-1} $, that is, $\lambda_0 \propto B^{s/2} \lambda_1^{1/2}  \Big( \frac{ \log n}{n} \Big)^{s/2}$, and then
$\lambda_1^{(3s-1)/(2s)} B^{(s+1)/2}  \Big( \frac{  \log n}{n} \Big)^{(s-1)/2}$, finally leading to
$\ds
A = O\Big(  B^{s(s+1)/(3s-1)} \Big( \frac{  \log n}{n} \Big)^{s(s-1)/(3s-1)} \Big).
$

\paragraph{Exponential decays.}
If ${\rm df}(\lambda) \leqslant B \log( C / \lambda) ^d $, then we have:
$\ds
\frac{n}{\beta_n} \
\leqslant   B \log( C / (\varepsilon \alpha \lambda_0)) ^d ,
$
and thus $\ds \varepsilon \leqslant \frac{C}{\lambda_0 \alpha} \exp( - ( B n / \beta_n)^{1/d} )$.
Moreover
$\ds \int_0^{\alpha \lambda_0} {\rm df}(  \lambda) d\lambda \sim \alpha\lambda_0 B 
\log( C / \lambda_0 \alpha) ^d$.
Thus, asymptotically,
\BEAS
A & =  & O \Big(
\lambda_0 B 
\log( C / \lambda_0  ) ^d
+ 
    \frac{n}{\beta_n}
  \exp (- \beta_n / 4  )  
 \lambda_1 B 
 \big( 
\log( C / \lambda_1   ) \big) ^d
 +  \lambda_1 \varepsilon   B 
\log( C /  \varepsilon\lambda_1  ) ^d
  + \frac{1}{\lambda_1} \Big).
\EEAS
With $\lambda_0 \propto  C^{1/2} \lambda_1^{1/2} \exp( - \frac{1}{2} (Bn/\beta_n)^{1/d})$, 
the first and third terms lead to
$$
B C^{1/2} \lambda_1^{1/2}  \exp( - \frac{1}{2} (Bn/\beta_n)^{1/d})
 \log( C / \lambda_0  ) ^d  
 = B C^{1/2} \lambda_1^{1/2}  \exp( - \frac{1}{2} (Bn/\beta_n)^{1/d})
 \Big( \log C + \log \frac{1}{\lambda_1} +   (Bn/\beta_n)^{1/d} ) ^d .  $$
With 
$\lambda_1\propto  \frac{1}{B} \exp(\beta_n/8)$, 
the second term is of order
$\ds
    \frac{n}{\beta_n}
  \exp (- \beta_n / 8  )  
 \big( 
\log C + \log \frac{1}{\lambda_1}   ) \big) ^d
$.
We can then choose $\beta_n  \propto (Bn)^{1/(d+1)}$, to get  the overall bound
$ \ds
A =  \big( \log ( B C ) \big)^d ( \sqrt{BC} + B ) \exp\big( - \frac{1}{4} (Bn)^{1/(d+1)}\big) .
$
\end{proof}

\subsection{Computation of degrees of freedom}
\label{app:df}

We have, when $\sigma$ tends to zero, with $\ds \nu = \frac{\lambda}{\hat{k}(0)}=\frac{\lambda}{ \tanh^d \frac{\sigma}{2} }$,
\BEAS
\sum_{\omega \in \Z^d} \frac{1}{1+ \nu\exp(  \sigma \| \omega\|_1)}
& = & \sum_{k \geqslant 0}
\frac{1}{1+ \nu\exp(  \sigma k)}
\big| \{ \omega \in \Z^d, \ \| \omega\|_1 = k\} \big|
\leqslant \sum_{k \geqslant 0}
\frac{1}{1+ \nu\exp(  \sigma k)}
2^d k^{d-1}
\\
& \leqslant & \sum_{k \geqslant \frac{1}{\sigma} \log \frac{1}{\nu}}
\frac{1}{ \nu\exp(  \sigma k)}
2^d k^{d-1}
+ \sum_{k \leqslant \frac{1}{\sigma} \log \frac{1}{\nu}}
2^d k^{d-1}
\\
& \leqslant &  
2^d  \frac{ (d-1)!   }{ (1-e^{-\sigma})^{-d}} \Big[
1 + \frac{1}{2} \Big( \log \frac{1}{\nu} \Big)^{d-1} 
\big]
+  
2^d \Big( \frac{1}{\sigma} \log \frac{1}{\nu} \Big)^{d} \\
& \leqslant & 
2^d  \frac{ d!  }{ (1-e^{-\sigma})^{-d}} \Big[
1 +   \Big( \log \frac{1}{\nu} \Big)^{d} 
\big]
\\
& \leqslant & 
e^{-\sigma d / 2}  \frac{ d!  }{ \sinh^d(\sigma/2)} \Big[
1 +   \Big( \log \frac{1}{\lambda} + d \log \tanh  \frac{\sigma}{2} \Big)^{d} 
\big]
,
\EEAS
where we have used that, for $\nu \geqslant e$,
$$
( 1 - e^{-\sigma}) 
 \sum_{k \geqslant \frac{1}{\sigma} \log \frac{1}{\nu}}
 \exp(  - \sigma k )
k^{\alpha} \leqslant \big( \frac{1}{\sigma} \log \frac{1}{\nu} \big) ^\alpha  
\nu + \alpha  \sum_{k \geqslant \frac{1}{\sigma} \log \frac{1}{\nu}}
 \exp(  - \sigma k )
k^{\alpha-1} 
$$
and
$ \ds
  \sum_{k \geqslant \frac{1}{\sigma} \log \frac{1}{\nu}}
 \exp(  - \sigma k ) \leqslant \frac{\nu}{1-e^{-\sigma}},
$
which leads to, with $u_\alpha =  \sum_{k \geqslant \frac{1}{\sigma} \log \frac{1}{\nu}}
 \exp(  - \sigma k )
k^{\alpha} $:
$$\frac{1}{\alpha!} (1-e^{-\sigma})^{\alpha+1}u_\alpha \leqslant
\frac{1}{\alpha!} (1-e^{-\sigma})^{\alpha-1} \nu \Big( \log \frac{1}{\nu} \Big)^\alpha 
+ \frac{1}{(\alpha-1)!} (1-e^{-\sigma})^{\alpha}u_{\alpha-1},
$$ 
and thus
$u_\alpha \leqslant \frac{ \alpha! \nu }{ (1-e^{-\sigma})^{-\alpha-1}} \Big[
1 + \frac{1}{2} \Big( \log \frac{1}{\nu} \Big)^\alpha 
\big]$.

  In order to bound the constant $c$ in Prop.~\ref{prop:proof_est}, we can split the integral into. We  have for  $\beta =  \tanh^d   \frac{\sigma}{2}> 0$,
\BEAS
c &  \leqslant &  \int_\beta^{\infty} \frac{1}{\lambda^2}d\lambda  + \int_0^\beta
\ds e^{-\sigma d }  \frac{ d!^2  }{ \sinh^{2d}(\sigma/2)} \Big[
2 +   2 \Big( \log \frac{ \tanh^d   \frac{\sigma}{2}}{\lambda}   \Big)^{2d} 
\Big] d\lambda \\
& \leqslant & \frac{1}{\beta} 
+ \beta \int_0^1
\ds e^{-\sigma d }  \frac{ d!^2  }{ \sinh^{2d}(\sigma/2)} \Big[
2 +   2 \Big( \log \frac{1}{\lambda}   \Big)^{2d} 
\Big] d\lambda = O(\sigma^{-d})
\EEAS
when $\sigma$ goes to zero.
 
\bibliography{quantum}

\begin{thebibliography}{10}

\bibitem{ando1979concavity}
Tsuyoshi Ando.
\newblock Concavity of certain maps on positive definite matrices and
  applications to {H}adamard products.
\newblock {\em Linear Algebra and its Applications}, 26:203--241, 1979.

\bibitem{araki2002entropy}
Huzihiro Araki and Elliott~H. Lieb.
\newblock Entropy inequalities.
\newblock In {\em Inequalities}, pages 47--57. Springer, 2002.

\bibitem{bach2013learning}
Francis Bach.
\newblock Learning with submodular functions: A convex optimization
  perspective.
\newblock {\em Foundations and Trends in Machine Learning}, 6(2-3):145--373,
  2013.

\bibitem{bach2017breaking}
Francis Bach.
\newblock Breaking the curse of dimensionality with convex neural networks.
\newblock {\em Journal of Machine Learning Research}, 18(1):629--681, 2017.

\bibitem{bach2017equivalence}
Francis Bach.
\newblock On the equivalence between kernel quadrature rules and random feature
  expansions.
\newblock {\em Journal of Machine Learning Research}, 18(1):714--751, 2017.

\bibitem{bach2002kernel}
Francis Bach and Michael~I. Jordan.
\newblock Kernel independent component analysis.
\newblock {\em Journal of Machine Learning Research}, 3(Jul):1--48, 2002.

\bibitem{bach2005predictive}
Francis Bach and Michael~I. Jordan.
\newblock Predictive low-rank decomposition for kernel methods.
\newblock In {\em Proceedings of the International Conference on Machine
  Learning}, pages 33--40, 2005.

\bibitem{beck2009fast}
Amir Beck and Marc Teboulle.
\newblock A fast iterative shrinkage-thresholding algorithm for linear inverse
  problems.
\newblock {\em SIAM Journal on Imaging Sciences}, 2(1):183--202, 2009.

\bibitem{berger2012differential}
Marcel Berger and Bernard Gostiaux.
\newblock {\em Differential Geometry: Manifolds, Curves, and Surfaces:
  Manifolds, Curves, and Surfaces}, volume 115.
\newblock Springer Science \& Business Media, 2012.

\bibitem{berlinet2011reproducing}
Alain Berlinet and Christine Thomas-Agnan.
\newblock {\em Reproducing Kernel Hilbert Spaces in Probability and
  Statistics}.
\newblock Springer Science \& Business Media, 2011.

\bibitem{berthier2021infinite}
Elo{\"\i}se Berthier, Justin Carpentier, Alessandro Rudi, and Francis Bach.
\newblock Infinite-dimensional sums-of-squares for optimal control.
\newblock Technical Report 2110.07396, arXiv, 2021.

\bibitem{bhatia2009positive}
Rajendra Bhatia.
\newblock {\em Positive Definite Matrices}.
\newblock Princeton University Press, 2009.

\bibitem{binkowski2018demystifying}
Miko{\l}aj Bi{\'n}kowski, Danica~J. Sutherland, Michael Arbel, and Arthur
  Gretton.
\newblock Demystifying {MMD} {GANs}.
\newblock In {\em International Conference on Learning Representations}, 2018.

\bibitem{boutsidis2009improved}
Christos Boutsidis, Michael~W. Mahoney, and Petros Drineas.
\newblock An improved approximation algorithm for the column subset selection
  problem.
\newblock In {\em Proceedings of the Symposium on Discrete algorithms}, pages
  968--977, 2009.

\bibitem{bregman1967relaxation}
Lev~M. Bregman.
\newblock The relaxation method of finding the common point of convex sets and
  its application to the solution of problems in convex programming.
\newblock {\em USSR Computational Mathematics and Mathematical Physics},
  7(3):200--217, 1967.

\bibitem{brezis80analyse}
Ha{\"{i}}m Brezis.
\newblock {\em Analyse Fonctionelle}.
\newblock Masson, Paris, France, 1980.

\bibitem{cardoso2003dependence}
Jean-Fran{\c{c}}ois Cardoso.
\newblock Dependence, correlation and {G}aussianity in independent component
  analysis.
\newblock {\em Journal of Machine Learning Research}, 4:1177--1203, 2003.

\bibitem{cover1999elements}
Thomas~M. Cover and Joy~A. Thomas.
\newblock {\em Elements of Information Theory}.
\newblock John Wiley \& Sons, 1999.

\bibitem{csiszar1967information}
Imre Csisz{\'a}r.
\newblock Information-type measures of difference of probability distributions
  and indirect observation.
\newblock {\em Studia Scientiarum Mathematicarum Hungarica}, 2:229--318, 1967.

\bibitem{csiszar2008axiomatic}
Imre Csisz{\'a}r.
\newblock Axiomatic characterizations of information measures.
\newblock {\em Entropy}, 10(3):261--273, 2008.

\bibitem{de2021regularization}
Ernesto De~Vito, Lorenzo Rosasco, and Alessandro Rudi.
\newblock Regularization: From inverse problems to large-scale machine
  learning.
\newblock In {\em Harmonic and Applied Analysis}, pages 245--296. Springer,
  2021.

\bibitem{dwork2006calibrating}
Cynthia Dwork, Frank McSherry, Kobbi Nissim, and Adam Smith.
\newblock Calibrating noise to sensitivity in private data analysis.
\newblock In {\em Theory of Cryptography Conference}, pages 265--284. Springer,
  2006.

\bibitem{edelman1998geometry}
Alan Edelman, Tom{\'a}s~A. Arias, and Steven~T. Smith.
\newblock The geometry of algorithms with orthogonality constraints.
\newblock {\em SIAM journal on Matrix Analysis and Applications},
  20(2):303--353, 1998.

\bibitem{foias1990commutant}
Ciprian Foias and Arthur~E. Frazho.
\newblock {\em The Commutant Lifting Approach to Interpolation Problems}.
\newblock Springer, 1990.

\bibitem{frank2013monotonicity}
Rupert~L. Frank and Elliott~H. Lieb.
\newblock Monotonicity of a relative {R}{\'e}nyi entropy.
\newblock {\em Journal of Mathematical Physics}, 54(12):122201, 2013.

\bibitem{fukumizu2009kernel}
Kenji Fukumizu, Francis Bach, and Michael~I. Jordan.
\newblock Kernel dimension reduction in regression.
\newblock {\em Annals of Statistics}, 37(4):1871--1905, 2009.

\bibitem{fukumizu2007kernel}
Kenji Fukumizu, Arthur Gretton, Xiaohai Sun, and Bernhard Sch{\"o}lkopf.
\newblock Kernel measures of conditional dependence.
\newblock In {\em Advances in Neural Information Processing Systems},
  volume~20, pages 489--496, 2007.

\bibitem{han2020optimal}
Yanjun Han, Jiantao Jiao, Tsachy Weissman, and Yihong Wu.
\newblock Optimal rates of entropy estimation over {L}ipschitz balls.
\newblock {\em The Annals of Statistics}, 48(6):3228--3250, 2020.

\bibitem{hormander1984analysis}
Lars H{\"o}rmander.
\newblock {\em The Analysis of Linear Partial Differential Operators I:
  Distribution Theory and Fourier Analysis}.
\newblock Springer, 1984.

\bibitem{isham2001lectures}
Chris~J. Isham.
\newblock {\em Lectures on Quantum Theory: Mathematical and Structural
  Foundations}.
\newblock Allied Publishers, 2001.

\bibitem{jordan2003semidefinite}
Michael~I. Jordan and Martin~J. Wainwright.
\newblock Semidefinite relaxations for approximate inference on graphs with
  cycles.
\newblock {\em Advances in Neural Information Processing Systems}, 16, 2003.

\bibitem{kato}
Tosio Kato.
\newblock {\em Perturbation Theory for Linear Operators}.
\newblock Spring{\-}er-Ver{\-}lag, 1966.

\bibitem{lesniewski1999monotone}
Andrew Lesniewski and Mary~Beth Ruskai.
\newblock Monotone {R}iemannian metrics and relative entropy on noncommutative
  probability spaces.
\newblock {\em Journal of Mathematical Physics}, 40(11):5702--5724, 1999.

\bibitem{micchelli2006universal}
Charles~A. Micchelli, Yuesheng Xu, and Haizhang Zhang.
\newblock Universal kernels.
\newblock {\em Journal of Machine Learning Research}, 7(12), 2006.

\bibitem{muandet2017kernel}
Krikamol Muandet, Kenji Fukumizu, Bharath Sriperumbudur, and Bernhard
  Sch{\"o}lkopf.
\newblock Kernel mean embedding of distributions: A review and beyond.
\newblock {\em Foundations and Trend in Machine Learning}, 10(1-2):1--141,
  2017.

\bibitem{muller2013quantum}
Martin M{\"u}ller-Lennert, Fr{\'e}d{\'e}ric Dupuis, Oleg Szehr, Serge Fehr, and
  Marco Tomamichel.
\newblock On quantum {R}{\'e}nyi entropies: A new generalization and some
  properties.
\newblock {\em Journal of Mathematical Physics}, 54(12):122203, 2013.

\bibitem{nesterov2005smooth}
Yurii Nesterov.
\newblock Smooth minimization of non-smooth functions.
\newblock {\em Mathematical Programming}, 103(1):127--152, 2005.

\bibitem{pauwels2018relating}
Edouard Pauwels, Francis Bach, and Jean-Philippe Vert.
\newblock Relating leverage scores and density using regularized {C}hristoffel
  functions.
\newblock In {\em Advances in Neural Information Processing Systems}, pages
  1663--1672, 2018.

\bibitem{petz1986sufficient}
D{\'e}nes Petz.
\newblock Sufficient subalgebras and the relative entropy of states of a von
  {N}eumann algebra.
\newblock {\em Communications in Mathematical Physics}, 105(1):123--131, 1986.

\bibitem{petz2003monotonicity}
D{\'e}nes Petz.
\newblock Monotonicity of quantum relative entropy revisited.
\newblock {\em Reviews in Mathematical Physics}, 15(01):79--91, 2003.

\bibitem{reed1980functional}
Michael Reed and Barry Simon.
\newblock {\em Functional analysis}.
\newblock Methods of Modern Mathematical Physics. Academic Press, 1980.

\bibitem{robert2007bayesian}
Christian~P. Robert.
\newblock {\em The Bayesian choice: from decision-theoretic foundations to
  computational implementation}, volume~2.
\newblock Springer, 2007.

\bibitem{rudi2015less}
Alessandro Rudi, Raffaello Camoriano, and Lorenzo Rosasco.
\newblock Less is more: {N}ystr{\"o}m computational regularization.
\newblock {\em Advances in Neural Information Processing Systems}, 28, 2015.

\bibitem{rudi2020finding}
Alessandro Rudi, Ulysse Marteau-Ferey, and Francis Bach.
\newblock Finding global minima via kernel approximations.
\newblock Technical Report 2012.11978, arXiv, 2020.

\bibitem{rudi2017generalization}
Alessandro Rudi and Lorenzo Rosasco.
\newblock Generalization properties of learning with random features.
\newblock In {\em Advances in Neural Information Processing Systems}, pages
  3215--3225, 2017.

\bibitem{ruskai2002inequalities}
Mary~Beth Ruskai.
\newblock Inequalities for quantum entropy: A review with conditions for
  equality.
\newblock {\em Journal of Mathematical Physics}, 43(9):4358--4375, 2002.

\bibitem{ruskai2007another}
Mary~Beth Ruskai.
\newblock Another short and elementary proof of strong subadditivity of quantum
  entropy.
\newblock {\em Reports on Mathematical Physics}, 60(1):1--12, 2007.

\bibitem{scholkopf-smola-book}
Bernhard Sch{\"o}lkopf and Alex~J. Smola.
\newblock {\em Learning with Kernels}.
\newblock MIT Press, 2001.

\bibitem{shawe2004kernel}
John Shawe-Taylor and Nello Cristianini.
\newblock {\em Kernel Methods for Pattern Analysis}.
\newblock Cambridge University Press, 2004.

\bibitem{smola2001regularization}
Alex~J. Smola, Zoltan~L. Ovari, and Robert~C. Williamson.
\newblock Regularization with dot-product kernels.
\newblock {\em Advances in Neural Information Processing Systems}, pages
  308--314, 2001.

\bibitem{song2012feature}
Le~Song, Alex Smola, Arthur Gretton, Justin Bedo, and Karsten Borgwardt.
\newblock Feature selection via dependence maximization.
\newblock {\em Journal of Machine Learning Research}, 13(5), 2012.

\bibitem{sriperumbudur2008injective}
Bharath~K. Sriperumbudur, Arthur Gretton, Kenji Fukumizu, Gert Lanckriet, and
  Bernhard Sch{\"o}lkopf.
\newblock Injective {H}ilbert space embeddings of probability measures.
\newblock In {\em Annual Conference on Learning Theory (COLT)}, pages 111--122,
  2008.

\bibitem{sriperumbudur2010hilbert}
Bharath~K. Sriperumbudur, Arthur Gretton, Kenji Fukumizu, Bernhard
  Sch{\"o}lkopf, and Gert R.~G. Lanckriet.
\newblock Hilbert space embeddings and metrics on probability measures.
\newblock {\em Journal of Machine Learning Research}, 11:1517--1561, 2010.

\bibitem{steinwart2001influence}
Ingo Steinwart.
\newblock On the influence of the kernel on the consistency of support vector
  machines.
\newblock {\em Journal of Machine Learning Research}, 2(Nov):67--93, 2001.

\bibitem{topsoe2000some}
Flemming Topsoe.
\newblock Some inequalities for information divergence and related measures of
  discrimination.
\newblock {\em IEEE Transactions on Information Theory}, 46(4):1602--1609,
  2000.

\bibitem{tropp2015introduction}
Joel~A. Tropp.
\newblock An introduction to matrix concentration inequalities.
\newblock {\em Foundations and Trends in Machine Learning}, 8(1-2):1--230,
  2015.

\bibitem{vacher2021dimension}
Adrien Vacher, Boris Muzellec, Alessandro Rudi, Francis Bach, and
  Francois-Xavier Vialard.
\newblock A dimension-free computational upper-bound for smooth optimal
  transport estimation.
\newblock In {\em Conference on Learning Theory}, pages 4143--4173, 2021.

\bibitem{wainwright2008graphical}
Martin~J. Wainwright and Michael~I. Jordan.
\newblock {\em Graphical Models, Exponential Families, and Variational
  Inference}.
\newblock Now Publishers Inc., 2008.

\bibitem{wilde2013quantum}
Mark~M. Wilde.
\newblock {\em Quantum Information Theory}.
\newblock Cambridge University Press, 2013.

\bibitem{williams2000using}
Christopher Williams and Matthias Seeger.
\newblock Using the {N}ystr{\"o}m method to speed up kernel machines.
\newblock {\em Advances in Neural Information Processing Systems}, 13, 2000.

\bibitem{blakecolt22}
Blake~E. Woodworth, Francis Bach, and Alessandro Rudi.
\newblock Non-convex optimization with certificates and fast rates through
  kernel sums of squares.
\newblock In {\em Annual Conference on Learning Theory (COLT)}, 2022.

\bibitem{yu2013strong}
Yao-Liang Yu.
\newblock The strong convexity of von {N}eumann’s entropy.
\newblock {\em Unpublished note}, 2013.
\newblock \url{http://www.cs.cmu.edu/~yaoliang/mynotes/sc.pdf}.

\end{thebibliography}

  \end{document}